\documentclass[
 amsmath,amssymb,
 aps,
onecolumn, 
nofootinbib]{revtex4-2}

\usepackage{amsmath,amsthm, amssymb}
\usepackage[margin=3cm]{geometry}
\usepackage{mathtools}
\usepackage{dsfont}
\usepackage{xcolor}
\usepackage{algorithm,algpseudocode}
\usepackage{mathrsfs}
\usepackage{thm-restate}
\usepackage{hyperref}

\usepackage{xcolor}
\hypersetup{
    colorlinks,
    linkcolor={black},
    citecolor={blue!50!black},
    urlcolor={blue!80!black}
}

\usepackage{caption}
\usepackage{subcaption}

\newtheorem{theorem}{Theorem}
\newtheorem{lemma}[theorem]{Lemma}
\newtheorem{corollary}[theorem]{Corollary}

\newtheorem{definition}[theorem]{Definition}

\newcommand{\on}{\rm on}
\newcommand{\off}{\rm off}
\newcommand{\haar}{\text{Haar}}

\newcommand{\ket}[1]{|#1\rangle}
\newcommand{\bra}[1]{\langle #1|}
\newcommand{\braket}[2]{\langle #1|#2\rangle}
\newcommand{\ketbra}[2]{| #1\rangle\! \langle #2|}
\newcommand{\parens}[1]{\left( #1 \right)}
\newcommand{\brackets}[1]{\left[ #1 \right]}
\newcommand{\abs}[1]{\left| #1 \right|}
\newcommand{\norm}[1]{\left|\left| #1 \right|\right|}

\newcommand{\set}[1]{\left\{ #1 \right\}}

\newcommand{\EE}{\mathbb{E}}
\newcommand{\TT}{\mathcal{T}}

\newcommand{\prob}[1]{{\rm Prob}\left[ #1 \right]}
\newcommand{\bigo}[1]{O\left(#1\right)}
\newcommand{\bigotilde}[1]{\widetilde{O} \left( #1 \right)}
\newcommand{\ts}{\textsuperscript}

\DeclareMathOperator{\Tr}{Tr}
\newcommand{\trace}[1]{\Tr \brackets{ #1 }}
\newcommand{\partrace}[2]{\Tr_{#1} \brackets{ #2 }}

\newcommand{\hilb}{\mathcal{H}}
\newcommand{\partfun}{\mathcal{Z}}
\newcommand{\identity}{\mathds{1}}

\DeclareMathOperator{\sinc}{sinc}

\begin{document}

\title{ {The Thermodynamic Cost of Ignorance:} \\
Thermal State Preparation with One Ancilla Qubit}

\author{Matthew Hagan}
\affiliation{
University of Toronto, Department of Physics, Toronto ON, Canada
}
\email{matt.hagan@mail.utoronto.ca}

\author{Nathan Wiebe}
\affiliation{
University of Toronto, Department of Computer Science, Toronto ON, Canada
}%
\affiliation{
Pacific Northwest National Laboratory, Richland WA, USA
}%
\affiliation{
Canadian Institute for Advanced Research, Toronto ON, Canada
}%

\begin{abstract}
In this work we investigate a model of thermalization wherein a single ancillary qubit randomly interacts with the system to be thermalized. This not only sheds light on the emergence of Gibbs states in nature, but also provides a routine for preparing arbitrary thermal states on a digital quantum computer. For desired $\beta$ and random interaction $G$, which has Haar random eigenvalues and I.I.D Gaussian eigenvalues, the routine boils down to time independent Hamiltonian simulation and is represented by the channel $\Phi : \rho \mapsto \EE_G \partrace{{\rm Env}}{e^{-i(H + \alpha G)t} \left(\rho \otimes \frac{e^{-\beta H_E}}{\partfun}\right) e^{i (H + \alpha G)t}}$. We rigorously prove that these dynamics reduce to a Markov chain process in the weak-coupling regime with the thermal state as the approximate fixed point. This allows us to upper bound the total simulation time required as a function of the Markov chain spectral gap $\lambda_\star$, which we are able to compute exactly in the ground state limit. These results are independent of any eigenvalue knowledge of the system, but we are further able to show that if one has knowledge of eigenvalue differences $\lambda_S(i) - \lambda_S(j)$, then the total simulation time is dramatically reduced. The ratio of the complete ignorance simulation cost to the perfect knowledge simulation cost scales as $\bigotilde{\frac{\norm{H_S}^7}{\delta_{\min}^7 \epsilon^{3.5} \lambda_\star(\beta)^{3.5}}}$, where $\delta_{\min}$ is related to the eigenvalue differences of the system. Additionally, we provide more specific results for single qubit and harmonic oscillator systems as well as numeric experiments with hydrogen chains. In addition to the algorithmic merits, these results can be viewed as broad extensions of the Repeated Interactions model in open quantum systems to generic Hamiltonians with unknown interactions, giving a complete picture of the thermalization process for quantum systems.
\end{abstract}

\maketitle
\newpage
\tableofcontents

\section{Introduction}

The simulation of quantum systems and materials is among the most promising applications for exponential advantages of digital quantum computers over classical computers \cite{aspuru2005simulated} \cite{reiher2017elucidating} \cite{tensorHypercontraction}. A critical step in quantum simulation algorithms, as well as other quantum algorithms such as Semi-Definite Program (SDP) solvers \cite{brandao2019sdp} and Hamiltonian learning routines \cite{anshu_sample-efficient_2021}, is the preparation of good input states, which are typically thermal states $\frac{e^{-\beta H}}{\trace{e^{-\beta H}}}$. Thermal states at low temperatures (high $\beta$) have large overlap with the ground states of the system, indicating that preparing thermal states is just as difficult as the QMA-Hard $k$-local ground state preparation problem \cite{kempe2005complexitylocalhamiltonianproblem}. 

Many classical algorithms have been developed to estimate the measurement outcomes of quantum experiments with the workhorse behind many of these algorithms typically being some kind of Metropolis-Hastings algorithm \cite{metropolis1953equation} to implement a Markov Chain Monte Carlo (MCMC) program. The Metropolis-Hastings algorithm solves the problem of sampling from arbitrary probability distributions and can be used to estimate partition functions, a \#P-Hard problem \cite{roth1996hardness}. Despite the difficulty of the problem it solves and minimal theoretic guarantees on the runtime, the Metropolis-Hastings algorithm has worked resoundingly well in practice. This is in part due to its elegant simplicity and ease of implementation. The algorithm does tend to breakdown in a few important areas though, namely quantum systems with ``sign problems" (\cite{signProblemOG}, \cite{troyer2005sign}), very high dimension systems \cite{beskos2010optimaltuninghybridmontecarlo}, and Hamiltonians with many deep local minima \cite{betancourt2018conceptualintroductionhamiltonianmonte}. The sign problem in particular serves as an important impetus for developing quantum computers, which naturally do not have to deal with it. However attempts to naively port the classical Metropolis-Hastings algorithm to quantum computers has been rather difficult due to inherent difficulties with quantum information, such as no-cloning. Initial attempts \cite{temme2011} are rather cumbersome and attempts to deal with the filtering and rejection stages relying on ``quantum unwinding" techniques developed by Marriott and Watrous \cite{marriott2005quantum}. These complications make the resulting quantum algorithms difficult to analyze. 

In recent years new approaches have been developed \cite{chen2023quantumthermalstatepreparation}, \cite{gilyen2024quantumgeneralizationsglaubermetropolis}, \cite{motlagh2024ground}, \cite{motta2019} \cite{ding2024single}, many of which are based on the simulation of Linblad operators from open quantum systems \cite{davies1974markovian}. These algorithms have seen a marked improvement in recent years, ranging from ground state preparation routines with single-ancilla overhead \cite{ding2024single} to the first constructions that satisfy a discrete-time detailed balance condition \cite{gilyen2024quantumgeneralizationsglaubermetropolis}. The correctness of many of these algorithms, such as \cite{ding2024efficientquantumgibbssamplers}, is based on satisfying the Kubo-Martin-Schwinger (KMS) condition (\cite{kms2}, \cite{kms1}) which guarantees that the thermal state is a fixed point of the dynamics. The literature on this class of algorithms is already significant and continues to grow, so we point the reader to \cite{gilyen2024quantumgeneralizationsglaubermetropolis}, \cite{dalzell2023quantumalgorithmssurveyapplications}, \cite{chen2023quantumthermalstatepreparation} and \cite{rouze2024efficientthermalizationuniversalquantum} for their thorough literature reviews. 

One of the main drawbacks to the above approaches is the sheer complexity of the resulting algorithms. These algorithms tend to rely on coherently weighted sums of Heisenberg evolved jump operators and the construction of circuits to simulate the resulting Linbladians is nontrivial, as mentioned in Section 1.2 of \cite{gilyen2024quantumgeneralizationsglaubermetropolis}. Further, these algorithms tend to require logarithmically more ancilla qubits to allow for the addition of jump operators whereas our routine explicitly utilizes only a single ancilla qubit. Turning to ground states specifically, there exists single ancilla algorithms \cite{ding2024single}  but we remark that our channel is the first general purpose thermal state preparation routine for finite $\beta$ that utilizes only one qubit explicitly. Further, our routine avoids the complication of simulating weighted Linbladians and has incredibly simple circuits only relying on time independent Hamiltonian simulation and Haar 2-designs. 

Another algorithm that avoids some of the drawbacks of Linbladian simulation is the thermal state preparation routine by Shtanko and Movassagh \cite{shtanko2023preparingthermalstatesnoiseless}. They propose two algorithms, one tailored to a specific class of Hamiltonians and another more general algorithm. The first couples the system to a large bath of ancilla qubits via random $k$-local Pauli terms but is only proven to work for input Hamiltonians satisfying the Eigenstate Thermalization Hypothesis (ETH). The second couples the system to a bath of ancilla qubits via random circuit rotations. This second algorithm works for more generic Hamiltonians, but has much weaker guarantees on the convergence time needed and involves an extra rejection step similar, but not identical to, Temme et al \cite{temme2011}. Further, both of these algorithms have runtimes with explicit $\beta$ scaling which leads to a breakdown of their analysis in the $\beta \to \infty$ limit. Improving the analysis of these algorithms to more general Hamiltonians and to the ground state limit are important open problems, as mentioned in the conclusion of \cite{shtanko2023preparingthermalstatesnoiseless}.

In this paper, we propose a novel thermal state preparation routine inspired by the classical Hamiltonian Monte Carlo algorithm (\cite{neal1993probabilistic}, \cite{hoffman2011nouturnsampleradaptivelysetting}), which is also known as Hybrid Monte Carlo. Hamiltonian Monte Carlo resolves some of the previously mentioned issues, specifically the poor scaling with high dimensional systems, by introducing extra variables that serve as momentum and using time dynamics to generate state transitions. These extra momentum variables are sampled from their thermal distribution $\propto e^{-\beta p^2}$ and the total system is then simulated with classical Hamiltonian dynamics to generate updates to the state. Our channel was developed out of efforts to port this philosophy of using extra variables and time dynamics to prepare thermal states over to quantum systems. The extra states that we add are a single ancilla qubit prepared in a thermal state and we simulate the system-ancilla pair with a randomized interaction term $G$, see Section \ref{sec:prelim} and Eq. \eqref{eq:PhiDef} for more details. 

We analyze the resulting channel in a weak-coupling regime and obtain a reduction of the quantum dynamics to an underlying Markov chain over the eigenstates of the system. This allows us to show the correctness of the procedure fairly easily by finding the fixed point of the Markov chain directly. We explicitly do not rely on KMS conditions, instead directly proving conditions on the fixed points that more resembles classical detailed balance, see Eq. \eqref{eq:detailed_balance}. One potential downside to this approach is that the convergence time of our channel is now dependent on the mixing time of a classical Markov process, which is near impossible to know a priori. This is a feature of virtually all previous thermal state preparation routines and computing the spectral gaps that dictate these relaxation times is a major open problem. We are unable to compute these spectral gaps for general systems and finite $\beta$, but we resolve both open problems posed by Shtanko and Movassagh \cite{shtanko2023preparingthermalstatesnoiseless} by extending the range of systems that can be thermalized to all non-degenerate Hamiltonians and by computing the spectral gap of the Markov chain in the ground state, $\beta \to \infty$, limit.

Although our channel was inspired by the classical Hamiltonian Monte Carlo algorithm, the end result bears a striking resemblence to the Repeated Interactions (RI) framework studied in open quantum systems \cite{ciccarello2022repeated}. This framework involves coupling small environments, typically a single photon prepared in a thermal state, with the system and repeating the process over and over. This is then studied in thermodynamic limits of infinite time, infinite interactions, and weak coupling, which serves as a key difference with our work as we are concerned with bounding finite resources for use in an algorithm. Our work could be viewed as an extension of the RI framework to arbitrary Hamiltonians with completely unknown interactions. This also highlights a key difference between our work and existing RI literature, as previous RI work is typically concerned with specific systems and interactions \cite{polla2021quantum}. The RI framework has recently been fully worked out in the single qubit instance \cite{prositto2025equilibrium} where Prositto, Forbes, and Segal bound finite work and energy requirements needed for thermalization of a qubit with a single ancilla qubit. Their results are complimentary to our single qubit results in Section \ref{sec:single_qubit}. Recently there have been algorithms developed to take Linbladian simulations and convert them to RI simulations on quantum computers \cite{pocrnic2024quantumsimulationlindbladiandynamics}. It remains to be seen if the previously developed Linbladian based algorithms, in particular the single qubit ground state prep algorithm by Ding et al \cite{ding2024single}, resemble our results if this machinery is used to obtain the RI-equivalent algorithm.

\subsection{Main Results}
The remainder of the paper is split into three main parts. Section \ref{sec:weak_coupling} contains a derivation of the weak-coupling expansion in Lemma \ref{lem:big_one} and outlines the underlying Markov chain behavior in Section \ref{sec:markov}. This weak-coupling expansion is presented in as much generality as possible as it may be of use in other applications beyond our thermalization procedure, such as thermometry or spectroscopy. Section \ref{sec:specific_systems} has two theorems concerning single qubit systems and harmonic oscillators, Theorems \ref{thm:single_qubit} and \ref{thm:harmonic_oscillator} respectively, as well as numerics exploring the $\beta$ and $\epsilon$ dependence of the channel. Section \ref{sec:general_systems} contains our most general results in Theorems \ref{thm:zero_knowledge} and \ref{thm:perfect_knowledge} in which we show that the thermal state is an approximate fixed point for arbitrary Hamiltonians, bound the runtime in terms of a Markovian spectral gap, and finally compute this spectral gap for the ground state limit. The main difference between these two theorems is that one requires only an uppper bound on the spectral norm $\norm{H_S}$ while the other takes advantage of eigenvalue knowledge. Lastly numerics on small Hydrogen chains that explore strong coupling error rates and the effects of noisy eigenvalue samples on thermalization are presented in Section \ref{sec:general_numerics}. 

One of the key aspects of our thermalization procedure is that the analysis is dependent on the ability to tune the environment gap $\gamma$ to match the system energy differences. One of our main results in Theorem \ref{thm:zero_knowledge} shows that even if the user cannot tune $\gamma$ at all and is reduced to uniform guessing within an interval containing all the differences $\Delta_S(i,j)$, then thermalization can still occur. We show that the thermal state at finite $\beta$ is an approximate fixed state, with the error going to 0 as the coupling constant $\alpha \to 0$. This zero coupling limit can be taken with the opposite limit $t \to \infty$ to yield a nonzero simulation time for the random interaction $G$. Further, we show that the ground state is exactly the fixed point in the $\beta \to \infty$ limit. In this limit we are also able to bound the total simulation time required as $L \cdot t \in \bigotilde{\frac{\dim_S^{16} \norm{H_S}^7}{\delta_{\min}^8 \epsilon^6}} $, where $\delta_{\min}$ represents a ``resolution" type distance and is the smallest difference between two distinct eigenvalue differences $|\Delta_S(i,j) - \Delta_S(k,l)|$. When preparing finite $\beta$ thermal states we pick up an extra factor of $\frac{1}{\widetilde{\lambda}_\star(\beta)^7}$ related to the spectral gap of the transition matrix.

Although Theorem \ref{thm:zero_knowledge} is likely to be the more applicable result for most scenarios, as eigenvalue knowledge is just as difficult to obtain as preparing thermal states, in Theorem \ref{thm:perfect_knowledge} we are able to quantify the effects of this knowledge on the total simulation time needed by the channel. We obtain stronger results than the zero knowledge Theorem \ref{thm:zero_knowledge}: we find that the thermal state is the exact fixed point at all $\beta \in [0, \infty]$ and that the amount of simulation time needed scales as $L \cdot t \in \bigotilde{\frac{\dim_S^{16}}{\delta_{\min} \epsilon^{2.5} \lambda_\star^{3.5}}}$. This represents a reduction over the zero knowledge scenario of $\bigotilde{\frac{\norm{H_S}^7}{\delta_{\min}^7 \epsilon^{3.5} \widetilde{\lambda}_\star(\beta)^{3.5}}}$. We remark that although these asymptotic scalings are not optimal, they are qualitatively interesting as they represent the first analysis for thermal state preparation that eigenvalue heuristics have a quantifiable impact on. 

We explored this reduction in total simulation time numerically in Section \ref{sec:general_numerics} with small Hydrogen chains of 2 or 3 nuclei. We started off by measuring the total simulation time needed to reach a fixed $\epsilon = 0.05$ with perfect samples of eigenvalue differences $\Delta_S(i,j)$. We then tracked the total simulation time as we added noise to these guesses and find that low noise guesses can reduce the total simulation time by a factor of 2. These results that indicate that even very noisy guesses can still lead to thermalization are also complimented by our study of the error convergence for hydrogen chain systems, as for these tests we chose $\gamma$ from a Gaussian centered at the average energy $\trace{H_S} / \dim_S$ with a standard deviation of $\norm{H_S} / 2$. Even with this heuristic choice for $\gamma$ we can achieve fairly low errors of $\epsilon = 0.03$ in as little as $L = 1500$ interactions.


\section{Weak Interaction Expansion for $\Phi$} \label{sec:weak_coupling}
\subsection{Preliminaries and Notation} \label{sec:prelim}
We will be working with a bipartite Hilbert space consisting of a system space $\hilb_S$ with dynamics governed by the Hamiltonian $H_S$ and an environment space $\hilb_E$ with Hamiltonian $H_E$. The total space is $\hilb = \hilb_S \otimes \hilb_E$ with Hamiltonian $H = H_S \otimes \identity_E + \identity_E \otimes H_E = H_S + H_E$. We will assume without loss of generality that our spaces are encoded in qubits so that $\hilb_S = \mathbb{C}^{2^n}$ and $\hilb_E = \mathbb{C}^{2^m}$. We use $\dim_S$ to refer to the dimension of the system's Hilbert space ($2^n$), $\dim_E$ the environment, and $\dim$ the total Hilbert space. As for the basis we will use for our spaces, we will work directly in the eigenbasis of each Hamiltonian. Besides simplifying our calculations, we can do so because the interaction term we will introduce later is unitarily invariant. We denote these basis in a 1-indexed fashion as
\begin{equation}
    H_{S} = \sum_{i = 1}^{2^n} \lambda_S(i) \ketbra{i}{i} ~,~ H_{E} = \sum_{j=1}^{2^m } \lambda_E(j) \ketbra{j}{j} ~,~ H = \sum_{i=1}^{2^n } \sum_{j=1}^{2^m } \lambda(i,j) \ketbra{i,j}{i,j},
\end{equation}
where $\lambda(i,j) = \lambda_S(i) + \lambda_E(j)$ and we will sort the eigenvalues in nondecreasing order such that $i > j \implies \lambda_S(i) \geq \lambda_S(j)$. We note that the ground state in our 1-indexed notation is therefore $\ketbra{1}{1}$. We also make use of the following notation for the energy differences of the system-environment Hamiltonian and just the system
\begin{equation}
\Delta(i,j|k,l) \coloneqq \lambda(i,j) - \lambda(k,l), \quad \Delta_S(i,i') = \lambda_S(i) - \lambda_S(i'), \label{eq:delta_def}
\end{equation}
and because our eigenvalues are sorted $i > j \implies \Delta_S(i,j) \geq 0$. We will need a few other notations for eigenvalue differences. First we denote the degeneracy of an eigenvalue $\lambda(i,j)$ using $\eta(i,j)$ and the number of times a system eigenvalue \emph{difference} is present as $\eta_\Delta(i,j)$. For example, in a truncated harmonic oscillator with 4 energy levels the lowest gap $\Delta$ is present 3 times, so $\eta_\Delta(1, 2) = 3$. The second is that we will need to eventually analyze interferences between eigenvalue differences of the system, so we define
\begin{equation}
    \delta_{\min} \coloneqq \min_{\Delta_S(i,j) \neq \Delta_S(k,l)} \left| \Delta_S(i,j) - \Delta_S(k, l) \right|. \label{eq:delta_min_def}
\end{equation}
Note that nothing in this definition prevents one of the summands, say $\Delta_S(k,l)$, from being 0. This implies that $\delta_{\min} \leq \Delta_S(i,j)$ for all $i$ and $j$.

Currently our dynamics involved a system separated from the environment, so we need to fix this by adding an interaction term $G : \hilb_S \otimes \hilb_E \to \hilb_S \otimes \hilb_E$. We will choose $G$ randomly via the eigendecomposition 
\begin{equation}
    G = U_{\haar} D U_{\haar}^\dagger, U_{\haar} \sim {\rm Haar}(\hilb_S \otimes \hilb_E) \text{ and } D_{ii} \sim \mathcal{N}(0,1), \label{eq:interaction_def}
\end{equation}
where the eigenvectors are Haar distributed and the eigenvalues I.I.D. normal Gaussian variables. We then add this random interaction term to our system-environment dynamics with a coupling constant $\alpha$, yielding a total dynamics governed by $H_S + H_E + \alpha G$. We define the following rescaled coupling constant 
 \begin{equation}
    \widetilde{\alpha} \coloneqq \frac{\alpha t}{\sqrt{\dim + 1}}, \label{eq:a_tilde_def}
\end{equation}
where the $\dim$ is the total Hilbert space $\hilb$ dimension. The rescaling with respect to $\dim$ is to capture the factors of $1/(\dim + 1)$ in the transition amplitudes that appear later and leads to much more compact expressions.
This gives a decomposition of expectation values over $G$ into two parts 
\begin{equation}
    \mathbb{E}_G f(G) = \mathbb{E}_{\haar} \mathbb{E}_{D} f(G),
\end{equation}
where the two expectations on the right commute with each other $\mathbb{E}_{\haar} \mathbb{E}_{D} = \mathbb{E}_{D} \mathbb{E}_{\haar} $. 

We will use this interaction term to couple our system to an environment prepared in the thermal state $\rho_E(\beta) = e^{-\beta H_E} /\partfun_E(\beta)$, where $\partfun_E(\beta) = \trace{e^{-\beta H_E}}$, and then trace out the environment. This gives the definition of our thermalizing channel $\Phi : \mathcal{L}(\hilb_S) \to \mathcal{L}(\hilb_S)$ as
\begin{equation}\label{eq:PhiDef}
    \Phi(\rho ; \alpha, \beta, t) :=  \Tr_{\hilb_E} \mathbb{E}_{G} \left[ e^{+i(H + \alpha G)t} \rho \otimes \rho_E(\beta) e^{-i(H + \alpha G) t}\right]. 
\end{equation}
We will typically drop the implicit parameters of $\alpha, \beta$ and $t$. Our goal is to show how this channel can be used to prepare the system in the thermal state $\rho(\beta) = \frac{e^{-\beta H_S}}{\partfun(\beta)}$. It will be useful to introduce a fixed-interaction channel $\Phi_G : \mathcal{L}(\hilb_S \otimes \hilb_E) \to \mathcal{L}(\hilb_S \otimes \hilb_E)$ over the total Hilbert space $\hilb$ as 
\begin{equation}
    \Phi_G(\rho \otimes \rho_E; \alpha, t) \coloneqq e^{+i(H + \alpha t)} \rho \otimes \rho_E e^{- i(H + \alpha G)t}, \label{eq:phi_g_definition}
\end{equation}
giving us $\Phi(\rho; \alpha, \beta, t) = \Tr_{\hilb_E} \mathbb{E}_G \Phi_G(\rho \otimes \rho_E(\beta); \alpha, t)$. Another alternative notation for $\Phi$ that we will use is whenever $\hilb_E$ is a single qubit with energy gap $\gamma$ we will use $\Phi_\gamma$ to draw attention to this specific energy gap. We will also make frequent use of indicator functions, denoted $\mathbf{I}[P]$, which is 1 if the predicate $P$ is true and 0 if $P$ is false.

\subsection{First and Second Order Terms}
In order to understand our thermalizing channel $\Phi$ we will compute a Taylor Series for the output of the channel with respect to the coupling constant $\alpha$. We will perform the $\alpha$ expansion about $\alpha = 0$ and we will use the mean value form of the remainder, in which we are guaranteed a special value $\alpha_{\star} \in (0, \infty)$ such that the final derivative evaluated at $\alpha_{\star}$ is the exact amount needed. We use a second-order expansion and will need to explicitly compute terms up to order $\alpha^2$, which will give the following expansion
\begin{align}
    \Phi(\rho; \alpha) &= \Phi(\rho; 0) + \alpha \frac{\partial}{\partial \alpha} \Phi(\rho; \alpha) \big|_{\alpha = 0} + \frac{\alpha^2}{2} \frac{\partial^2}{\partial \alpha^2} \Phi(\rho; \alpha) \big|_{\alpha = 0} + R_{\Phi}(\rho; \alpha_{\star}) \label{eq:phi_taylor_series}
\end{align}
We use
\begin{equation}
    \mathcal{T}(\rho) \coloneqq \frac{\alpha^2}{2} \frac{\partial^2}{\partial \alpha^2} \Phi(\rho; \alpha) \bigg|_{\alpha = 0} = \frac{\alpha^2}{2}  \Tr_{\hilb_E} \mathbb{E}_{G} \left[\frac{\partial^2}{\partial \alpha^2} \Phi_G(\rho; \alpha) \big|_{\alpha = 0}\right] \label{eq:transition_def}
\end{equation} to denote the transition terms, as it will be revealed that the first two terms do not cause transitions in the system state, and $R_{\Phi}$ to denote the remainder. 
Further we will often leave the dependence on the  $\alpha$ parameter implicit and only include it when necessary.

We start off with the $\bigo{\alpha^0}$ term, which can be trivially computed as
\begin{align}
\Phi(\rho; 0) = \Tr_{\hilb_E}\int e^{i(H + \alpha G) t} \rho \otimes \rho_E(\beta) e^{-i (H + \alpha G) t} dG \bigg|_{\alpha = 0} = e^{i H t} \rho e^{-i H t}.
\end{align}
We then see that if $[ \rho, H] = 0$ then $\Phi(\rho; 0) = \identity(\rho)$, and as we restrict ourselves to such input states we will use this throughout the remainder of the paper. The next order correction is the $\bigo{\alpha^1}$ term. 
\begin{theorem} \label{lem:first_order_phi}
Let $\Phi$ be the thermalizing quantum channel given by Eq.\eqref{eq:PhiDef} and $G$ the randomly chosen interaction term as given by Eq. \eqref{eq:interaction_def}. The $O(\alpha)$ term in the weak-coupling expansion in Eq. \eqref{eq:phi_taylor_series} vanishes
   \begin{equation}
        \frac{\partial}{\partial \alpha} \Phi(\rho; \alpha) \big|_{\alpha = 0} = 0.
   \end{equation}
\end{theorem}
\begin{proof}
    We start by moving the $\alpha$ derivative through the linear operations of partial tracing and integrals so that it can act on the fixed interaction map $\Phi_G$
    \begin{align}
        \frac{\partial}{\partial \alpha} \Phi(\rho) \bigg|_{\alpha = 0} &= \frac{\partial}{\partial \alpha} \partrace{\mathcal{H}_E}{\int \Phi_G(\rho) dG} \bigg|_{\alpha = 0} \\
         &= \partrace{\mathcal{H}_E}{\int \frac{\partial}{\partial \alpha} \Phi_G(\rho) dG \bigg|_{\alpha = 0} } .
    \end{align}
    Now we use the expression for $\Phi_G$ in Eq. \eqref{eq:phi_g_definition} to compute the derivatives,
    \begin{align}
        \frac{\partial}{\partial \alpha} \Phi_G (\rho) &= \parens{\frac{\partial}{\partial \alpha} e^{+ i (H + \alpha G)t}} \rho \otimes \rho_E e^{-i (H + \alpha G) t} + e^{+i (H + \alpha G)t} \rho \otimes \rho_E \parens{\frac{\partial}{\partial \alpha} e^{- i (H + \alpha G)t}} \\
        &= \parens{\int_{0}^{1} e^{i s (H+\alpha G)t} (i t G) e^{i (1-s) (H+\alpha G)t} ds} \rho \otimes \rho_E e^{-i(H+\alpha G)t} \nonumber \\
    &~ ~+ e^{i(H+\alpha G)t} \rho \otimes \rho_E \parens{\int_{0}^1 e^{-i s (H+\alpha G) t} (- i t G) e^{-i (1-s) (H+\alpha G)t} ds}. \label{eq:first_order_alpha_derivative}
    \end{align}
    Now we can set $\alpha = 0$ in the above and introduce the expectation over $G$ that will be required
    \begin{align}
        &\mathbb{E}_G\left[ \frac{\partial}{\partial \alpha} \Phi_G(\rho) \bigg|_{\alpha = 0}\right] = i t \mathbb{E}_G \int_0^1 e^{i s H t} G e^{-i s H t} ds e^{i H t} \rho \otimes \rho_E e^{-i H t} \nonumber\\
&- i t e^{+i H t} \rho \otimes \rho_E \mathbb{E}_G \int_0^1 e^{-is H t} G e^{-i(1-s) Ht} ds \\ 
        &\quad= i t \int_0^1 e^{i s H t} \mathbb{E}_G[G] e^{-i s H t} ds e^{i H t} \rho \otimes \rho_E e^{-i H t} - i t e^{+i H t} \rho \otimes \rho_E \int_0^1 e^{-is H t} \mathbb{E}_G[G] e^{-i(1-s) Ht} ds.
    \end{align}
    Since our eigenvalues, $D_{ii}$, are mean zero ($\int D dD = 0$) we can compute $\mathbb{E}_G [G] $ and arrive at the lemma statement
    \begin{align}
        \mathbb{E}_G [G] &= \int G dG = \int \int U_{\haar} D U_{\haar}^\dagger dD dU_{\haar} = \int U_{\haar} \left( \int D dD \right) U_{\haar}^\dagger dU_{\haar} = 0.
    \end{align}
\end{proof}

Now we move on to the $O(\alpha^2)$ term in the weak-coupling expansion of $\Phi$. We first will compute the combined system-environment output of a generic system-environment basis state and we note that this result holds for an arbitrary dimension environment. We will use this to draw two results:  the first being for a single qubit environment the transition amplitudes of just the system can be split into on-resonance and off-resonance terms based on the tuning of the environment qubit Hamiltonian. The second result is that coherences are not introduced to the state at this order of $\Phi$, meaning if an input density matrix $\rho$ is diagonal then $(\identity + \mathcal{T})(\rho)$ will also be diagonal. This will be crucial for our later understanding of the channel as a Markov chain.
\begin{restatable}{lemma}{secondOrderChannelHaar} \label{lem:big_one}
    Given a system Hamiltonian $H_{S}$, an environment Hamiltonian $H_{E}$, a simulation time $t$, and coupling coefficient $\alpha$, let $\Phi_G$ denote the time evolution channel under a fixed interaction term $G$ as given in Eq. \eqref{eq:phi_g_definition}, let $\chi$ denote the following coherence prefactor
$$ \chi(i,j) \coloneqq \sum_{a,b: \Delta(i,j,|a,b) \neq 0} \frac{1 - i \Delta(i,j|a,b)t - e^{-i \Delta(i,j|a,b) t}}{\Delta(i,j|a,b)^2}, $$
and let $\eta(i,j)$ denote the degeneracy of the $(i,j)$\ts{th} eigenvalue of $H = H_S + H_E$. Then the $O(\alpha^2)$ term of $\Phi_G$ in a weak-coupling expansion is given by
 \begin{align}
 \frac{\alpha^2}{2} \mathbb{E}_G \left[ \frac{\partial^2}{\partial \alpha^2} \Phi_G(\ketbra{i,j}{k,l}) \big|_{\alpha = 0} \right] &= -\frac{\alpha^2  e^{i \Delta(i,j|k,l) t}}{\dim + 1} \bigg(\chi(i,j) + \chi(k,l)^*  + \frac{t^2}{2}(\eta(i,j) + \eta(k,l)) \bigg) \ketbra{i,j}{k,l} \nonumber \\
    &~ + \braket{i,j}{k,l}  \frac{\alpha^2 t^2}{\dim+1} \sum_{a,b} \sinc^2 \left( \frac{\Delta(i,j|a,b)t}{2} \right) \ketbra{a,b}{a,b}.  \label{eq:el_gigante}
 \end{align}

 For $\ket{i, j} = \ket{k, l}$ the above expression simplifies to
 \begin{align}
     &\frac{\alpha^2}{2} \mathbb{E}_G \left[ \frac{\partial^2}{\partial \alpha^2} \Phi_G(\ketbra{i,j}{i,j}) \big|_{\alpha = 0} \right] \nonumber \\
     &=  - \widetilde{\alpha}^2 \left(\sum_{(a,b) \neq (i,j)} \sinc^2 \left(\frac{\Delta(i,j | a,b)t}{2} \right) \right) \ketbra{i, j}{i,j} + \widetilde{\alpha}^2 \sum_{(a,b) \neq (i,j)} \sinc^2 \left(\frac{\Delta(i,j | a,b)t}{2} \right) \ketbra{a, b}{a,b} ,\label{eq:el_gigante_dos}
 \end{align}
 which also demonstrates that $\Tr \mathcal{T}(\rho) = 0$ for $\rho$ such that $[\rho, H_S] = 0$.
\end{restatable}
\noindent The proof of this Lemma uses similar techniques to Theorem \ref{lem:first_order_phi}, and formal proof can be found in Appendix \ref{sec:haar_integral_appendix}.

Next we will compute the effects of the channel on just the system alone which involves computing the partial trace $\Tr_{\hilb_E}$. We can either do this for a generic environment, which results in summations over $\hilb_E$ floating around, or specialize to a specific choice of $\hilb_E$ and compute the summation. For the remainder of this paper we will choose the latter option with a single qubit environment $\hilb_E = \mathbb{C}^2$ and denote the Hamiltonian $H_E = \begin{bmatrix} 0 & 0 \\ 0 & \gamma \end{bmatrix}$. Our environment input states then become
\begin{equation}
    \rho_E(\beta) = \frac{e^{-\beta H_E}}{\partfun_E(\beta)} = \frac{1}{1 + e^{-\beta \gamma}} \ketbra{0}{0} + \frac{e^{-\beta \gamma}}{1 + e^{-\beta \gamma}} \ketbra{1}{1} \eqqcolon q(0) \ketbra{0}{0} + q(1) \ketbra{1}{1} \label{eq:env_state_def},
\end{equation}
where we will use the environment qubit probabilities $q(0)$ and $q(1)$ in calculations for brevity. It will turn out that the value chosen for $\gamma$ is highly critical to the convergence of our algorithm, tuning it to match eigenvalue \emph{differences} of the system $H_S$ will allow us to analyze the convergence of the algorithm. As we can see in Eq. \eqref{eq:el_gigante} there will be a lot of $\sinc$ functions used, we will characterize a $\sinc$ function as being on-resonance or off-resonance if the inputs are sufficiently close to zero (the max for sinc). As for how close ``sufficiently close" actually is will depend on various parameters, such as $t, \alpha, \epsilon$, and the spectral properties of $H_S$.  
\begin{theorem}[Second-Order Expansion Term $\mathcal{T}$] \label{thm:second_order_transition}
Let $\mathcal{T}$ denote the second-order correction for a weak coupling expansion for a thermalizing channel $\Phi$ with a single qubit environment. Let 0 and $\gamma$ be the eigenvalues of the environment hamiltonian $H_E$ and the state of the environment be given by $\rho_E(\beta) = e^{-\beta H_E} / \partfun_E(\beta)$. The following properties hold for the second order correction terms.
\begin{enumerate}
    \item 
The transition element from system state $\ketbra{i}{i}$ to $\ketbra{j}{j}$, for $i \neq j$, is given by
\begin{align}
    \bra{j}\mathcal{T}(\ketbra{i}{i}) \ket{j} = \widetilde{\alpha}^2 \Biggr(& \sinc^2 \left( \frac{\Delta_S(i,j)t}{2} \right) + \frac{1}{1 + e^{-\beta \gamma}} \sinc^2 \left( \frac{(\Delta_S(i,j) - \gamma)t}{2} \right) \nonumber\\
    &\quad+  \frac{e^{-\beta \gamma}}{1 + e^{-\beta \gamma}} \sinc^2 \left( \frac{(\Delta_S(i,j) + \gamma)t}{2} \right) \Biggr). \label{eq:transition_terms_total}
\end{align}
\item For same-state transitions $\ketbra{i}{i}$ to $\ketbra{i}{i}$ we have
\begin{equation}
    \bra{i}\mathcal{T}(\ketbra{i}{i})\ket{i} = - \sum_{j \neq i} \bra{j} \mathcal{T}(\ketbra{i}{i}) \ket{j},
\end{equation}
which follows from $\Tr \mathcal{T}(\rho) = 0$ as shown in Lemma \ref{lem:big_one}. 
\item There are no coherences, or off-diagonal density matrix elements, introduced to the system up to second order in $\alpha$
\begin{equation}
    j \neq k \implies \bra{j} \mathcal{T}(\ketbra{i}{i}) \ket{k} = 0.
\end{equation}

\end{enumerate}
\end{theorem}
In the proof of the theorem, as well as in subsequent propositions, it is useful to introduce the concept of on- and off-resonant couplings which we give below.
\begin{definition}[Definition of Transition Matrix Elements]\label{def:transition}
The transition elements in Eq. \eqref{eq:transition_terms_total} can be divided into on-resonance and off-resonance terms based on the arguments to the sinc function. We define the on-resonance transitions as
\begin{align}
    \bra{j} \TT_{\on}(\ketbra{i}{i})\ket{j} &\coloneqq \widetilde{\alpha}^2 \frac{1}{1 + e^{-\beta \gamma}} \mathbf{I}[|\Delta_S(i,j) - \gamma| \le \delta_{\min}]  \sinc^2\left(\frac{(\Delta_S(i,j) - \gamma)t}{2}\right) \nonumber \\
    &~+ \widetilde{\alpha}^2 \frac{e^{-\beta \gamma}}{1 + e^{-\beta \gamma}} \mathbf{I}[|\Delta_S(i,j) + \gamma| \le \delta_{\min}]  \sinc^2\left(\frac{(\Delta_S(i,j) + \gamma)t}{2}\right) \label{eq:on_resonance}
\end{align}
and the off-resonance terms as
\begin{equation}
\begin{split}
    \bra{j} \TT_{\off}(\ketbra{i}{i})\ket{j} &\coloneqq \widetilde{\alpha}^2 \frac{1}{1 + e^{-\beta \gamma}} \mathbf{I}[|\Delta_S(i,j) - \gamma| > \delta_{\min}]  \sinc^2\left(\frac{(\Delta_S(i,j) - \gamma)t}{2}\right)  \\
    &~+ \widetilde{\alpha}^2 \frac{e^{-\beta \gamma}}{1 + e^{-\beta \gamma}} \mathbf{I}[|\Delta_S(i,j) + \gamma| > \delta_{\min}]  \sinc^2\left(\frac{(\Delta_S(i,j) + \gamma)t}{2}\right)  \\
    &~+ \widetilde{\alpha}^2 \sinc^2 \left( \frac{\Delta_S(i,j)t}{2} \right). \label{eq:off_resonance}
    \end{split}
\end{equation}
For same-state transitions $\ketbra{i}{i}$ to $\ketbra{i}{i}$ the on- and off-resonance transitions are equal to
\begin{equation}
    \bra{i} \TT_{\on}(\ketbra{i}{i}) \ket{i} = - \sum_{j \neq i}\bra{j} \TT_{\on}(\ketbra{i}{i}) \ket{j} \text{ and } \bra{i} \TT_{\off}(\ketbra{i}{i}) \ket{i} = - \sum_{j \neq i}\bra{j} \TT_{\off}(\ketbra{i}{i}) \ket{j}. \label{eq:same_state_transition_resonances}
\end{equation}
\end{definition}
We will now use these definitions to prove Theorem~\ref{thm:second_order_transition}.
\begin{proof}[Proof of Theorem~\ref{thm:second_order_transition}]
    The bulk of this proof will be based on straightforward reductions from Eq. \eqref{eq:el_gigante}. To start we will first show that no off-diagonal elements are introduced to the density matrix. By taking the $(j,k)$ matrix element of the output from Eq. \eqref{eq:el_gigante} we see
    \begin{align}
        \bra{j} \mathcal{T}(\ketbra{i}{i})\ket{k} &= \sum_{l, m} \frac{e^{-\beta \lambda_E(m)}}{1 + e^{-\beta \lambda_E(m)}} \bra{j, l} \frac{\alpha^2}{2} \mathbb{E}_{G} \left[ \frac{\partial^2}{\partial \alpha^2} \Phi_G(\ketbra{i, m}{i, m}) \big|_{\alpha = 0} \right] \ket{k, l} \\
        &= - \sum_{l,m} q(m) \widetilde{\alpha}^2 (\chi(i,m) + \chi(i,m)^* + t^2 \eta(i,m)) \braket{j, l}{i, m} \braket{i, m}{k, l} \nonumber \\
        &~ + \sum_{l, m} q(m) \sum_{a,b} \widetilde{\alpha}^2 \sinc^2 \left(\frac{\Delta(i,m|a,b)t }{2} \right) \braket{j, l}{a,b} \braket{a, b}{k, l} \\
        &= 0,
    \end{align}
    where we introduce $q(m)$ for $m=0,1$ to be a placeholder for the prefactors in Eq. \eqref{eq:el_gigante} and the last equality is due to the fact that $j \neq k$ implies that $\braket{j, l}{i,m}$ and $\braket{i,m}{k, l}$ cannot both be nonzero and likewise for $\braket{j, l}{a,b}$ and $\braket{a,b}{k,l}$.
    
    Since we have shown that coherences are not introduced to our system we can focus on the transitions from diagonal entries to diagonal entries in $\rho$. We make heavy use of Eq. \eqref{eq:el_gigante_dos} which tells us that for $i \neq k$ the system-environment transition amplitude is
    \begin{equation}
        \frac{\alpha^2}{2}\bra{k, l} \mathbb{E}_G \left[ \frac{\partial^2}{\partial \alpha^2} \Phi_G(\ketbra{i, j}{i,j}) \big|_{\alpha = 0} \right] \ket{k, l} = \widetilde{\alpha}^2 \sinc^2 \left( \frac{\Delta(i,j | k, l) t}{2} \right). 
    \end{equation}
    Now because all the operations present in the above expression are linear we can compute this map for the initial environment state $\rho_E(\beta)$ straightforwardly. Taking the output of this linear combination and computing the trace over the environment then gives us the expression for $\mathcal{T}$ using the assumption that the environment is a single qubit we find using the definition of $\gamma$ and $\Delta_S$ in~\eqref{eq:delta_def}
    \begin{align}
        \bra{j} \mathcal{T}(\ketbra{i}{i}) \ket{j} &= \sum_{k, l} q(k) \frac{\alpha^2}{2}\bra{j, l} \mathbb{E}_G \left[ \frac{\partial^2}{\partial \alpha^2} \Phi_G(\ketbra{i, k}{i,k}) \big|_{\alpha = 0} \right] \ket{j, l} \\
        &= \widetilde{\alpha}^2 \sum_{k, l} q(k) \sinc^2 \left(\frac{\Delta(i, k | j , l) t}{2} \right) \\
        &= \widetilde{\alpha}^2 \left(q(0) \sinc^2 \left(\frac{\Delta(i, 0 | j , 0) t}{2} \right) + q(0) \sinc^2 \left(\frac{\Delta(i, 0 | j , 1) t}{2} \right) \right) \nonumber \\
        & ~+ \widetilde{\alpha}^2 \left(q(1) \sinc^2 \left(\frac{\Delta(i, 1 | j , 0) t}{2} \right) + q(1) \sinc^2 \left(\frac{\Delta(i, 1 | j , 1) t}{2} \right) \right) \\
        &= \widetilde{\alpha}^2 \left(q(0) \sinc^2 \left(\frac{\Delta_S(i,j) t}{2} \right) + q(0) \sinc^2 \left(\frac{(\Delta_S(i,j) - \gamma) t}{2} \right) \right) \nonumber \\
        & ~+ \widetilde{\alpha}^2 \left(q(1) \sinc^2 \left(\frac{(\Delta_S(i, j) + \gamma) t}{2} \right) + q(1) \sinc^2 \left(\frac{\Delta_S(i,j) t}{2} \right) \right),
    \end{align}
    where we see that combining the terms with $\sinc^2 \left(\frac{\Delta_S(i,j) t}{2} \right)$, as $q(0) + q(1) = 1$, we immediately get Eq. \eqref{eq:transition_terms_total}. 

    To classify these terms as on-resonance or off-resonance we will focus on the argument to the sinc function, which is of the form $\Delta_S(i,j) t/ 2$ or $(\Delta_S(i,j) \pm \gamma) t/ 2$. The idea is that we will take $t$ large enough so that only the energy differences that are less than $\delta_{\min}$, as defined in Eq. \eqref{eq:delta_min_def}, will be non-negligible. Clearly the term $\widetilde{\alpha}^2 \sinc^2 \left( \frac{\Delta_S(i,j)t}{2} \right)$ will always be off-resonance, as $\delta_{\min} \le \Delta_S(i,j)$. 
    
    Now we have three terms to classify as either on-resonance or off-resonance, we refer to each term by their argument to the $\sinc$ function. The first we can categorically declare as being off-resonance is the $\Delta_S(i,j)$ term. By Lemma \ref{lem:sinc_poly_approx} we know $\sinc^2(\Delta_S(i,j) t/ 2) \le 4 / (\delta_{\min}^2 t^2)$, which we will make arbitrarily small in later sections. The other two can only be classified as on or off resonance depending if $\Delta_S(i,j)$ is positive or negative. If $i > j$ then we know that $\Delta_S(i,j) \ge 0$ and therefore $\sinc^2((\Delta_S(i,j) - \gamma)t/2)$ term can be close to 1 if $\gamma \approx \Delta_{S}(i,j)$, which also shows the $\Delta_S(i,j) + \gamma$ term is off-resonance for all $\gamma$. We say that the $\Delta_S(i,j) - \gamma$ term in this scenario is on-resonance if $|\Delta_S(i,j) - \gamma| \le \delta_{\min}$. This classification is best described symbolically as
    \begin{equation}
    i > j \text{ and } |\Delta(i,j) - \gamma| \le \delta_{\min} \implies  \bra{j} \mathcal{T}_{\on}(\ketbra{i}{i}) \ket{j} = \widetilde{\alpha}^2 q(0) \sinc^2 \left( \frac{(\Delta_S(i,j) - \gamma)t}{2} \right).\label{eq:on_resonance_i_geq_j}
    \end{equation}
    The $q(0)$ prefactor indicates that the ancilla started in it's low energy state and since $\sinc^2$ is symmetric we can write the argument as $\gamma - \Delta_S(i,j)$ which can be remembered as the ancilla gaining $\gamma$ amount of energy and the system losing $\Delta_S(i,j)$. In this scenario the $\Delta_S(i,j) + \gamma$ term is therefore put in the off-resonance map
    \begin{align}
        i > j \text{ and } |\Delta(i,j) - \gamma| \le \delta_{\min} \implies  \bra{j} \mathcal{T}_{\off}(\ketbra{i}{i}) \ket{j} = \widetilde{\alpha}^2\left( \sinc^2 \left( \frac{\Delta_S(i,j) t}{2} \right) + q(1) \sinc^2 \left(\frac{(\Delta_S(i,j) + \gamma) t}{2} \right) \right).
    \end{align}
    
    Now for $i < j$ we find that the on-resonance term is 
    \begin{equation}
    i < j \text{ and } |\Delta_S(i,j) + \gamma| \le \delta_{\min} \implies \bra{j} \mathcal{T}_{\on}(\ketbra{i}{i}) \ket{j} = \widetilde{\alpha}^2 q(1) \sinc^2 \left( \frac{(\Delta_S(i,j) + \gamma)t}{2} \right). \label{eq:on_resonance_i_leq_j}
    \end{equation}
    Similarly to before the $q(1)$ prefactor tells us the ancilla starts in the excited state. This matches with the energy argument by noting that $\Delta_S(i,j) \le 0$ and that the argument to $\sinc$ is symmetric, which allows us to write it as $|\Delta_S(i,j)| - \gamma$; indicating that the system gains energy $|\Delta_S(i,j)|$ and the ancilla energy drops by $-\gamma$. In this scenario the $\Delta_S(i,j) - \gamma$ term is off-resonance and we have
    \begin{align}
        i < j \text{ and } |\Delta(i,j) + \gamma| \le \delta_{\min} \implies  \bra{j} \mathcal{T}_{\off}(\ketbra{i}{i}) \ket{j} = \widetilde{\alpha}^2\left( \sinc^2 \left( \frac{\Delta_S(i,j) t}{2} \right) + q(0) \sinc^2 \left(\frac{(\Delta_S(i,j) - \gamma) t}{2} \right) \right).
    \end{align}

    Now to compute the $i = j$ case, it is sufficient to utilize our results from the $i \neq j$ scenario. This is because our second order correction has zero trace $\trace{\mathcal{T}(\rho)} = 0$ from Theorem~\ref{thm:second_order_transition}, so we can define the on-resonance and off-resonance terms as the following
    \begin{align}
        \bra{i} \mathcal{T}(\ketbra{i}{i}) \ket{i} &= - \widetilde{\alpha}^2 \sum_{k \neq i} \bra{k} \mathcal{T}(\ketbra{i}{i}) \ket{k} \\
        &= - \widetilde{\alpha}^2 \sum_{k \neq i} \bra{k} \left( \mathcal{T}_{\on}(\ketbra{i}{i}) + \mathcal{T}_{\off}(\ketbra{i}{i}) \right) \ket{k} \\
        &\eqqcolon  \bra{i} \mathcal{T}_{\on}(\ketbra{i}{i}) \ket{i} + \bra{i} \mathcal{T}_{\off}(\ketbra{i}{i}) \ket{i}.
    \end{align}
    By plugging in Eqs. \eqref{eq:on_resonance_i_geq_j} and \eqref{eq:on_resonance_i_leq_j} we are done with the self-transition terms. 
\end{proof}

\subsection{Markovian Dynamics and Error Terms} \label{sec:markov}

Now that we have fully computed the significant contributors to the output of our channel $\Phi$, we move on to characterize the behavior of the channel as a Markov chain with noise. 
A Markov chain is a random process that involves a walker transitioning to vertices on a graph wherein the probability of transition does not depend on the history of the walker.  Specifically, in this context we view the vertices in this graph as the eigenstates of the Hamiltonian.  The repeated interaction model because of the lack of coherences in the weak coupling limit can be interpreted as a Markov process over these eigenstates with transitions probabilities given by the above analysis. 

Specifically, the Markov chain is dictated by the $\Phi(\rho; 0)$ and $\mathcal{T}_{\on}$ terms in the weak-coupling expansion, for $[\rho, H_S] = 0$ we showed that $\Phi(\rho; 0) = \identity(\rho)$, so from now on we will specifically only deal with such density matrices and characterize the zeroth order term as an identity map. As for the Markov chain, we will use normal font to denote matrices, such as $I$ for the identity matrix and $T$ for the transition term added on. We use $e_i$ to denote the basis vector associated with the quantum state $\ketbra{i}{i}$ and $p$ to denote the probability vector for $\rho$ associated with its eigenvalues.
\begin{lemma}[Quantum Dynamics to Classical Markov Chain] \label{lem:quantum_to_classical}
    Let $T$ be the matrix defined by 
    \begin{equation}
        e_i^T T e_j \coloneqq \bra{i} \mathcal{T}_{\on}(\ketbra{j}{j}) \ket{i}.
    \end{equation}
    The matrix $I + T$ is a column stochastic matrix and models the Markovian dynamics of our thermalizing channel up to $\bigo{\alpha^2 t^2}$,
    \begin{equation}
        \bra{j} (\identity + \mathcal{T}_{\on})^{\circ L} (\ketbra{i}{i}) \ket{j} = e_j^T (I + T)^L e_i.
    \end{equation}
    By linearity of $\identity + \mathcal{T}_{\on}$ this identity extends to any diagonal density matrix input $\rho = \sum_i p(i) \ketbra{i}{i}$.
\end{lemma}
\begin{proof}
    We prove this inductively on $L$. The base case of $L = 1$ is trivial from the defintion of $T$
    \begin{align}
    \bra{j} (\identity + \mathcal{T}_{\on})(\ketbra{i}{i}) \ket{j} = \delta_{i,j} + \bra{j} \mathcal{T}_{\on}(\ketbra{i}{i}) \ket{j} = e_j^T (I +  T) e_i.
\end{align}
For the inductive step we will rely on the fact that there are no off-diagonal elements for diagonal inputs. 
\begin{align}
    \bra{j} \mathcal{T}_{\on} (\ketbra{i}{i}) \ket{k} = \delta_{j,k} \bra{j} \mathcal{T}_{\on} (\ketbra{i}{i}) \ket{j} \implies \bra{j} \mathcal{T}_{\on}^{\circ L} (\ketbra{i}{i}) \ket{k} = \delta_{j,k} \bra{j} \mathcal{T}_{\on}^{\circ L} (\ketbra{i}{i}) \ket{j}.
\end{align}
This is again by induction where the case $L = 1$ is proved in Theorem \ref{thm:second_order_transition} and the inductive step is 
\begin{align}
    \bra{j} \mathcal{T}_{\on}^{\circ L} (\ketbra{i}{i}) \ket{k} &= \bra{j} \mathcal{T}_{\on} \left( \mathcal{T}_{\on}^{\circ L - 1} (\ketbra{i}{i}) \right) \ket{k} \\
    &= \sum_{m, n} \bra{j} \mathcal{T}_{\on}\left( \ketbra{m}{m}\mathcal{T}_{\on}^{\circ L - 1}(\ketbra{i}{i}) \ketbra{n}{n}\right) \ket{k} \\
    &= \sum_{m, n} \delta_{m,n} \bra{m}\mathcal{T}_{\on}^{\circ L - 1}(\ketbra{i}{i}) \ket{m} \bra{j} \mathcal{T}_{\on}\left(   \ketbra{m}{m} \right) \ket{k} \\
    &= \sum_{m} \bra{m}\mathcal{T}_{\on}^{\circ L - 1}(\ketbra{i}{i}) \ket{m} \delta_{j,k} \bra{j} \mathcal{T}_{\on}\left(   \ketbra{m}{m} \right) \ket{j} \\
    &= \delta_{j,k} \bra{j} \mathcal{T}_{\on}^{\circ L} (\ketbra{i}{i}) \ket{j}.
\end{align} 
This argument points the way towards how we will prove the inductive step in our stochastic conversion, starting with
\begin{align}
    \bra{j} (\identity + \mathcal{T}_{\on})^{\circ L}(\ketbra{i}{i}) \ket{j} &= \bra{j} \left( (\identity + \mathcal{T}_{\on})^{\circ L - 1} (\ketbra{i}{i}) + \mathcal{T}_{\on} \circ (\identity + \mathcal{T}_{\on})^{\circ L - 1} (\ketbra{i}{i}) \right)\ket{j} \\
    &= e_j^T (\identity + T)^{L - 1} e_i + \bra{j} \mathcal{T}_{\on} \circ (\identity + \mathcal{T}_{\on})^{\circ L - 1} (\ketbra{i}{i}) \ket{j} . \label{eq:matrix_reloaded1}
\end{align}
We can use the inductive hypothesis on the term on the left and we now have to break down the $\mathcal{T}_{\on}$ term. 
\begin{align}
    \bra{j} \mathcal{T}_{\on} \circ (\identity + \mathcal{T})^{\circ L - 1} (\ketbra{i}{i}) \ket{j} &= \sum_{m, n} \bra{j} \mathcal{T}_{\on}\left( \ketbra{m}{m} (\identity + \mathcal{T}_{\on})^{\circ L - 1} (\ketbra{i}{i}) \ketbra{n}{n} \right) \ket{j} \\
    &= \sum_{m} \bra{j} \mathcal{T}_{\on} \left( \ketbra{m}{m} \right) \ket{j} e_m^T (I + T)^{L - 1} e_i \\
    &= \sum_m e_j^T T e_m e_m^T (I + T)^{L -1} e_i \\
    &= e_j^T T(I + T)^{L-1} e_i.
\end{align}
Substituting this into Eq. \eqref{eq:matrix_reloaded1} yields \begin{equation}
     \bra{j} (\identity + \mathcal{T}_{\on})^{\circ L}(\ketbra{i}{i}) \ket{j} = e_j^T (I + T)^{L} e_i.
\end{equation}

Our final step in the proof is to show that $I + T$ is column-stochastic. This is straightforward from our definition of $T$
\begin{equation}
    \sum_i e_i^T (I + T) e_j = 1 + \sum_i \bra{i} \mathcal{T}_{\on}(\ketbra{j}{j}) \ket{i}.
\end{equation}
Now we use the fact that $\bra{j} \TT_{\on}(\ketbra{j}{j}) \ket{j} = - \sum_{i \neq j} \bra{i} \TT_{\on}(\ketbra{j}{j}) \ket{i}$ from Eq. \eqref{eq:same_state_transition_resonances} to conclude that $I + T$ is column stochastic.
\end{proof}

Since we will be effectively reducing our quantum dynamics to classical dynamics over the eigenbasis for $H_S$ we will need bounds on the convergence of Markov chains. This is a very deep area of research, with many decades of results, so we point interested readers to the comprehensive book by Levin and Peres \cite{levin2017markov}. As we will be dealing with non-reversible Markov chains we unfortunately cannot use the relatively well-developed theory for reversible Markov chains. Luckily, we will only need the following theorem due to Jerison.
\begin{theorem}(Jerison's Markov Relaxation Time \cite{jerison2013general}) \label{thm:markov_chain_bound}
    Let $M : \mathbb{R}^{N} \to  \mathbb{R}^{N}$ be an ergodic Markov transition matrix acting on an $N$ dimensional state space with absolute spectral gap $\lambda_{\star} \coloneqq 1 - \max_{i > 1} |\lambda_i(M)|$, where the eigenvalues of $M$ are ordered $1 = \lambda_1 \ge \lambda_2 \ge \ldots \ge \lambda_N \geq -1$. Given this gap, if the number of steps $L$ in the Markov chain satisfies the following bound
    \begin{align}
        L &\ge \frac{N}{\lambda_{\star}} \left( 2\log \frac{1}{\lambda_{\star}} + 4(1 + \log 2) +  \frac{1}{N} (2 \log \left( \frac{1}{\epsilon} \right) - 1) \right) \eqqcolon \frac{N}{\lambda_\star} J,
    \end{align}
    where $J$ is the collection of logarithmic and constant terms that we will typically ignore in asymptotic notation, then the resulting state $M^L \vec{x}$ is $\epsilon$ close to the fixed point
    \begin{equation}
        \forall \vec{x} \text{ s.t. } x_i \ge 0 \text{ and } \sum_i x_i = 1, \quad \norm{\vec{\pi} - M^L \vec{x}}_1 \le \epsilon.
    \end{equation}
    We use $\vec{\pi}$ to denote the unique eigenvector of eigenvalue 1 for $M$.
\end{theorem}

Now that we have an idea of how long it takes for our Markov chain to converge to the fixed points we need to show which states are actually fixed points. We demonstrate that for finite $\beta$ any fixed point must satisfy a summation of detailed-balance terms. This fixed point is unique if the Markov chain is ergodic, which we do not argue in this lemma as an arbitrary thermalization channel $\Phi$ may not be ergodic. For the ground state limit of $\beta \to \infty$ we show that the Markov matrix $I + T$ is upper triangular, which is crucial to our analysis of the spectral gap of the Markov chain in later results. We also demonstrate that the ground state is a fixed point in this limit nearly trivially.
\begin{lemma}[Markov Chain Fixed Points]\label{lem:fixed_point}
    Let $T$ be the transition matrix with sum zero columns $\sum_j e_j^T T e_i$ for all $i$, negative diagonal entries $e_i^T T e_i \leq 0$, and off-diagonals smaller than 1 $e_j^T T e_i \leq 1$ for $j \neq i$, associated with the on-resonance term $\TT_{\on}$ of an arbitrary thermalizing channel $\Phi$. A vector $\vec{p}$ is a fixed point of the Markovian dynamics $I + T$ if and only if it is in the kernel of $T$. This holds for finite $\beta$ if the following is satisfied for all $j$
    \begin{equation}
        \sum_{i \neq j} \frac{e^{-\beta \lambda_S(i)}}{\partfun_S(\beta)} e_j^T T e_i - \frac{e^{-\beta \lambda_S(j)}}{\partfun_S(\beta)} e_i^T T e_j = 0.\label{eq:detailed_balance}
    \end{equation}
    In the $\beta \to \infty$ limit the ground state $e_1$ is a fixed point and $T$ is upper triangular
    \begin{equation}
        \lim_{\beta \to \infty} (I + T) e_1 = e_1 \text{ and }i > j \implies \lim_{\beta \to \infty} e_i^T T e_j = 0.
    \end{equation}
\end{lemma}
\begin{proof}
    To show that the thermal state is the fixed point of the zero knowledge thermalizing channel we need to show that 
$T \vec{p}_{\beta} = 0$ and that the Markov chain is ergodic. Ergodicity will be easy to prove so we focus on showing that $T \vec{p}_{\beta} = 0$. This condition can be expressed as
\begin{equation}
    e_j^T T \vec{p}_{\beta} = \sum_i \frac{e^{-\beta \lambda_S(i)}}{\partfun_S(\beta)} e_j^T T e_i = 0.   \label{eq:thermal_state_tmp_1}
\end{equation}
We can make a quick substitution as we know the diagonal elements must equal the sum of the remainder of the column 
\begin{equation}
    e_i^T T e_i = - \sum_{j \neq i} e_j^T T e_i,
\end{equation}
which we can then pull out the $i = j$ term from the sum in Eq. \eqref{eq:thermal_state_tmp_1}
\begin{align}
    e_j^T T \vec{p}_{\beta} &= \sum_{i \neq j} \frac{e^{-\beta \lambda_S(i)}}{\partfun_S(\beta)} e_j^T T e_i - \frac{e^{-\beta \lambda_S(j)}}{\partfun_S(\beta)} \sum_{k \neq j} e_k^T T e_j,
\end{align}
which is 0 if and only if $\vec{p}_{\beta}$ is a fixed point of $I + T$.

We now show the $\beta \to \infty$ case. We can show that $T$ is upper triangular using Theorem \ref{thm:second_order_transition} which gives us the on-resonance transition amplitude. We assume $i < j$, which implies $\Delta_S(i,j) \le 0$, and get
    \begin{align}
        \lim_{\beta \to \infty} e_j^T T e_i &= \lim_{\beta \to \infty} \bra{j} \TT_{\on}(\ketbra{i}{i}) \ket{j} \\
        &= \widetilde{\alpha}^2 \lim_{\beta \to \infty} \left[ \frac{e^{-\beta \gamma}}{1 + e^{-\beta \gamma}} \mathbf{I}[|\Delta_S(i,j) + \gamma| \le \delta_{\min}] \sinc^2 \left(\frac{(\Delta_S(i, j) + \gamma) t}{2} \right) \right] \\
        &= \widetilde{\alpha}^2 \mathbf{I}[|\Delta_S(i,j) + \gamma| \le \delta_{\min}] \sinc^2 \left(\frac{(\Delta_S(i, j) + \gamma) t}{2} \right) \lim_{\beta \to \infty} \frac{e^{-\beta \gamma}}{1 + e^{-\beta \gamma}} \\
        &= 0.
    \end{align}
    This further shows that the ground state is a fixed point, as every other eigenvector must have higher energy and therefore all on-resonance transitions \emph{from} the ground state must be 0
    \begin{align}
        \lim_{\beta \to \infty} e_1^T T e_1 &= \lim_{\beta \to \infty} \bra{1} \TT_{\on}(\ketbra{1}{1}) \ket{1} = -  \sum_{j > 1} \lim_{\beta \to \infty} \bra{j} \TT_{\on}(\ketbra{1}{1}) \ket{j}  = 0.
    \end{align}
    This then shows that the ground state is fixed
    \begin{equation}
        (I + T) e_1 = e_1,
    \end{equation}
    and completes the proof.
\end{proof}

Using the decomposition from Theorem \ref{thm:second_order_transition} and intermediate expressions in its proof we can now show why the off-resonance map $\mathcal{T}_{\off}$ is named ``off-resonance"; even in the worst case scenario of choosing a bad value of $\gamma$ such that all terms in $\mathcal{T}$ end up in $\mathcal{T}_{\off}$ the trace norm of its output is always controllably small via $\alpha$.
\begin{corollary} \label{cor:t_off_norm}
    The induced trace norm of the off-resonance map $\mathcal{T}_{\off}(\rho)$, for all density matrices $\rho$ such that $[\rho, H_S] = 0$ and $\rm{dim}\ge 2$, is upper bounded for all choices of the environment Hamiltonian $\gamma$ by
    \begin{align}
        \norm{\mathcal{T}_{\off}(\rho)}_1 \le \frac{8 \alpha^2}{\delta_{\min}^2}
    \end{align}
\end{corollary}
\begin{proof}

This result follows from applying bounds on the sinc function from Lemma \ref{lem:sinc_poly_approx} (given in Appendix \ref{sec:appendix_sinc}) to the worst-case scenario off-resonance terms given in Eq. \eqref{eq:off_resonance}. 
    \begin{align}
        i \neq j \implies \abs{\bra{j}\mathcal{T}_{\off}(\ketbra{i}{i})\ket{j}} &\le \widetilde{\alpha}^2\frac{4}{\delta_{\min}^2 t^2} \left( 1 + q(0) + q(1)\right) = \frac{8 \alpha^2}{\delta_{\min}^2(\dim + 1)}.
    \end{align}
    This allows us to bound the off-resonance self-transition term in Theorem~\ref{thm:second_order_transition} as
    \begin{align}
        \abs{\bra{i}\mathcal{T}_{\off}(\ketbra{i}{i})\ket{i}} &= \abs{- \sum_{j \neq i} \bra{j} \mathcal{T}_{\off}(\ketbra{i}{i}) \ket{j}} \le (\dim_S - 1) \frac{8 \alpha^2}{\delta_{\min}^2 (\dim + 1)} \le \frac{4 \alpha^2}{\delta_{\min}^2}.
    \end{align}
    Now we can use this, along with our no off-diagonal output elements of $\mathcal{T}$, to compute the trace norm of the off-resonance map
    \begin{align}
        \norm{\mathcal{T}_{\off}(\rho)}_1 &= \sum_{j} \abs{\bra{j} \mathcal{T}_{\off}(\rho) \ket{j}} \\
        &\le \sum_{i, j} \rho_{i,i} \abs{\bra{j}\mathcal{T}_{\off}(\ketbra{i}{i})\ket{j}} \\
        &= \sum_{i} \rho_{i,i} \left(\sum_{j \neq i} |\bra{j} \mathcal{T}_{\off}(\ketbra{i}{i}) \ket{j}| + |\bra{i} \mathcal{T}_{\off}(\ketbra{i}{i})\ket{i}| \right) \\
        &\le \sum_{i} \rho_{i,i} \left( (\dim_S - 1) \frac{8 \alpha^2}{\delta_{\min}^2(\dim + 1)} + \frac{4 \alpha^2}{\delta_{\min}^2} \right) \\
        &\le \frac{8 \alpha^2}{\delta_{\min}^2}.
    \end{align}
\end{proof}

The last result in this section that we will need is a bound on the trace norm of the remainder term, which we state in the following theorem.
\begin{theorem}[Remainder Bound] \label{thm:remainder_bound}
    Let $R_{\Phi}(\rho)$ be the remainder term for the second-order Taylor series expansion 
of the quantum channel $\Phi$ acting on an input state $\rho$ about $\alpha=0$ defined via
    $$\Phi(\rho; \alpha) = \Phi(\rho; 0) + \alpha \frac{\partial}{\partial \alpha} \Phi(\rho; \alpha) \big|_{\alpha = 0} + \frac{\alpha^2}{2} \frac{\partial^2}{\partial \alpha^2} \Phi(\rho; \alpha) \big|_{\alpha = 0} + R_{\Phi}(\rho; \alpha_{\star})$$ where the Sch\"{a}tten 1-norm of the remainder operator is bounded above by
    \begin{equation}
        \norm{R_{\Phi}(\rho;\alpha)}_1 \le \frac{16 \sqrt{2}}{\sqrt{\pi}} \dim_S (\alpha t)^3.
    \end{equation}
\end{theorem}
The proof of the remainder bound follows from the triangle inequality and remainder bounds on Taylor series and is given in Section \ref{sec:weak_coupling_remainder_bound}.

\section{Specific Systems} \label{sec:specific_systems}
In this section we analyze the dynamics of the weak-coupling expansion from Section \ref{sec:weak_coupling} as a thermal state preparation routine for specific systems. We keep these specific systems separate from the generic results in Section \ref{sec:general_systems} for a few reasons. The first being that these systems are simple enough that most readers should have a solid grasp of their spectrums. The second being that they are interesting in their own right and may be of use for researchers in other fields than quantum algorithms design who are looking to implement thermal state preparation for their projects. For example, even the single qubit scenario in Section \ref{sec:single_qubit} is already complex enough that there is active research on this Hamiltonian in the field of open quantum systems \cite{prositto2025equilibrium}. Lastly, these scenarios are simple enough that they allow for concrete computation of their spectra, which can be a helpful warm-up for the more abstract settings later.

For a single qubit we are able to give asymptotic runtime bounds for preparing arbitrary $\beta$ thermal states with specific values of $\alpha$ and $t$ to guarantee a trace distance error of $\epsilon$ from the thermal state. Further, our results hold without exact knowledge of the system energy gap $\Delta$. We do require knowledge of a window of width $2 \sigma$ that contains $\Delta$, and our runtime is parametrized in terms of $\sigma$ with a simple expression in the exact knowledge limit of $\sigma \to 0$.

After we develop the theory for a single qubit we extend our sights to a truncated harmonic oscillator in Section \ref{sec:harmonic_oscillator}. This section can be viewed as a warmup to the proofs that work for general systems in Section \ref{sec:general_systems}. One unique aspect of the harmonic oscillator is that all of the energy eigenvalues are integer multiples of single difference $\Delta$, which means that by tuning our environment gap $\gamma$ to $\Delta$ we should be able to thermalize. One difficulty with this however is that ergodicity is now a bit harder to show, as with randomly choosing $\gamma$ there is always a nonzero probability to transition between arbitrary states $i$ and $j$ but with the harmonic oscillator multiple ``hops'' may be necessary if the states are not adjacent in energy eigenvalues.

Lastly, another reason to study these two specific systems is that they serve as crucial test beds for numerical studies. As our analysis tends to work with upper bounds to obtain realistic estimates of total simulation times we need to simulate the channel numerically. Although we do study more realistic systems numerically in Section \ref{sec:general_numerics}, larger systems are more costly in terms of classical computation needed to simulate. Because of the reduction in classical compute time needed, we are able to get a more complete view of our channel through the lens of the small systems in Section \ref{sec:specific_numerics}.

\subsection{Single Qubit} \label{sec:single_qubit}
The first system we study is the qubit $\hilb_S = \mathbb{C}^2$. This system is simple enough that we can explicitly write the dynamics as a $2 \times 2$ transition matrix, which makes it easy to compute required simulation times and easy for the reader to follow. Although this system could be viewed as a warmup to the more general systems in Section \ref{sec:general_systems}, as the proof techniques are very similar, we remark that this system does have some unique properties. The biggest difference is that we do not assume any kind of belief distribution of the eigenvalue gap $\Delta$ of the system. We only require that a window of width $2 \sigma$ is known that contains $\Delta$. We can then characterize the runtime in terms of $\sigma$ and in addition to determining runtime we find it determines an upper bound on the $\beta$ that can be prepared at low error.

The other unique phenomenon with the single qubit scenario is that the total simulation time needed is \emph{independent} of $\beta$. Although this may seem incorrect, as most existing thermal state preparation algorithms tend to scale at least linearly with $\beta$, this is in fact a property of the underlying Markov chain. The rate of convergence of the Markov chain is dictated by the spectral gap, which for this system is shown to be $\widetilde{\alpha}^2$. The only aspect of the Markov chain that changes with $\beta$ is what the fixed point is and the Markov Relaxation Theorem \ref{thm:markov_chain_bound} provides relaxation guarantees regardless of initial or final state.

\begin{theorem} \label{thm:single_qubit}
    Let $H_S$ be an arbitrary single qubit Hamiltonian with eigenvalue gap $\Delta$, $\rho$ any input state diagonal in the $H_S$ basis, $[\rho, H_S] = 0$, and $L$ the number of interactions simulated. Given a window of width $2 \sigma$ that is promised to contain $\Delta$ and satisfies the inequality
    \begin{align}
        \sigma \le \min \set{\frac{\epsilon}{2\beta}, \Delta \sqrt{\frac{\epsilon}{2}}},
    \end{align}
    then the following parameter choices 
    \begin{align}
        \alpha &= \frac{1}{t^3(\Delta + \sigma)^2}, \nonumber \\
        t &\in \frac{1}{\sigma}\left[\sqrt{1 - \sqrt{1 - \frac{2 \sigma^2}{\Delta^2 \epsilon}} }, \sqrt{1 + \sqrt{1 - \frac{2 \sigma^2}{\Delta^2 \epsilon}} } \right], \nonumber\\
        \text{and } L &= \left\lceil\frac{10}{\alpha^2 t^2 (1 - \sigma^2 t^2 / 2)} \left( 2 \log \frac{5}{\alpha^2 t^2 \sinc^2(|\Delta - \gamma|t/2)} + 4( 1 + \log 2) - \frac{1}{2} + \log \frac{2}{\epsilon} \right)\right\rceil
    \end{align}
    are sufficient to guarantee thermalization of the form $\norm{\rho_S(\beta) - \Phi^{\circ L}(\rho)}_1 \in \bigotilde{\epsilon}$. In the limit as $\sigma \to 0$, the total simulation time required for all $L$ interactions required scales as
    \begin{equation}
        L \cdot t \in \widetilde{O} \left( \frac{1}{\Delta \epsilon^{2.5}}\right).
    \end{equation}
\end{theorem}
\begin{proof}
    The proof will be structured into three parts. First, we will need a bound on how close the fixed point of the Markov chain is to the thermal state, because  the fixed point  is exactly the thermal state only when $\gamma = \Delta$ and our window of width $\sigma$ is sufficiently small given our error budget. Second, once we have these bounds we then need to determine the number of interactions $L$ that will be necessary to reach the fixed point within trace distance $\epsilon$. Lastly, we use this value of $L$ to bound the accumulative error from the off-resonance mapping $\mathcal{T}_{\off}$ and remainder term $R_{\Phi}$.

    We start by breaking down the trace distance into three components, one for the fixed-point distance from the thermal state, one for the Markov dynamics distance to the fixed-point, and lastly the remainder terms
    \begin{align}
        &\norm{\rho_S(\beta; \Delta) - \Phi^{\circ L}(\rho)}_1 \nonumber \\
        &\le \norm{\rho_S(\beta; \Delta) - \rho_S(\beta; \gamma)}_1 + \norm{\rho_S(\beta; \gamma) - \Phi^{\circ L}(\rho)}_1 \\
        &\le \norm{\rho_S(\beta; \Delta) - \rho_S(\beta; \gamma)}_1 + \norm{\rho_S(\beta; \gamma) - (\identity + \mathcal{T}_{\on})^{\circ L}(\rho)}_1 + \norm{(\identity + \mathcal{T}_{\on})^{\circ L}(\rho) - \Phi^{\circ L}(\rho)}_1 \\
        &\le \norm{\rho_S(\beta; \Delta) - \rho_S(\beta; \gamma)}_1 + \norm{\rho_S(\beta; \gamma) - (\identity + \mathcal{T}_{\on})^{\circ L}(\rho)}_1 +  L \left(\norm{\mathcal{T}_{\off}(\rho)}_1 + \norm{R_{\Phi}}_1 \right). \label{eq:single_qubit_three_errors}
    \end{align}
    We proceed with the leftmost term first. The trace distance can be computed explicitly for a single qubit state as 
 \begin{align}
     \norm{\rho_S(\beta; \gamma) - \rho_S(\beta; \Delta)}_1 &= \abs{\bra{1} \rho_S(\beta; \gamma) \ket{1} - \bra{1}\rho_S(\beta;\Delta)\ket{1}} + \abs{\bra{2} \rho_S(\beta; \gamma) \ket{2} - \bra{2}\rho_S(\beta; \Delta)\ket{2}} \\
     &= \abs{\bra{1} \rho_S(\beta; \gamma) \ket{1} - \bra{1}\rho_S(\beta; \Delta)\ket{1}} + \abs{1 - \bra{1} \rho_S(\beta; \gamma) \ket{1} -1 + \bra{1}\rho_S(\beta; \Delta)\ket{1}} \\
     &= 2 \abs{\bra{1} \rho_S(\beta; \gamma) \ket{1} - \bra{1}\rho_S(\beta; \Delta)\ket{1}}. \label{eq:single_qubit_int_1}
 \end{align}
    Now we expand $\bra{1} \rho_S(\beta; \gamma) \ket{1}$ about $\gamma = \Delta$
    \begin{align}
    \bra{1} \rho_S(\beta; \gamma) \ket{1} = \frac{1}{1 + e^{-\beta \gamma}} &= \frac{1}{1 + e^{-\beta \Delta}} + (\gamma - \Delta) \beta \frac{1}{1 + e^{-\beta \gamma_{\star}}} \frac{e^{-\beta \gamma_{\star}}}{1 + e^{-\beta \gamma_{\star}}} \\
    &= \bra{1} \rho_S(\beta; \Delta) \ket{1} + (\gamma - \Delta) \beta \frac{1}{1 + e^{-\beta \gamma_{\star}}} \frac{e^{-\beta \gamma_{\star}}}{1 + e^{-\beta \gamma_{\star}}},
    \end{align}
    where $\gamma_{\star}$ denotes the special value of $\gamma$ that is guaranteed to make the above equation hold by Taylor's Remainder Theorem.
    Since the rightmost factors can be upper bounded by $\frac{1}{1 + e^{-\beta \gamma_{\star}}} \frac{e^{-\beta \gamma_{\star}}}{1 + e^{-\beta \gamma_{\star}}} \le 1$, this can be rearranged and plugged into Eq. \eqref{eq:single_qubit_int_1} to give the upper bound
 \begin{equation}
 \norm{\rho_S(\beta, \gamma) - \rho_S(\beta, \Delta)}_1 \le 2 \beta |\Delta - \gamma|.
 \end{equation}
Since we require this distance to be less than $\epsilon$, we can upper bound $|\Delta - \gamma| \le \sigma$ and require
\begin{equation}
    \sigma \le \frac{\epsilon}{2 \beta}. \label{eq:single_qubit_ineq_1}
\end{equation}

 Now we move on to the second stage of the proof: computing the number of interactions needed to reach the fixed point of the Markov chain. As the Markov transition matrix is only $2 \times 2$ we will compute it explicitly. To do so, we need the matrix elements for $T$, which can be computed using Theorem \ref{thm:second_order_transition} 
\begin{align}
        \vec{e}_1^T T \vec{e}_1 = \bra{1} \mathcal{T}_{\on}(\ketbra{1}{1}) \ket{1} &= - \widetilde{\alpha}^2 \frac{e^{-\beta \gamma}}{1 + e^{-\beta \gamma}} \sinc^2 \left( \frac{(-\Delta + \gamma)t}{2}\right) \label{eq:single_qubit_markov_1}\\
        \vec{e}_2^T T \vec{e}_1 = \bra{2} \mathcal{T}_{\on}(\ketbra{1}{1}) \ket{2} &=  \widetilde{\alpha}^2 \frac{e^{-\beta \gamma}}{1 + e^{-\beta \gamma}} \sinc^2 \left( \frac{(-\Delta + \gamma)t}{2}\right) \\
        \vec{e}_1^T T \vec{e}_2 = \bra{1} \mathcal{T}_{\on}(\ketbra{2}{2}) \ket{1} &=  \widetilde{\alpha}^2 \frac{1}{1 + e^{-\beta \gamma}} \sinc^2 \left( \frac{(\Delta - \gamma)t}{2}\right)\\
        \vec{e}_2^T T \vec{e}_2 = \bra{2} \mathcal{T}_{\on}(\ketbra{2}{2}) \ket{2} &= - \widetilde{\alpha}^2 \frac{1}{1 + e^{-\beta \gamma}} \sinc^2 \left( \frac{(\Delta - \gamma)t}{2}\right). \label{eq:single_qubit_markov_4}
    \end{align}
This gives us the total Markov chain matrix as 
\begin{align}
    \identity + T &= \begin{bmatrix} 1 & 0 \\ 0 & 1 \end{bmatrix} + \widetilde{\alpha}^2 \sinc^2 \left(\frac{(\Delta - \gamma)t}{2} \right) \frac{1}{1 + e^{-\beta \gamma}}\begin{bmatrix} -e^{-\beta \gamma} & 1 \\ e^{-\beta \gamma} & -1\end{bmatrix},\label{eq:markov_matrix_single_qubit_gamma}
\end{align}
 where it can be seen that
 the solution of
 \begin{equation}
     (\identity +T)\vec{p}_{\beta,\gamma} = \vec{p}_{\beta,\gamma},
 \end{equation}
 is
 \begin{equation}
     \vec{p}_{\beta, \gamma} = \frac{1}{1 + e^{-\beta \gamma}} \vec{e}_1 + \frac{e^{-\beta \gamma}}{1 + e^{-\beta \gamma}} \vec{e}_2
 \end{equation}
 In other words, $\vec{p}_{\beta,\gamma}$ is the fixed point of the Markov chain. To show convergence we will need the spectral gap of Eq. \eqref{eq:markov_matrix_single_qubit_gamma}, which is given as $\lambda_{\star} = \widetilde{\alpha}^2 \sinc^2 \left( \frac{(\Delta - \gamma) t}{2} \right)$. Plugging this in to the Markov Relaxation Theorem \ref{thm:markov_chain_bound} we can compute a lower bound on the number of interactions needed 
 \begin{align}
     L &\ge \frac{2}{\widetilde{\alpha}^2 \sinc^2(|\Delta - \gamma|t/2)} \left( 2 \log \frac{1}{\widetilde{\alpha}^2 \sinc^2(|\Delta - \gamma|t/2)} + 4( 1 + \log 2) - \frac{1}{2} + \log \frac{2}{\epsilon} \right) \\
     &\eqqcolon \frac{2}{\widetilde{\alpha}^2 \sinc^2(|\Delta - \gamma| t/2)} J, \label{eq:one_qubit_l_bound_1}
 \end{align}
where $J$ captures the logarithmic factors.  We then choose $L$ to satisfy this lower bound.
 
 Our next goal is to simplify these bounds so that we can propagate them to our final error requirements. We first use Lemma \ref{lem:sinc_poly_approx} to produce a  bound on $\sinc$ whenever $\gamma$ is within our window and $|\Delta - \gamma| \le \sigma$
 \begin{align}
     \sinc^2\left( \frac{|\Delta - \gamma| t}{2} \right) \ge 1 - \frac{\sigma^2 t^2}{2},
 \end{align}
 provided that $t \sigma \le \sqrt{2}$ to make the bound meaningful. Recalling that the dimension of the system is $4$, we can then create a new lower bound for $L$ by plugging this expression for sinc in to to Eq. \eqref{eq:one_qubit_l_bound_1} and use  the definition of $\tilde{\alpha}$ to get the following sufficient condition on $L$
 \begin{equation}
     L \ge \frac{10}{\alpha^2 t^2(1 - \sigma^2 t^2 / 2)} J \label{eq:single_qubit_l_bound_2}
 \end{equation} which is larger than the lower bound in \eqref{eq:one_qubit_l_bound_1}. If we choose $L$ to be twice this lower bound then we will for sure meet the Markov chain error.   
 
 The third stage of the proof utilizes the above bound on $L$ to bound the  off-resonance and remainder terms. The magnitude of the total off-resonance contribution is $L \norm{\mathcal{T}_{\off}}_1 \le \frac{8 }{\Delta^2}$, given by Corollary \ref{cor:t_off_norm}, and the remainder term is $L \norm{R_{\Phi}}_1 \le 32 \sqrt{\frac{2}{\pi}} \alpha t^3$, from Theorem \ref{thm:remainder_bound}. 
 set $\alpha = \frac{1}{t^3(\Delta + \sigma)^2} \le \frac{1}{t^3 \Delta^2}$. This then allows us to make the following inequalities
 \begin{align}
     L(\norm{\mathcal{T}_{\off}}_1 + \norm{R_{\Phi}}_1) &\le \frac{20}{t^2(1 - \sigma^2 t^2 / 2)}J \left( \frac{8}{\Delta^2} + 32 \sqrt{\frac{2}{\pi}} \alpha t^3 \right) \\
    &\le \frac{20}{t^2 \Delta^2 (1 - \sigma^2 t^2 / 2)}J \left( 8 + 32 \sqrt{\frac{2}{\pi}}\right). \label{eq:single_qubit_tmp_2}
 \end{align}
The last step is then to show the above is contained in $\bigotilde{\epsilon}$. As $J$ contains only log factors, it is sufficient to show that there exists a $t$ such that $\frac{1}{t^2(1 - \sigma^2 t^2 / 2)} \le \epsilon$. Rearranging yields a quadratic in $t^2$ that must satisfy the following
\begin{equation}
    \Delta^2 t^2 \left(1 - \frac{\sigma^2 t^2}{2}\right) - \frac{1}{\epsilon} \ge 0. \label{eq:single_qubit_tmp_3}
\end{equation}
The roots of this quadratic are
\begin{align}
    t^2 = \frac{1}{\sigma^2}\left(1 \pm \sqrt{1 - \frac{2 \sigma^2}{\Delta^2 \epsilon}} \right),
\end{align}
meaning if $t$ lies within these two roots then our expression will be satisfied. Our first observation is that in order for any solution to the inequality to exist we require $\sigma \le \Delta \sqrt{\frac{\epsilon}{2}}$, otherwise the roots are complex. As $\sigma \to 0$ we note that the larger root $\frac{1}{\sigma^2}(1 + \sqrt{1 - 2 \sigma^2 / \Delta^2 \epsilon})$ approaches infinity and the smaller root approaches $\frac{1}{\Delta^2 \epsilon}$.

Now that we know the inequality in Eq. \eqref{eq:single_qubit_tmp_3} has valid solutions, provided $\sigma$ is small enough, we can now bound the remainder and off-resonance errors by $\bigotilde{\epsilon}$. 
We now will summarize the results. First our distance from the target state can be upper bounded in three contributions, which is done in Eq. \eqref{eq:single_qubit_three_errors}.
Provided $\sigma \le \frac{\epsilon}{2 \beta}$, then the distance from our Markov chain fixed point to the desired thermal state is less than $\epsilon$. Second, if $L \ge \frac{10 J}{\alpha^2 t^2(1 - \sigma^2 t^2 /2) }$, then the output state of our channel will be less than a trace distance of $\epsilon$ away from the Markov chain fixed point. Lastly, by choosing $\alpha  = \frac{1}{t^3(\Delta + \sigma)^2}$ and $t \in \left[ \frac{1}{\sigma}\sqrt{1 - \sqrt{1 - \frac{2 \sigma^2}{\Delta^2 \epsilon}} }, \frac{1}{\sigma}\sqrt{1 + \sqrt{1 - \frac{2 \sigma^2}{\Delta^2 \epsilon}} } \right]$ we can guarantee that the remainder and off-resonance error terms are of order $\bigotilde{\epsilon}$.  We choose the lower limit for our choice of $t$ because it does not diverge as $\sigma\rightarrow 0$.  Adding these contributions together gives us $\norm{\rho_S(\beta) - \Phi^{\circ L }(\rho)}_1 \in \bigotilde{\epsilon}$. The total simulation time needed can then be bounded
\begin{align}
    Lt &\ge \frac{10 J t}{\alpha^2 t^2 (1 - \sigma^2 t^2 / 2)} \ge \frac{10 J t^5 (\Delta + \sigma)^4}{\sqrt{1 - \frac{2 \sigma^2}{\Delta^2 \epsilon}}} \ge \frac{10 J \sigma \left( 1 - \sqrt{1 - \frac{2 \sigma^2}{\Delta^2 \epsilon}}\right)^{5/2} (\Delta + \sigma)^4}{\sqrt{1 - \frac{2 \sigma^2}{\Delta^2 \epsilon}}}.
\end{align}
If we make the assumption that $\sigma \to 0$ then we can use the fact that the lower bound for $t$ approaches $\frac{1}{\Delta \sqrt{\epsilon}}$, giving
\begin{equation}
    \lim_{\sigma \to 0} Lt \ge \frac{10 J}{\Delta \epsilon^{2.5}}.
\end{equation}
The final result then follows by noting that $J$ is polylogarithmic in $\Delta$ and $\epsilon$ for our choice of $t$ and $\alpha$.
 \end{proof}

\subsection{Harmonic Oscillator} \label{sec:harmonic_oscillator}
Now that we have explored the thermalization channel completely for the single qubit case we turn our attention to a more complicated system: a truncated harmonic oscillator. For this scenario we will assume that the oscillator gap, $\Delta$, is known. This is mostly to simplify proofs of ergodicity and should not be an issue in practice, as evidenced by later theorems that show thermalization without eigenvalue knowledge. The reason behind this proof requirement is that by tuning $\gamma$ to be the spectral gap we can create a ``ladder" transition matrix in which states can move one level up or down. The proof of ergodicity relies on this ladder. Once we remove knowledge of $\Delta$ if $\gamma$ has some probability of being close to $2 \Delta$ this special ladder structure is destroyed. To avoid this annoyance and focus on the special structure granted by the harmonic oscillator we will assume $\gamma = \Delta$.

This system also represents a transition from the single qubit to more general settings discussed later as the guarantees on total simulation time as a function of $\beta$ are similar. For the harmonic oscillator we are only able to bound the spectral gap in the ground state limit as $\beta \to \infty$, meaning that the convergence time for finite $\beta$ has to be characterized in terms of the spectral gap of the Markov chain. For infinite $\beta$ we are able to compute the spectral gap exactly, as the Markov transition matrix is upper triangular. The following theorem introduces this technique in a straightforward setting before it is used later for more complicated transition matrices. 

\begin{theorem}\label{thm:harmonic_oscillator}
    Let $H_S$ denote a truncated harmonic oscillator with $\dim_S$ energy levels that are (up to an irrelevant constant offset) $\lambda_S(i) = i \Delta$ for $1 \le i \le \dim_S$, $\gamma$ be chosen such that $\gamma = \Delta$ be the known energy gap, and let $\rho$ be any input state that commutes with $H_S$. Setting the following parameters for the thermalizing channel $\Phi$
    \begin{equation}
        \alpha = \frac{\epsilon^{1.5} \widetilde{\lambda}_\star(\beta)^{1.5} \Delta}{\dim_S^4}, t = \frac{\dim_S}{\Delta \sqrt{\epsilon \widetilde{\lambda}_\star(\beta)}}, \text{ and } L \in \bigotilde{\frac{\dim_S^2}{\alpha^2 t^2 \widetilde{\lambda}_\star(\beta)}},
    \end{equation}
    where $\widetilde{\lambda}_\star(\beta)$ denotes the spectral gap of the scaled transition matrix $T / \widetilde{\alpha}^2$, is sufficient for thermalization for arbitrary $\beta$ as
    \begin{equation}
        \norm{\rho_S(\beta) - \Phi^{\circ L}(\rho)}_1 \in \bigotilde{\epsilon}.
    \end{equation}
    This gives the total simulation time required as
    \begin{equation}
        L \cdot t \in \bigotilde{\frac{\dim_S^9}{\Delta \epsilon^{2.5} \widetilde{\lambda}_\star(\beta)^{2.5}}}.
    \end{equation}
    In the limit as $\beta \to \infty$ the above settings for $\alpha, L$, and $t$ are valid for preparing the ground state with the spectral gap of the rescaled transition matrix given by
    \begin{equation}
        \lim_{\beta \to \infty} \widetilde{\lambda}_\star(\beta) = 1.
    \end{equation}

\end{theorem}
\begin{proof}
    We first show that the thermal state is the unique fixed point for finite $\beta$. This will be done by computing the nonzero on-resonance transitions and plugging in to Lemma \ref{lem:fixed_point}. As $\gamma = \Delta$, $\Delta_S(i,j) = (i - j) \Delta$, and $\delta_{\min} = \Delta$ we can deduce that the on-resonance transitions will only be nonzero for adjacent states $\ketbra{i}{i}$ and $\ketbra{i \pm 1}{i \pm 1}$. This can be seen explicitly for $i \neq j$ by evaluating the transition elements given by Definition~\ref{def:transition}.
    \begin{align}
        &\bra{j} \TT_{\on}(\ketbra{i}{i}) \ket{j} \\
        &= \widetilde{\alpha}^2 \frac{1}{1 + e^{-\beta \gamma}} \mathbf{I}[|\Delta_S(i,j) - \gamma| \le \delta_{\min}]  \sinc^2\left(\frac{(\Delta_S(i,j) - \gamma)t}{2}\right) \nonumber \\
    &~+ \widetilde{\alpha}^2 \frac{e^{-\beta \gamma}}{1 + e^{-\beta \gamma}} \mathbf{I}[|\Delta_S(i,j) + \gamma| \le \delta_{\min}]  \sinc^2\left(\frac{(\Delta_S(i,j) + \gamma)t}{2}\right) \\
    &= \widetilde{\alpha}^2 q(0) \mathbf{I}[j = i - 1]  \sinc^2\left(\frac{(\Delta_S(i,j) - \gamma)t}{2}\right) + \widetilde{\alpha}^2 q(1) \mathbf{I}[j = i + 1]  \sinc^2\left(\frac{(\Delta_S(i,j) + \gamma)t}{2}\right) \\
    &= \widetilde{\alpha}^2 \left(q(0) \mathbf{I}[j = i - 1] +  q(1) \mathbf{I}[j = i + 1] \right). \label{eq:harmonic_oscillator_t_matrix}
    \end{align}
    We now plug this expression into Eq. \eqref{eq:detailed_balance} of Lemma \ref{lem:fixed_point} and use the fact that $\Delta_S(i,i+1)=\Delta$ for the harmonic oscillator
    \begin{align}
        &\sum_{i \neq j} \frac{e^{-\beta \lambda_S(i)}}{\partfun_S(\beta)} \bra{j} \TT_{\on}(\ketbra{i}{i}) \ket{j} - \frac{e^{-\beta \lambda_S(j)}}{\partfun_S(\beta)} \bra{i} \TT_{\on}(\ketbra{j}{j}) \ket{i} \\
        &= \widetilde{\alpha}^2 \frac{e^{-\beta \lambda_S(j)}}{\partfun_S(\beta)} \sum_{i \neq j} \left[ e^{-\beta \Delta_S(i,j)} \left(q(0) \mathbf{I}[j = i - 1] +  q(1) \mathbf{I}[j = i + 1] \right) - \left(q(0) \mathbf{I}[j = i + 1] +  q(1) \mathbf{I}[j = i - 1] \right) \right] \\
        &= \widetilde{\alpha}^2 \frac{e^{-\beta \lambda_S(j)}}{\partfun_S(\beta)} \left( \left(e^{-\beta \Delta} q(0) - q(1) \right) + \left(e^{+\beta \Delta} q(1) - q(1) \right)\right) \\
        &= \widetilde{\alpha}^2 \frac{e^{-\beta \lambda_S(j)}}{\partfun_S(\beta)} \left( \left(e^{-\beta \Delta} \frac{1}{1 + e^{-\beta \gamma}} - \frac{e^{-\beta \gamma}}{1 + e^{-\beta \gamma}} \right) + \left(e^{+\beta \Delta} \frac{e^{-\beta \gamma}}{1 + e^{-\beta \gamma}} - \frac{1}{1 + e^{-\beta \gamma}} \right)\right) \\
        &= 0,
    \end{align}
    where the final equality comes from setting $\gamma = \Delta$. By Lemma \ref{lem:fixed_point} this is sufficient for $\rho_S(\beta)$ to be a fixed point of $\identity + \TT_{\on}$. 
    
    To show that $\rho_S(\beta)$ is the unique fixed point of the Markov chain it suffices to  show that the walk is ergodic. This means that we need to show that the walk can generate transitions between any two sites, or in other words, the hitting time for any two states $i \neq j$ is nonzero. We prove this by induction on $i - j$ first for $i > j$. For $i = j + 1$ we have
    \begin{align}
        \bra{j} \left( \identity + \TT_{\on} \right)(\ketbra{i}{i}) \ket{j} &= \bra{j} \TT_{\on}(\ketbra{j + 1}{j + 1}) \ket{j} \\
        &= \widetilde{\alpha}^2 q(0),
    \end{align} 
    which is nonzero and therefore the base case holds. Assuming $i = j + n$ holds we show that the hitting time for $i = j + n + 1$ is nonzero. Let $p$ denote the probability of transitioning from $j$ to $j + n$ after $n$ applications of $\identity + \TT_{\on}$. 
    \begin{align}
        &\bra{j } \left(\identity + \TT_{\on} \right)^{\circ n + 1} (\ketbra{j+ n + 1}{j+ n + 1}) \ket{j } \nonumber \\
        &= \sum_{k_1, k_2}\bra{j } \left(\identity + \TT_{\on} \right)^{\circ n} \circ \left(\ketbra{k_1}{k_1}\left(\identity + \TT_{\on} \right) (\ketbra{j+ n + 1}{j+ n + 1}) \ketbra{k_2}{k_2} \right) \ket{j } \\
        &= \sum_{k}\bra{j } \left(\identity + \TT_{\on} \right)^{\circ n} \circ \left(\ketbra{k}{k}\left(\identity + \TT_{\on} \right) (\ketbra{j+ n + 1}{j+ n + 1}) \ketbra{k}{k} \right) \ket{j } \\
        &\ge \bra{j } \left(\identity + \TT_{\on} \right)^{\circ n} \circ \left(\ketbra{j + n}{j + n}\left(\identity + \TT_{\on} \right) (\ketbra{j+ n + 1}{j+ n + 1}) \ketbra{j + n}{j + n} \right) \ket{j } \\
        &= \widetilde{\alpha}^2 q(0) \bra{j } \left(\identity + \TT_{\on} \right)^{\circ n} \left(\ketbra{j + n}{j + n}\right)\ket{j } \\
        &= \widetilde{\alpha}^2 q(0) p,
    \end{align}
    which is greater than 0. We are able to bound the summation over $k \neq j + n$ because $I + T$, the associated Markov matrix for $\identity + \TT_{\on}$, has only positive entries. To prove the case where $i < j$ the same argument holds but now we get factors of $q(1)$ as opposed to $q(0)$, which is nonzero for finite $\beta$. 

    Now that we have shown that the thermal state is the fixed point we would like to bound the total simulation time needed. To do so we first decompose our error into two parts, a Markov chain error and an off-resonance and remainder error
    \begin{equation}
        \norm{\rho_S(\beta) - \Phi^{\circ L}(\rho)}_1 \le \norm{\rho_S(\beta) - (\identity + \TT_{\on})^{\circ L}(\rho)}_1 + L(\norm{\TT_{\off}}_1 + \norm{R_{\Phi}}_1 ). \label{eq:harmonic_oscillator_error_breakdown}
    \end{equation}
    We first bound the number of interactions, $L$, needed for the output of the Markov chain to be $\epsilon$ close to the fixed point and then use this bound on $L$ to upper bound the off-resonance and remainder error. Unfortunately in the finite $\beta$ scenario we are unable to determine the spectral gap of $T$, the entries of which are given in Eq. \eqref{eq:harmonic_oscillator_t_matrix}.  The spectral gap of $T$ is necessary to use Jerison's Markov Relaxation Theorem \ref{thm:markov_chain_bound} which poses a problem for our understanding of the evolution time needed. Instead, we will pull out the overall factor of $\widetilde{\alpha}^2$ and let $\widetilde{\lambda}_\star(\beta)$ denote the spectral gap of $T/\widetilde{\alpha}^2$. This then allows us to use Theorem \ref{thm:markov_chain_bound} but we will have to leave the number of interactions required in terms of $\widetilde{\lambda}_\star(\beta)$.

    Theorem \ref{thm:markov_chain_bound} tells us that requiring
    \begin{align}
        L \ge \frac{\dim_S}{\widetilde{\alpha}^2 \widetilde{\lambda}_\star(\beta)} J \in \bigotilde{\frac{\dim_S^2}{\alpha^2 t^2 \widetilde{\lambda}_\star(\beta)}}
    \end{align}
    is sufficient for the total variational distance between the stationary distribution to be $\epsilon$-small, in other words $\norm{\rho_S(\beta) - (\identity + \TT_{\on})^{\circ L}(\rho)}_1 \in \bigotilde{\epsilon}$. Now we use this expression for $L$ to bound the off-resonance and remainder errors. To do so we first want to asymptotically bound the two contributions, which can be found in Corollary \ref{cor:t_off_norm} and Theorem \ref{thm:remainder_bound}. The sum of the two is
    \begin{equation}
        \norm{\TT_{\off}}_1 + \norm{R_{\Phi}}_1 \le \frac{8 \alpha^2}{\Delta^2} + 16 \sqrt{\frac{\pi}{2}} \dim_S (\alpha t)^3.
    \end{equation}
    By setting $\alpha = \frac{1}{\dim_S \Delta^2 t^3}$ we can simplify the above as
    \begin{equation}
        \norm{\TT_{\off}}_1 + \norm{R_{\Phi}}_1 \le \frac{\alpha^2}{\Delta^2} \left(8 + 16 \sqrt{\frac{\pi}{2}} \right).
    \end{equation}
    Using the sub-additivity property of the trace distance the total error scales as
    \begin{align}
        L (\norm{\TT_{\off}}_1 + \norm{R_{\Phi}}_1) &\le \frac{\dim_S^2}{\alpha^2 t^2 \widetilde{\lambda}_\star(\beta)} J \frac{\alpha^2}{\Delta^2} \left(8 + 16 \sqrt{\frac{\pi}{2}} \right) \\
        &\in \bigotilde{\frac{\dim_S^2}{t^2 \Delta^2 \widetilde{\lambda}_\star(\beta)}}.
    \end{align}
    We can make this $\bigotilde{\epsilon}$ by setting $t = \frac{\dim_S}{\Delta \sqrt{\epsilon \widetilde{\lambda}_\star(\beta)}}$.
    This then gives the following total simulation time as
    \begin{align}
        L\cdot t \in \bigotilde{\frac{\dim_S^2}{\alpha^2 t \widetilde{\lambda}_\star(\beta)} } = \bigotilde{\frac{\dim_S^9}{\epsilon^{2.5} \Delta  \widetilde{\lambda}_\star(\beta)^{3.5}}}.
    \end{align}
    
    Now that we have analyzed the finite $\beta$ regime, we turn to the $\beta \to \infty$ limit. Our proof above for the fixed points only worked for finite $\beta$, but Lemma \ref{lem:fixed_point} tells us that in the $\beta \to \infty$ limit the ground state is a fixed point. We will show that it is the unique fixed point by computing the spectrum of $T$, which will be rather easy to do. Lemma \ref{lem:fixed_point} further tells us that as $\beta \to \infty$ the matrix $T$ is upper triangular, which means we can compute the spectrum if we can compute the diagonal elements. We will do so via Eq. \eqref{eq:harmonic_oscillator_t_matrix} and~\eqref{eq:env_state_def}, which says for $1 < i < \dim_S$
    \begin{align}
        e_i^T T e_i &= \bra{i} \TT_{\on}(\ketbra{i}{i}) \ket{i} \\
        &= -\sum_{j \neq i} \bra{j} \TT_{\on}(\ketbra{i}{i}) \ket{j} \\
        &= - \widetilde{\alpha}^2 \sum_{j \neq i} \left(q(0) \mathbf{I}[j = i - 1] +  q(1) \mathbf{I}[j = i + 1] \right) \\
        &= -\widetilde{\alpha}^2 \left(q(0) + q(1) \right) \\
        &= - \widetilde{\alpha}^2,
    \end{align}
    where the summation is only nonzero for $j = i \pm 1$. For $i = 1$ we note that because $e_1$ is a fixed point we have $e_j^T T e_1 = 0$, so the diagonal entry is 0. The computation for $i = \dim_S$ is similar to the above but yields from~\eqref{eq:env_state_def}
    \begin{equation}
        \lim_{\beta \to \infty} \bra{\dim_S} \TT_{\on}(\ketbra{\dim_S}{\dim_S}) \ket{\dim_S} = -\widetilde{\alpha}^2 \lim_{\beta \to \infty}  q(0) = - \widetilde{\alpha}^2.
    \end{equation}
    This shows us that the zero temperature limit of the transition matrix $T$ is 
\begin{equation}
    \lim_{\beta \to \infty} T = \widetilde{\alpha}^2 \begin{bmatrix}
        0 & 1 &   &\\
        & -1 & 1 &  &\\
        & & -1  & & \\
        & & & \ddots & \\
        & &     &       & 1 \\
        & &  & & -1
    \end{bmatrix}.
\end{equation}
We can compute the spectrum via the characteristic polynomial $\det(\lambda \identity - T)$. This is because $T$ is upper triangular and the determinant we need to compute is
\begin{equation}
    \det \left(\lambda \identity - \lim_{\beta \to \infty} T \right) =  \begin{vmatrix}
        \lambda & -\widetilde{\alpha}^2 &   &\\
        & \lambda + \widetilde{\alpha}^2 & -\widetilde{\alpha}^2 &  &\\
        & & \lambda + \widetilde{\alpha}^2  & & \\
        & & & \ddots & \\
        & &     &       & -\widetilde{\alpha}^2 \\
        & &  & & \lambda + \widetilde{\alpha}^2
    \end{vmatrix}.
\end{equation}
    The roots of the above characteristic polynomial gives the spectrum of $T$ as 0 and $-\widetilde{\alpha}^2$ with multiplicity $\dim_S - 1$. This not only gives the spectral gap of $\widetilde{\alpha}^2$ but further shows that the ground state is the unique fixed point because 0 only has multiplicity 1. This shows that $\lim_{\beta \to \infty} \widetilde{\lambda}_\star(\beta) = 1$.

    We now can use this to repeat the simulation time bound arguments from the finite $\beta$ case. The decomposition in Eq. \eqref{eq:harmonic_oscillator_error_breakdown} is still valid and we can use the Markov Relaxation Theorem \ref{thm:markov_chain_bound} to bound
    \begin{equation}
        L \ge \frac{\dim_S}{\widetilde{\alpha}^2\lim_{\beta \to \infty} \widetilde{\lambda}_\star(\beta)} J \in \widetilde{\Theta}\left({\frac{\dim_S^2}{\alpha^2 t^2}}\right).
    \end{equation}
    The arguments for the off-resonance and remainder error bounds are the exact same and tell us that it suffices to set
    \begin{equation}
        \alpha = \frac{1}{\dim_S \Delta^2 t^3} \text{ and } t = \frac{\dim_S}{\Delta \sqrt{\epsilon}}.
    \end{equation}
    This gives the total simulation time needed as
    \begin{equation}
        L\cdot t \in \bigotilde{\frac{\dim_S^9}{\epsilon^{2.5} \Delta}}.
    \end{equation}
\end{proof}

\subsection{Numerics} \label{sec:specific_numerics}
Now that we have rigorous bounds on each of the parameters $\alpha, t$ and $L$ needed to prepare thermal states of simple systems, we turn to numerics to test these bounds. The first question we explore is how the total simulation time $L \cdot t$ behaves as a function of $\alpha$ and $t$. After, we examine the dependence of the total simulation time on the inverse temperature $\beta$ and we observe a Mpemba-like effect where we find higher temperature states can cool faster than lower temperature ones~\cite{auerbach1995supercooling}. Finally, we demonstrate how our proof techniques could be leading to worse $\epsilon$ scaling than appears numerically necessary. Throughout these experiments we have the same numeric method of starting with the maximally mixed state $\rho_S(0)$ and performing a search on the minimal number of interactions needed for the mean trace distance over all samples to be less than the target $\epsilon$. The number of samples is increased until the variance in the trace distance is less than an order of magnitude below the mean.

In Figure \ref{fig:tot_time_vs_single_time} we explore the total simulation time needed to prepare a thermal state with $\beta = 2.0$ and $\epsilon = 0.05$ for a single qubit system. We plot the total simulation time $L \cdot t$ needed as a function of $t$ for various settings of $\alpha$. We find that increasing both parameters tends to decrease the overall cost until a saturation point is reached, which is at a value of $t$ slightly larger $1/\alpha$. For a fixed value of $\alpha$ this initial decrease in $L \cdot t$ is inverse with $t$, in agreement with our finding of $L \in \bigotilde{t^{-2}}$ in Eq. \eqref{eq:single_qubit_l_bound_2} for $\sigma = 0$. However, this process of decreasing the cost by increasing $t$ can only scale so far and appears to run into a minimum number of interactions $L$ required to thermalize. After this saturation point $L \cdot t$ scales linearly with $t$, indicating that the number of interactions $L$ has reached a minimum. 

Another major take away from Figure \ref{fig:tot_time_vs_single_time} is that it demonstrates that our thermalizing channel is exceptionally robust beyond the weak-coupling expansion in which we can theoretically analyze it. The values of $\alpha t$ used in the far right of the plot completely break our weak-coupling expansion, as we have values of $\widetilde{\alpha}$ that reach up to 500. One interesting phenomenon that we do not have an explanation for is the ``clumping" of various settings of $\alpha$ in the large $t$ limit. As $\alpha t$ dictates the amount of time that the random interaction term $G$ is simulated for, it could be that once a minimum amount of randomness is added via this interaction it is no longer beneficial in causing transitions among system eigenstates. 

\begin{figure} 
    \centering
    \includegraphics[width=0.75\linewidth]{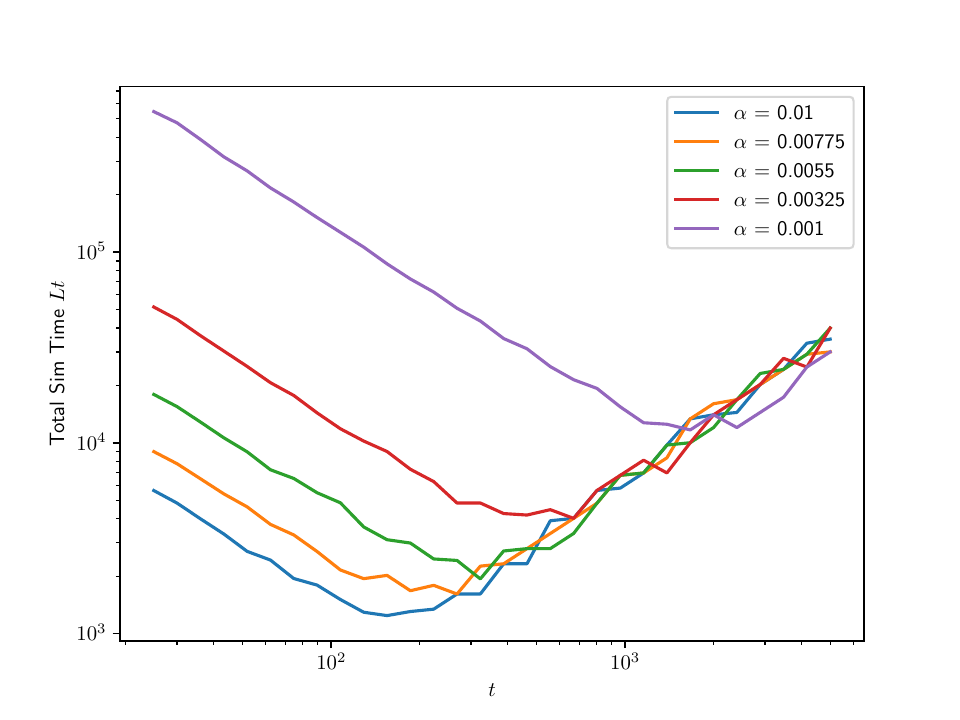}
    \caption{Total simulation time for a single qubit system to reach within trace distance of $0.05$ of the thermal state for $\beta = 2$ as a function of per-interaction simulation time $t$. The slope of the large $t$ asymptote is $\approx$ 1.01.}\label{fig:tot_time_vs_single_time}
\end{figure}

The next task we have is to examine the $\beta$ dependence. For the harmonic oscillator Theorem \ref{thm:harmonic_oscillator} is helpful for giving an idea of the total simulation time for the ground state but we cannot extend it to finite $\beta$ due to the special structure of the transition matrix in the $\beta \to \infty$ limit. Perturbation theory could possibly be used to extend the computation of the spectral gap to the low temperature regime, but even then it would break down for large temperature (small $\beta$). For generic $\beta$ the structure of the harmonic oscillator transition matrix is tridiagonal but it is not quite Toeplitz, as the main diagonals deviate in the upper left and bottom right corners. We could try to pull these deviations into a separate matrix and treat them as perturbations to a fully Toeplitz matrix, which we can then compute the spectrum of. The issue with this approach is that these deviations are on the order of $\widetilde{\alpha}^2 q(0)$ and $\widetilde{\alpha}^2 q(1)$, which are comparable to the eigenvalues of the unperturbed matrix.

In Figure \ref{fig:sho_total_time_vs_beta} we are able to probe the total simulation time and spectral gap of the harmonic oscillator as a function of $\beta$. We reveal a rather surprising Mpemba-like phenomenon where it takes longer for an infinite temperature initial state (the maximally mixed state) to cool to intermediate temperatures than low temperature states. The Mpemba effect \cite{mpemba} is a classical phenomenon related to the time needed to freeze hot water compared to room temperature water with mentions going all the way back to Aristotle. This phenomenon has been extended to quantum thermodynamics and observed in both theory \cite{nickMpemba}, \cite{mpembaExplanation} and in recent experimental research \cite{zhang2025mpembaObservation}. Our observations are not only a further analytic observation, but we are able to provide a proposed mechanism that explains the behavior. It is clear that the distance of our initial state to the target thermal state $\norm{\rho_S(\beta) - \rho_S(\infty)}_1$ increases monotonically with $\beta$ but what is not obvious is that the spectral gap of the underlying Markov chain is \emph{also} increasing. As larger spectral gaps lead to quicker convergences this acts in an opposite way on the total simulation time. The end result is that for small $\beta$ the increase in initial distance is stronger than the increase in the spectral gap and $L \cdot t$ increases. After some amount of $\beta$ these forces flip and the spectral gap effects become stronger than the initial state distance increasing, leading to a reduction in $L \cdot t$. This phenomenon appears to become more pronounced as the dimension of the harmonic oscillator increases, as can be seem in the $\dim_S = 10$ data. Two things remain unclear: the first is what parameters affect the position and height of the peak in total simulation time and the second is if this behavior is present in Hamiltonians with more complicated eigenvalue difference structure than the harmonic oscillator.

\begin{figure}[t]
    \centering
    \begin{subfigure}{0.45\textwidth}
    \includegraphics[width=\textwidth]{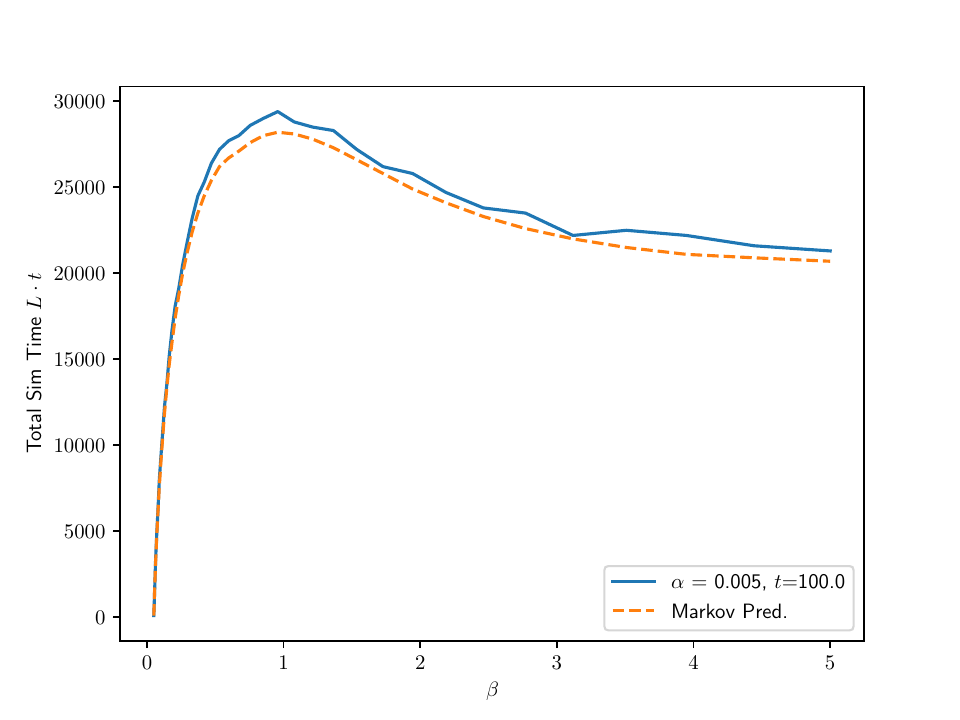}
    \caption{Minimum Interactions vs. $\beta$, $\dim = 4$}
    \label{fig:sho_l_vs_beta_dim_4}
    \end{subfigure}
    \begin{subfigure}{0.45\textwidth}
    \vspace{0.7cm}
    \includegraphics[width=\textwidth]{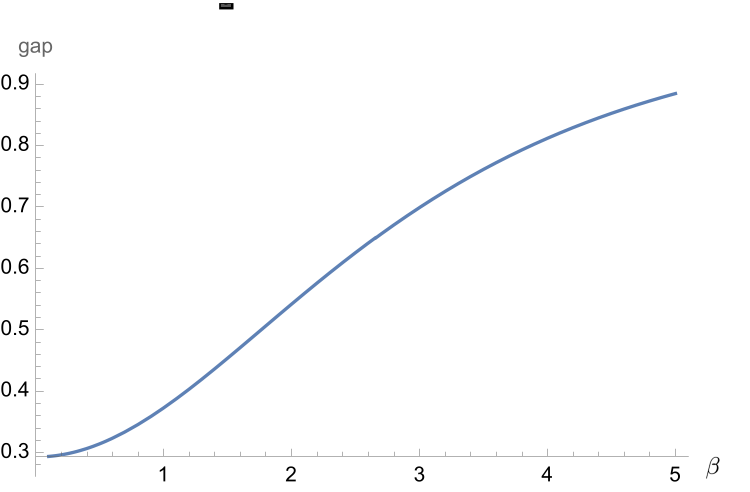}
    \caption{Spectral Gap $\widetilde{\lambda}_\star(\beta)$ vs. $\beta$, $\dim = 4$}
    \label{fig:sho_spectral_gap_vs_beta}
    \end{subfigure}
    \hfill
    \begin{subfigure}{0.5\textwidth}
    \includegraphics[width=\textwidth]{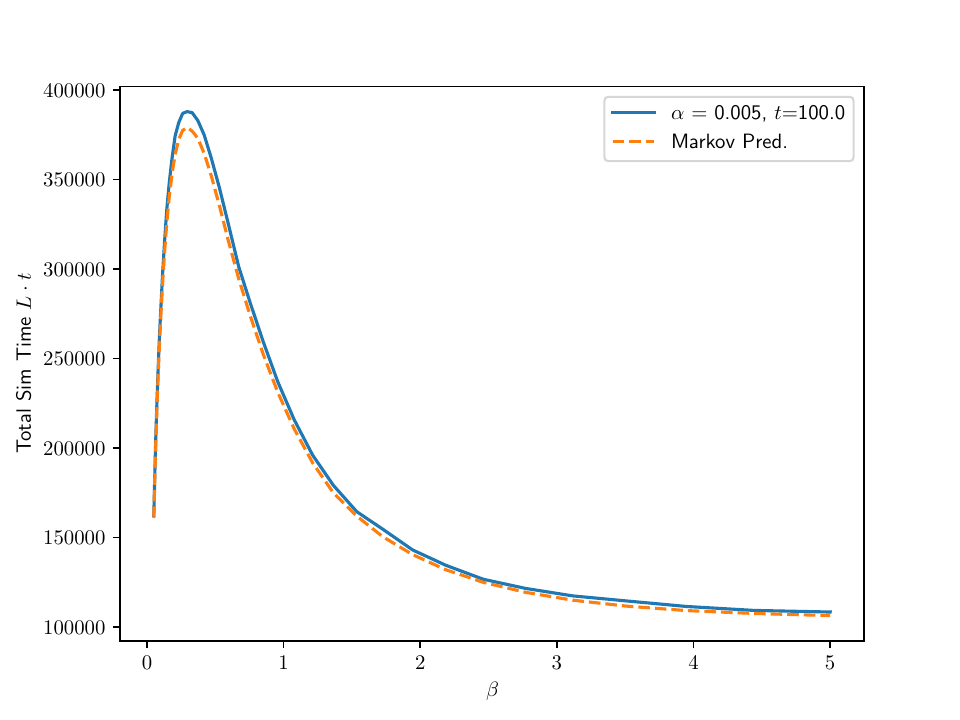}
    \caption{Total simulation time vs. $\beta$, $\dim = 10$}
    \label{fig:sho_l_vs_beta_dim_10}
    \end{subfigure}
    \caption{Demonstration of $\beta$ dependence of the thermalizing channel $\Phi$ for the truncated harmonic oscillator. The environment gap $\gamma$ was tuned to match the system gap $\Delta$ exactly. The minimal number of interactions was found by binary search over values of $L$ that have an average error of less than $\epsilon = 0.05$ with 100 samples.}
    \label{fig:sho_total_time_vs_beta}
\end{figure}

The analytic proofs given in Theorems \ref{thm:single_qubit} and \ref{thm:harmonic_oscillator} are entirely based on our weak-coupling expansion derived in Section \ref{sec:weak_coupling}. The high level picture of this expansion is that we have a remainder error that scales like $\bigo{(\alpha t)^3}$ and an off-resonance error that scales as $\bigo{\alpha^2}$. To balance these two terms we then set $\alpha = \bigo{1/t^3}$. However, as seen in Figure \ref{fig:tot_time_vs_single_time} our thermalization routine appears to be quite robust beyond this weak-coupling expansion, which could lead to significant improvements in runtime. In our derivation for the $\bigo{\alpha}$ and $\bigo{\alpha^2}$ terms we relied on our eigenvalues being I.I.D Gaussian variables, with the first and second order expressions containing factors with the first and second moments respectively of the Gaussian distribution. This would suggest that the third order term in a weak coupling expansion might also be 0, similarly to the first order term. This would lead to a supposed remainder error of $\bigo{\alpha^4 t^4}$, which after balancing with the off-resonance error would give $\alpha = \bigo{1/t^2}$. If the number of interactions then scales like $\bigo{1/\alpha^2 t^2}$, which is consistent with the spectral gap of $\TT_{\on}$ scaling as $\bigo{\alpha^2 t^2}$, then to make the total error of order $\bigo{\epsilon}$ we would require $t \in \bigotilde{1/\epsilon^{0.5}}$ as in Theorems \ref{thm:single_qubit} and \ref{thm:harmonic_oscillator}. This conjecture then leads to a total simulation time of order $\bigo{1/\epsilon^{1.5}}$. 

An even further conjecture would be to keep $\alpha \cdot t$ as a small constant, in this case we are essentially saying that the randomized dynamics $e^{i \alpha t G}$ are beneficial and should not be thought of as some remainder error to be minimized. If the $\alpha t$ constant is small enough then the dynamics will still be approximated by the Markov chain $\TT_{on}$. Our spectral gap will still scale as $\bigo{(\alpha t)^2}$ and $t$ as $\bigo{1/\epsilon^{0.5}}$. This would lead to our total simulation time scaling as $\bigo{1/\epsilon^{0.5}}$. In Figure \ref{fig:epsilon_scaling} we numerically explore these various scalings of $\alpha$ for the harmonic oscillator with $\beta = \dim_S = 4$. Our first remark is that the $\alpha = \bigo{1/t^3}$ scaling as dictated by Theorem \ref{thm:harmonic_oscillator} is numerically supported. 
 Specifically, the theorem suggests that we should observe $O(1/\epsilon^{2.5})$ scaling for $L\cdot t$. 
 Our experiment suggesting $L \cdot t \in \bigo{1/\epsilon^{2.764}}$ which is approximately consistent and deviations from this scaling may arise from the inclusion of data in the fit from outside of the weak coupling limit which is the only regime where we anticipate this scaling. 
 We obtained these exponents via least squares fitting of a power-law fit to $L\cdot t$ and $1/\epsilon$.

\begin{figure}
    \centering
    \includegraphics[width=0.66\linewidth]{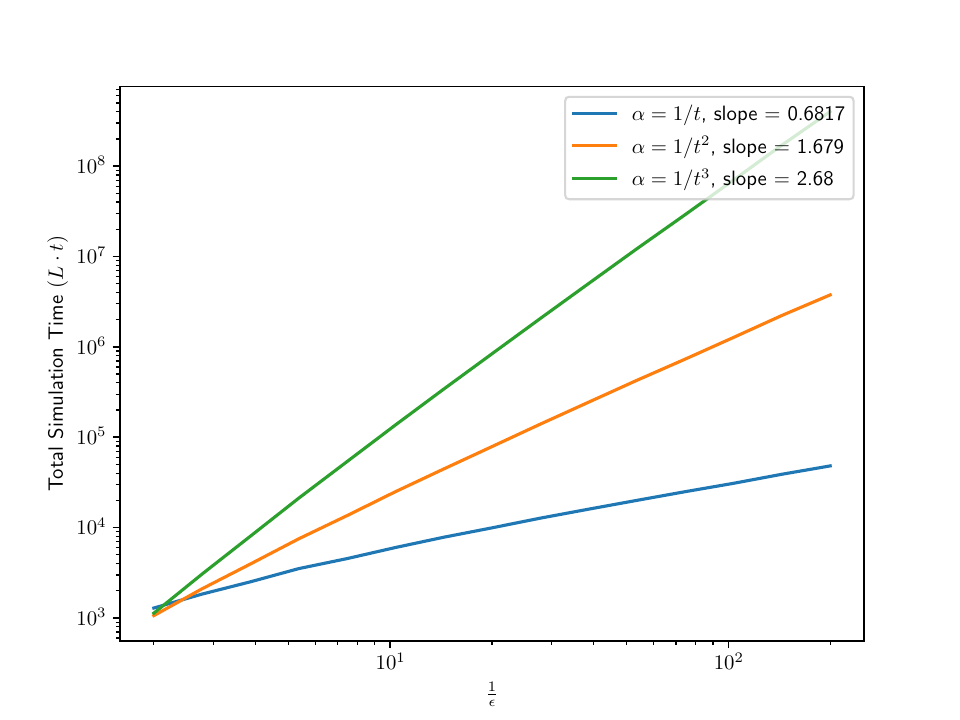}
    \caption{Scaling of $L \cdot t$ to prepare a harmonic oscillator thermal state with $\beta = \dim_S = 4$ with respect to $1/\epsilon$ in a log-log plot. For each line in the plot we scaled $\alpha$ by a constant value to make $\widetilde{\alpha}^2 \approx 0.05$ for the largest value of $\epsilon$. Each of these slopes is consistently larger by 0.18 compared to stated predictions.}
    \label{fig:epsilon_scaling}
\end{figure}

\section{General Systems} \label{sec:general_systems}

We now extend our thermalization techniques to arbitrary Hamiltonians with no degenerate eigenvalues. The first major difficulty that we run into is how to choose our environment gap $\gamma$. If one does not have any knowledge whatsoever about where the eigenvalues of $H_S$ may lie then we are reduced to uniform guessing. In Section \ref{sec:zero_knowledge} we show that even in this scenario the thermal state is an approximate fixed point for finite $\beta$ and the exact fixed point for ground states and we provide a bound on the total simulation time required. However, we show that this generality does come at a cost. If one has complete knowledge of the eigenvalue differences we show in Section \ref{sec:perfect_knowledge} that the total simulation time markedly decreases. Further, with complete knowledge the thermal state is an exact fixed point for all $\beta$. Finally, in Section \ref{sec:general_numerics} we study these impacts on small Hydrogen chain systems and observe the quantitative effects of noise added to $\gamma$. 

The assumption on non-degenerate eigenvalues is required for fairly technical conditions. In the $\beta \to \infty$ limit for our proof of the spectral gap we use the fact that the transition matrix $T$ is upper triangular in Lemma \ref{lem:fixed_point}. This is where the non-degeneracy is required because degenerate eigenvalues always have a non-zero transition amplitude that scales as $\bigo{\widetilde{\alpha}^2}$ without a factor of sinc. This means within the degenerate subspace in the transition matrix there is a uniform block. This makes computing the transition matrix spectrum a little more complicated than necessary, so we avoid this issue by requiring no degeneracies. This restriction could likely be lifted through an intelligent choice of eigenbasis for the degenerate subspace, or through better spectrum calculations of the resulting transition matrix, but we leave such explorations for future work.

\subsection{Zero Knowledge} \label{sec:zero_knowledge}
We now move on to show how our channel performs if one has no knowledge about the eigenvalue differences $\Delta_S(i,j)$ apart from a bound on the maximum value of these differences. This is represented by choosing $\gamma$ uniformly from the interval $[0, 4 \norm{H_S}]$, which technically constitutes an upper bound on the largest $\Delta_S(i,j)$, but estimates of $\norm{H_S}$ are often readily attainable from the specification of the Hamiltonian using the triangle inequality.  We also assume that an input state that commutes with the Hamiltonian can be provided, the maximally mixed state is sufficient as would a random eigenstate yielded by the quantum phase estimation algorithm.

\begin{theorem}[Zero Knowledge Thermal State Prep] \label{thm:zero_knowledge}
    Let $H_S$ be a Hermitian matrix of dimension ${\rm dim}_S$ with no degenerate eigenvalues, $\rho$ any input state that commutes with $H_S$, and $\gamma$ a random variable distributed uniformly in the interval $[0, 4 \norm{H_S}]$ and
    let $\rho_{\rm fix}$ denote the unique fixed point of the transition dynamics $\identity + \EE_\gamma \TT_{\on}^{(\gamma)}$ where $\TT_{\on}^{(\gamma)}$ is the on-resonance transition matrix used above with the dependence on $\gamma$ made explicit. The following statements then hold.
    \begin{enumerate}
\item For finite $\beta$ the thermal state is an approximate fixed point of the thermalizing channel $\EE_\gamma \Phi_\gamma$ with a deviation of
    \begin{equation}
        \norm{\rho_S(\beta) - \EE_\gamma \Phi_\gamma(\rho_S(\beta))}_1 \le \alpha^2 t e^{\beta \delta_{\min}} \norm{H_S}^{-1} \pi + 8 \frac{\alpha^2}{\delta_{\min}} + 16 \sqrt{\frac{\pi}{2}} \dim_S (\alpha t)^3.
    \end{equation}
    \item   The parameter settings for any $\beta\in [0,\infty]$ and error tolerance $\epsilon \in (0,2]$
    \begin{align}
        \alpha = \frac{\delta_{\min}^4 \epsilon^{3} \widetilde{\lambda}_\star(\beta)^{3}}{\dim_S^7 \norm{H_S}^3}, ~t = \frac{\dim_S^2 \norm{H_S}}{\epsilon \widetilde{\lambda}_\star(\beta) \delta_{\min}^2}, \text{ and } L \in \bigotilde{\frac{\dim_S^{14} \norm{H_S}^6}{\epsilon^5 \delta_{\min}^6 \widetilde{\lambda}_\star(\beta)^{6} }}
    \end{align}
    are sufficient to guarantee $\norm{\rho_{\rm fix} - \left(\EE_\gamma \Phi_\gamma \right)^{\circ L}(\rho)}_1 \in \bigotilde{\epsilon}$.
    The total simulation time needed is therefore
    \begin{equation}
        L \cdot t \in \bigotilde{\frac{\dim_S^{16} \norm{H_S}^7}{\delta_{\min}^8 \epsilon^6 \widetilde{\lambda}_\star(\beta)^7}}.
    \end{equation}
   \item    The fixed point is the ground state In the $\beta \to \infty$ limit and the spectral gap, $\widetilde{\lambda}_\star(\beta)$, of the rescaled transition matrix $\EE_\gamma T_\gamma \cdot \left(\frac{2 \norm{H_S} (\dim + 1)}{\alpha^2 t}\right)$ is lower bounded by a constant, giving the two limits
    \begin{equation}
        \lim_{\beta \to \infty} \rho_{\rm fix} = \ketbra{1}{1} \text{ and } \lim_{\beta \to \infty} \widetilde{\lambda}_\star(\beta) = 2 \int_{0}^{-\delta_{\min}t/2} \sinc^2(u) du \ge 2.43.
        \end{equation}

    \end{enumerate}
\end{theorem}
\begin{proof}
    We start by understanding the fixed points of $\identity + \EE_\gamma \TT_{\on}^{(\gamma)}$, conditions for the thermal state being fixed are given in Lemma \ref{lem:fixed_point}. As the condition boils down to a detailed balance like condition, we need to compute the off-diagonal transition elements first. Starting with $i > j$ we have from Definition~\ref{def:transition} that
    \begin{align}
        &\EE_\gamma \bra{j} \TT_{\on}^{(\gamma)}(\ketbra{i}{i})\ket{j} \nonumber \\
        &=  \widetilde{\alpha}^2 \EE_{\gamma} \frac{1}{1 + e^{-\beta \gamma}} \mathbf{I}[|\Delta_S(i,j) - \gamma| \le \delta_{\min}]  \sinc^2\left(\frac{(\Delta_S(i,j) - \gamma)t}{2}\right) \\ \nonumber \\
        &~+ \widetilde{\alpha}^2 \EE_{\gamma} \frac{e^{-\beta \gamma}}{1 + e^{-\beta \gamma}} \mathbf{I}[|\Delta_S(i,j) + \gamma| \le \delta_{\min}]  \sinc^2\left(\frac{(\Delta_S(i,j) + \gamma)t}{2}\right) \\
        &= \widetilde{\alpha}^2 \EE_{\gamma} \frac{1}{1 + e^{-\beta \gamma}} \mathbf{I}[|\Delta_S(i,j) - \gamma| \le \delta_{\min}]  \sinc^2\left(\frac{(\Delta_S(i,j) - \gamma)t}{2}\right) \\
        &= \widetilde{\alpha}^2 \frac{1}{4 \norm{H_S}} \int_{0}^{4 \norm{H_S}} \frac{1}{1 + e^{-\beta \gamma}} \mathbf{I}[|\Delta_S(i,j) - \gamma| \le \delta_{\min}]  \sinc^2\left(\frac{(\Delta_S(i,j) - \gamma)t}{2}\right) d\gamma \\
        &=  \frac{\widetilde{\alpha}^2}{4 \norm{H_S}} \int_{\Delta_S(i,j) - \delta_{\min}}^{\Delta_S(i,j) + \delta_{\min}} \frac{1}{1 + e^{-\beta \gamma}}  \sinc^2\left(\frac{(\Delta_S(i,j) - \gamma)t}{2}\right) d\gamma \\
        &= \frac{\widetilde{\alpha}^2}{2 t \norm{H_S}} \int_{-\delta_{\min} t /2}^{\delta_{\min} t / 2} \frac{1}{1 + e^{-\beta (\Delta_S(i, j) - 2 u /t)}} \sinc^2(u) du. \label{eq:zero_knowledge_transition_1}
    \end{align}
    The exact same calculation holds for $i < j$, which after repeating the steps that led to~\eqref{eq:zero_knowledge_transition_1} we arrive at a similar result  with a slightly different integrand
    \begin{align}
        \EE_\gamma \bra{j} \TT_{\on}^{(\gamma)}(\ketbra{i}{i})\ket{j} &= \frac{\widetilde{\alpha}^2}{2 t \norm{H_S}} \int_{-\delta_{\min} t /2}^{\delta_{\min} t / 2} \frac{e^{-\beta (\Delta_S(j,i) - 2 u /t)}}{1 + e^{-\beta (\Delta_S(j,i) - 2 u /t)}} \sinc^2(u) du,\label{eq:zero_knowledge_transition_2}
    \end{align}
    as we pick up a factor of $q(1)$ as opposed to $q(0)$. Note that we have also shown that $\EE_\gamma T_\gamma$ is ergodic, as there is a nonzero probability for any state $\ketbra{i}{i}$ to transition to any other state $\ketbra{j}{j}$ in one iteration on average over $\gamma$.

    For finite $\beta$ the condition for $\rho_S(\beta)$ being a fixed point is given in Eq. \eqref{eq:detailed_balance}, repeated here as
    \begin{equation}
        \sum_{i \neq j} \frac{e^{-\beta \lambda_S(i)}}{\partfun_S(\beta)} e_j^T \EE_\gamma T_\gamma e_i - \frac{e^{-\beta \lambda_S(j)}}{\partfun_S(\beta)}  e_i^T \EE_\gamma T_\gamma e_j = 0,
    \end{equation}
    for all $j$. We can plug in our calculation for the transition coefficients for summands with $i > j$ first
    \begin{align}
        &\frac{e^{-\beta \lambda_S(i)}}{\partfun_S(\beta)} e_j^T \EE_\gamma T_\gamma e_i - \frac{e^{-\beta \lambda_S(j)}}{\partfun_S(\beta)}  e_i^T \EE_\gamma T_\gamma e_j \nonumber \\ 
        &= \frac{e^{-\beta \lambda_S(j)}}{\partfun_S(\beta)} \left( e^{-\beta \Delta_S(i,j)} \EE_\gamma \bra{j} \TT_{\on}^{(\gamma)}(\ketbra{i}{i})\ket{j} - \bra{i} \TT_{\on}^{(\gamma)}(\ketbra{j}{j})\ket{i} \right) \\
        &= \frac{e^{-\beta \lambda_S(j)}}{\partfun_S(\beta)} \frac{\widetilde{\alpha}^2}{2 t \norm{H_S}} e^{-\beta \Delta_S(i,j)} \int_{-\delta_{\min} t /2 }^{\delta_{\min} t/ 2} \frac{1 - e^{\beta 2 u / t}}{1 + e^{-\beta(\Delta_S(i,j) - 2u/t)}} \sinc^2(u) du \\
        &= \frac{e^{-\beta \lambda_S(i)}}{\partfun_S(\beta)} \frac{\widetilde{\alpha}^2}{2 t \norm{H_S}} \int_{-\delta_{\min} t /2 }^{\delta_{\min} t/ 2} \frac{1 - e^{\beta 2 u / t}}{1 + e^{-\beta(\Delta_S(i,j) - 2u/t)}} \sinc^2(u) du. \label{eq:zero_knowledge_tmp_1}
    \end{align}
    For $i < j$ we have the very similar
    \begin{align}
        &\frac{e^{-\beta \lambda_S(i)}}{\partfun_S(\beta)} e_j^T \EE_\gamma T_\gamma e_i - \frac{e^{-\beta \lambda_S(j)}}{\partfun_S(\beta)}  e_i^T \EE_\gamma T_\gamma e_j \nonumber \\ 
        &= \frac{e^{-\beta \lambda_S(j)}}{\partfun_S(\beta)} \frac{\widetilde{\alpha}^2}{2 t \norm{H_S}} \int_{-\delta_{\min} t /2 }^{\delta_{\min} t/ 2} \frac{ e^{\beta 2 u / t} - 1}{1 + e^{-\beta(\Delta_S(j, i) - 2u/t)}} \sinc^2(u) du. \label{eq:zero_knowledge_tmp_2}
    \end{align}
    Unfortunately these integrals are not 0, which can be verified numerically, and it is unclear how to make the summation over $i \neq j$ equal to 0. 
    
    Our work around this is that instead of showing that the thermal state is exactly the fixed point we can use these results to show that it is an approximate fixed point. There are a few ways we could proceed. The first way could be to compute a Taylor series for the integrand and isolate the limits in which the remainder goes to 0. Unfortunately due to the $\sinc^2(u) = \sin(u)^2 / u^2$ term this means that the overall scaling will go like $1/t$, making the total expression independent of $t$. Instead the route we will take will be to upper bound the norm $\norm{\vec{p}_{\beta} - \EE_\gamma (I + T_\gamma)\vec{p}_\beta}_1 = \norm{\EE_\gamma T_\gamma \vec{p}_\beta}_1$, as this norm is only 0 if $\vec{p}_\beta$ is a fixed point. We reduce this to computations we have already performed as
    \begin{align}
        \norm{\EE_\gamma T_\gamma \vec{p}_\beta}_1 &= \sum_j \abs{e_j^T \EE_\gamma T_\gamma \vec{p}_\beta } \\
        &= \sum_j \abs{\sum_{i} \frac{e^{-\beta \lambda_S(i)}}{\partfun_S(\beta)} e_j^T \EE_\gamma T_\gamma e_i } \\
        &= \sum_j \abs{\sum_{i \neq j} \frac{e^{-\beta \lambda_S(i)}}{\partfun_S(\beta)} e_j^T \EE_\gamma T_\gamma e_i - \frac{e^{-\beta \lambda_S(j)}}{\partfun_S(\beta)} e_i^T \EE_\gamma T_\gamma e_j}.
    \end{align}
    This is essentially the derivation for the fixed point conditions described in Lemma \ref{lem:fixed_point}. We now plug in Eqs. \eqref{eq:zero_knowledge_tmp_1} and \eqref{eq:zero_knowledge_tmp_2} into the above and upper bound the integral as
    \begin{align}
        &\sum_j \abs{\sum_{i \neq j} \frac{e^{-\beta \lambda_S(i)}}{\partfun_S(\beta)} e_j^T \EE_\gamma T_\gamma e_i - \frac{e^{-\beta \lambda_S(j)}}{\partfun_S(\beta)} e_i^T \EE_\gamma T_\gamma e_j} \nonumber \\
        &\le \sum_j \abs{\sum_{i < j} \frac{e^{-\beta \lambda_S(i)}}{\partfun_S(\beta)} e_j^T \EE_\gamma T_\gamma e_i - \frac{e^{-\beta \lambda_S(j)}}{\partfun_S(\beta)} e_i^T \EE_\gamma T_\gamma e_j} + \sum_j  \abs{\sum_{i > j} \frac{e^{-\beta \lambda_S(i)}}{\partfun_S(\beta)} e_j^T \EE_\gamma T_\gamma e_i - \frac{e^{-\beta \lambda_S(j)}}{\partfun_S(\beta)} e_i^T \EE_\gamma T_\gamma e_j} \\
        &= \frac{\widetilde{\alpha}^2}{2 t \norm{H_S}} \sum_j \abs{\sum_{i < j} \frac{e^{-\beta \lambda_S(j)}}{\partfun_S(\beta)} \int_{-\delta_{\min} t /2 }^{\delta_{\min} t/ 2} \frac{ e^{\beta 2 u / t} - 1}{1 + e^{-\beta(\Delta_S(j, i) - 2u/t)}} \sinc^2(u) du} \nonumber \\
        &+ \frac{\widetilde{\alpha}^2}{2 t \norm{H_S}} \sum_j  \abs{\sum_{i > j} \frac{e^{-\beta \lambda_S(i)}}{\partfun_S(\beta)} \int_{-\delta_{\min} t /2 }^{\delta_{\min} t/ 2} \frac{1 - e^{\beta 2 u / t}}{1 + e^{-\beta(\Delta_S(i,j) - 2u/t)}} \sinc^2(u) du} \\
        &\le \frac{\widetilde{\alpha}^2}{2 t \norm{H_S}} \sum_j \sum_{i < j} \frac{e^{-\beta \lambda_S(j)}}{\partfun_S(\beta)} \int_{-\delta_{\min} t /2 }^{\delta_{\min} t/ 2}\abs{ \frac{ e^{\beta 2 u / t} - 1}{1 + e^{-\beta(\Delta_S(j, i) - 2u/t)}}} \sinc^2(u) du \nonumber \\
        &+ \frac{\widetilde{\alpha}^2}{2 t \norm{H_S}} \sum_j  \sum_{i > j} \frac{e^{-\beta \lambda_S(i)}}{\partfun_S(\beta)} \int_{-\delta_{\min} t /2 }^{\delta_{\min} t/ 2} \abs{ \frac{1 - e^{\beta 2 u / t}}{1 + e^{-\beta(\Delta_S(i,j) - 2u/t)}} } \sinc^2(u) du \\
        &\le \frac{\widetilde{\alpha}^2}{2 t \norm{H_S}} e^{\beta \delta_{\min}} \int_{-\delta_{\min}t/2}^{\delta_{\min}t /2} \sinc^2(u) du \left(\sum_j \frac{e^{-\beta \lambda_S(j)}}{\partfun_S(\beta)} \sum_{i < j} 1 + \sum_j \sum_{i > j} \frac{e^{-\beta \lambda_S(i)}}{\partfun_S(\beta)}  \right) \\
        &\le \frac{\widetilde{\alpha}^2 \dim_S}{t \norm{H_S}} e^{\beta \delta_{\min}} \pi \\
        &\le \alpha^2 t e^{\beta \delta_{\min}} \norm{H_S}^{-1} \pi.
    \end{align}
    For this we can have $\alpha t$, which represents the total simulation time multiplied by the strength of the random interaction $G$, be constant and still take $\alpha \to 0$ to achieve arbitrarily small error.

    Now we turn to bounding the total simulation time. We will let $\rho_{\rm fix}$ denote the fixed point of the dynamics. As before, we break the error into two pieces
    \begin{equation}
        \norm{\rho_{\rm fix} - \left(\EE_\gamma \Phi_\gamma \right)^{\circ L}}_1 \le \norm{\rho_{\rm fix} - \left(\EE_\gamma \identity + \TT_{\on}^{(\gamma)}\right)^{\circ L} (\rho)}_1 + L(\norm{\TT_{\off}}_1 + \norm{R_{\Phi}}_1).
    \end{equation}
    Let $\widetilde{\lambda}_\star(\beta)$ denote the spectral gap for the rescaled transition matrix $\EE_\gamma T_\gamma \cdot \left(\frac{2 \norm{H_S} (\dim + 1)}{\alpha^2 t}\right)$, as this is the dimensionful prefactor in front of the transitions derived in Eqs. \eqref{eq:zero_knowledge_transition_1} and \eqref{eq:zero_knowledge_transition_2}. Jerison's Markov Relaxation Theorem \ref{thm:markov_chain_bound} tells us that taking $L$ to satisfy
    \begin{equation}
        L \ge \frac{\dim_S}{\lambda_\star} J \in \bigotilde{\frac{\dim_S^2 \norm{H_S}}{\alpha^2 t \widetilde{\lambda}_\star(\beta)}}\label{eq:Lbd}
    \end{equation}
    is sufficient to guarantee $\norm{\rho_{\rm fix} - \left(\EE_\gamma \identity + \TT_{\on}^{(\gamma)}\right)^{\circ L} (\rho)}_1 \in \bigotilde{\epsilon}$. Now we balance the off-resonance and remainder errors
    \begin{equation}
        \norm{\TT_{\off}}_1 + \norm{R_{\Phi}}_1 \le \frac{8 \alpha^2}{\delta_{\min}^2} + 16 \sqrt{\frac{\pi}{2}} \dim_S (\alpha t)^3 = \frac{\alpha^2}{\delta_{\min}^2} \left( 8 + 16 \sqrt{\frac{\pi}{2}} \dim_S \alpha \delta_{\min}^2 t^3 \right),
    \end{equation}
    and we see setting $\alpha = \frac{1}{\dim_S \delta_{\min}^2 t^3}$ makes the parenthesis a constant.
    To bound the total off-resonance and remainder error we take the product
    \begin{equation}
        L(\norm{\TT_{\off}}_1 + \norm{R_{\Phi}}_1) \in \bigotilde{\frac{\dim_S^2 \norm{H_S}}{\alpha^2 t \widetilde{\lambda}_\star(\beta)} \frac{\alpha^2}{\delta_{\min}^2}} = \bigotilde{\frac{\dim_S^2 \norm{H_S}}{ t \delta_{\min}^2 \widetilde{\lambda}_\star(\beta)} }.
    \end{equation}
    Observe that setting 
    \begin{equation}
        t = \frac{\dim_S^2 \norm{H_S}}{\epsilon \delta_{\min}^2 \widetilde{\lambda}_\star(\beta)}
    \end{equation} 
    is sufficient to make the above product $L(\norm{\TT_{\off}}_1 + \norm{R_{\Phi}}_1) \in\bigotilde{\epsilon}$. 

    We now turn to the $\beta \to \infty$ limit. For this we note that Lemma \ref{lem:fixed_point} guarantees that the ground state is a fixed point and that $\EE_\gamma T_\gamma$ is upper triangular. We will show the ground state is unique by computing the spectrum of $\EE_\gamma T_\gamma$. For this we take the $\beta \to \infty$ limit of the transitions in Eqs. \eqref{eq:zero_knowledge_transition_1} and \eqref{eq:zero_knowledge_transition_2}, which will give us the diagonal elements and then the spectrum. Starting with $i > j$ given in Eq. \eqref{eq:zero_knowledge_transition_1} we get
    \begin{align}
        \lim_{\beta \to \infty} \EE_\gamma \bra{j} \TT_{\on}^{(\gamma)}(\ketbra{i}{i})\ket{j} &= \frac{\widetilde{\alpha}^2}{2 t \norm{H_S}} \int_{-\delta_{\min}t/2}^{\delta_{\min}t/2} \sinc^2(u) du
    \end{align}
    and for $i < j$ from Eq. \eqref{eq:zero_knowledge_transition_2} we have
    \begin{equation}
        \lim_{\beta \to \infty} \EE_\gamma \bra{j} \TT_{\on}^{(\gamma)}(\ketbra{i}{i})\ket{j} = 0.
    \end{equation}
    We denote the $\sinc$ integration above as
    \begin{equation}
        I_{\sinc}(t) \coloneqq \int_{-\delta_{\min}t/2}^{\delta_{\min} t/2} \sinc^2(u) du,
    \end{equation}
    and we will show later that this is constant for $\dim_S \ge 3$. Now these transitions allow us to compute the diagonal elements
    \begin{align}
        \lim_{\beta \to \infty} \EE_\gamma \bra{i} \TT_{\on}^{(\gamma)}(\ketbra{i}{i}) \ket{i} &= - \sum_{j \neq i} \lim_{\beta \to \infty} \EE_\gamma \bra{j} \TT_{\on}^{(\gamma)}(\ketbra{i}{i}) \ket{j} \\
        &= - \sum_{j < i} \lim_{\beta \to \infty} \EE_\gamma \bra{j} \TT_{\on}^{(\gamma)}(\ketbra{i}{i}) \ket{j} \\
        &= - \frac{\widetilde{\alpha}^2}{2 t \norm{H_S}} (i - 1) I_{\sinc}(t).
    \end{align}
    This gives a spectrum for $\EE_\gamma T_\gamma$ as 0 and $- \frac{\widetilde{\alpha}^2}{2 t \norm{H_S}} (i - 1) I_{\sinc}(t)$ for $i > 1$. This shows the ground state is the unique fixed point as 0 has multiplicity 1 in the spectrum. Further the spectral gap of the rescaled transition matrix $\widetilde{\lambda}_\star(\beta)$ is then given by
    \begin{equation}
        \lim_{\beta \to \infty} \widetilde{\lambda}_\star(\beta) = I_{\sinc}(t).
    \end{equation}
    We can repeat the analysis for finding suitable values for $\alpha, t, $ and $L$ to guarantee thermalization and we find that 
    \begin{equation}
        \alpha = \frac{1}{\dim_S \delta_{\min}^2 t^3}, ~ t = \frac{4 \dim_S^2 \norm{H_S}}{\epsilon \delta_{\min}^2}, \text{ and } L \in \bigotilde{\frac{\dim_S^2 \norm{H_S}}{\alpha^2 t}}
    \end{equation}
    are sufficient to guarantee $\norm{\ketbra{1}{1} - \left( \EE_\gamma \Phi_\gamma\right)^{\circ L} (\rho)}_1 \in \bigotilde{\epsilon}$.  Substituting this into~\eqref{eq:Lbd} yields
    \begin{equation}
        L\cdot t \in \bigotilde{\frac{\dim_S^{16} \norm{H_S}^7}{\delta_{\min}^8 \epsilon^6 \widetilde{\lambda}_\star(\beta)^7}}
    \end{equation}
    as stated in the second claim in the theorem.

    Our final task is to justify the third claim of the theorem, which involves showing that $I_{\sinc}(t)$ as constant is valid. Using the choice of $t$ directly above
    \begin{align}
        I_{\sinc}(t) = \int_{-\delta_{\min} t /2}^{\delta_{\min} t /2}  \sinc^2(u) du = 2 \int_{0}^{ \frac{\dim_S^2 4 \norm{H_S}} {\epsilon \delta_{\min}}} \sinc^2(u) du. \label{eq:zero_knowledge_sinc_integral}
    \end{align}
    Now we note that this integral is monotonic with respect to the upper limit of integration with a final value of $\lim_{t \to \infty} I_{\sinc}(t) = \pi$. We note that we can capture a significant amount of this integral by just requiring the upper limit to be greater than the first zero of sinc located at $\frac{\pi}{2}$, which is true if $\epsilon \le \frac{4 \dim_S^2 \norm{H_S}}{\pi \delta_{\min}}$. This value can be computed as $2 \int_0^{\pi / 2} \sinc^2(u) du \ge 2.43$. This can be guaranteed by noting that $\epsilon$ can be at most 2, so the upper limit in Eq. \eqref{eq:zero_knowledge_sinc_integral} is satisfied if
    \begin{align}
    \epsilon \le 2 \le \frac{3^2}{\pi} \le \frac{\dim_S^2}{\pi} \le \frac{\dim_S^2 4 \norm{H_S}}{\pi \delta_{\min}},
\end{align}
as $\delta_{\min} \le 4 \norm{H_S}$. This shows that for our choice of $t$ then $|I_{\sinc}(t) - \pi| \le 0.71$, rendering it asymptotically constant as claimed.
\end{proof}

There are a few points that need to be addressed with the above theorem. The first is that our proof of the approximate fixed point utilizes rather poor bounds, resulting in diverging behavior as $\beta \to \infty$. For finite $\beta$ our bounds on the change in the thermal state scales as $e^{\beta \delta_{\min}}$, which diverges as $\beta$ goes to $\infty$, but in this exact same limit we are able to show that the ground state is the \emph{exact} fixed point of the Markov chain. This clear divergence in approximation error is a result of loose bounds and could be a potential avenue for improvement. The second point we would like to address is the rather high asymptotic scaling. This is the byproduct of a few things, the most important of which is the introduction of a $1/t$ in the reduction of the $\sinc$ integral. This causes a downstream effect of increasing the degree of each asymptotic parameter. To improve this one would need some kind of knowledge of the eigenvalues to prevent a uniform integration of each $\sinc$ term. We study the limiting case of this by assuming sample access to the exact eigenvalue differences $\Delta_S(i,j)$ in Section \ref{sec:perfect_knowledge} and obtain much improved scaling. The second source of inflation in our asymptotic scaling could be our weak-coupling approach to studying the channel. As explored numerically in Section \ref{sec:specific_numerics} we find that using different $\alpha$ scalings with respect to $t$ can greatly effect the $\epsilon$ scaling of the total simulation time $L \cdot t$. A higher order analysis of this channel could lead to anytically better guarantees on the thermalization time required, even in this zero knowledge scenario.

\subsection{Perfect Knowledge} \label{sec:perfect_knowledge}

Oftentimes when studying a system some knowledge of the eigenvalue gaps may be present. Our goal in this section is to study the extreme case of this scenario where one has knowledge of the exact eigenvalue differences. This is unlikely to happen with realistic quantum materials but instead serves as an ideal scenario for our channel to benchmark the effects of eigenvalue knowledge. Further, we note for some computational tasks, such as amplitude amplification, the eigenvalues may be explicitly computable and the real task is to find the dominant eigenvectors. A more realistic model for studying the impacts of eigenvalue knowledge on the total simulation time might be to place Gaussians at each of the $\Delta_S(i,j)$ values with some width $\sigma$. This is the model we use for numeric investigations in Section \ref{sec:general_numerics}, but we were unable to compute the total simulation time required analytically. We find that our model of perfect knowledge allows us to show a reduced total simulation time budget, with the ratio of zero knowledge to perfect knowledge scaling as $\bigotilde{\frac{\norm{H_S}^7}{\delta_{\min}^7 \epsilon^{3.5} \widetilde{\lambda}_\star(\beta)^{3.5}}}$, 
     which gives an explicit worst-case simulation time bound for ground state preparation.

\begin{theorem}[Perfect Knowledge Thermal State Prep] \label{thm:perfect_knowledge}
    Let $H_S$ be a Hermitian matrix of dimension ${\rm dim}_S$ with no degenerate eigenvalues, $\rho$ any input state that commutes with $H_S$ , and let $\gamma$ be a random variable with distribution $\prob{\gamma = \Delta(i,j)} = \frac{\eta_\Delta(i,j)}{\binom{\dim_S}{2}}$ where $\eta_{i,j}$ is the number of times a particular eigenvalue difference appears. For any $\beta\in [0,\infty]$ the thermal state can be prepared with controllable error
    \begin{equation}
        \norm{\rho_S(\beta) - \left(\EE_\gamma \Phi_\gamma \right)^{\circ L}(\rho)}_1 \in \bigotilde{\epsilon}
    \end{equation}
     with the following parameter settings
     \begin{align}
         \alpha &= \frac{\delta_{\min} \epsilon^{1.5} \widetilde{\lambda}_\star(\beta)^{1.5}}{\dim_S^7}, t = \frac{\dim_S^2}{\delta_{\min} \epsilon^{0.5} \widetilde{\lambda}_\star(\beta)^{0.5}},\text{ and } L \in \bigotilde{\frac{\dim_S^{14}}{\epsilon^2 \widetilde{\lambda}_\star(\beta)^3}},
     \end{align}
     where $\widetilde{\lambda}_\star(\beta)$ is the spectral gap of the rescaled transition matrix $\EE_\gamma T_\gamma  \cdot \frac{\binom{\dim_S}{2}}{\widetilde{\alpha}^2}$. 
     This gives the total simulation time required as
     \begin{equation}
         L \cdot t \in \bigotilde{\frac{\dim_S^{16}}{\delta_{\min} \epsilon^{2.5} \widetilde{\lambda}_\star(\beta)^{3.5}}}.
     \end{equation}
     All of the above conditions hold in the ground state limit as $\beta \to \infty$ and further we can compute a lower bound on the spectral gap of the rescaled transition matrix as
     \begin{equation}
         \lim_{\beta \to \infty} \widetilde{\lambda}_\star(\beta) = \min_{i > 1} \sum_{j < i} \eta_\Delta(i,j) \ge 1.
     \end{equation}
\end{theorem}
\begin{proof}
This proof structure is structurally similar to the proof of Theorem \ref{thm:zero_knowledge}.
To show that the thermal state is the fixed point we will need to compute transition factors of the form $\EE_\gamma \bra{j}\TT_{\on}^{(\gamma)}(\ketbra{i}{i})\ket{j}$ for use in Lemma \ref{lem:fixed_point}. Using the on-resonance definition in Eq. \eqref{eq:on_resonance} we have for $i > j$
\begin{align}
    &\EE_\gamma \bra{j} \TT_{\on}^{(\gamma)}(\ketbra{i}{i})\ket{j} \nonumber \\
    &=  \widetilde{\alpha}^2 \EE_{\gamma} \frac{1}{1 + e^{-\beta \gamma}} \mathbf{I}[|\Delta_S(i,j) - \gamma| \le \delta_{\min}]  \sinc^2\left(\frac{(\Delta_S(i,j) - \gamma)t}{2}\right) \nonumber \\
    &~+ \widetilde{\alpha}^2 \EE_{\gamma} \frac{e^{-\beta \gamma}}{1 + e^{-\beta \gamma}} \mathbf{I}[|\Delta_S(i,j) + \gamma| \le \delta_{\min}]  \sinc^2\left(\frac{(\Delta_S(i,j) + \gamma)t}{2}\right) \\
    &= \widetilde{\alpha}^2 \sum_{\Delta_S(k,l)} \prob{\gamma = \Delta_S(k,l)} \frac{\mathbf{I}[|\Delta_S(i,j) - \Delta_S(k,l)| \le \delta_{\min}]}{1 + e^{-\beta \Delta_S(k,l)}}   \sinc^2\left(\frac{(\Delta_S(i,j) - \Delta_S(k,l))t}{2}\right) \\
    &= \widetilde{\alpha}^2 \frac{\eta_\Delta(i,j)}{\binom{\dim_S}{2}} \frac{1}{1 + e^{-\beta \Delta_S(i,j)}}.
\end{align}
$i < j$ can be computed similarly as
\begin{equation}
    \EE_\gamma \bra{j} \TT_{\on}^{(\gamma)}(\ketbra{i}{i})\ket{j} = \widetilde{\alpha}^2 \frac{\eta_\Delta(i,j)}{\binom{\dim_S}{2}} \frac{e^{-\beta \Delta_S(k,l)}}{1 + e^{-\beta \Delta_S(k,l)}}.
\end{equation}
This allows us to compute the detailed-balance like condition in Eq. \eqref{eq:detailed_balance} for $i > j$
\begin{align}
    &\frac{e^{-\beta \lambda_S(i)}}{\partfun_S(\beta)} \EE_\gamma \bra{j} \TT_{\on}^{(\gamma)}(\ketbra{i}{i}) \ket{j} - \frac{e^{-\beta \lambda_S(j)}}{\partfun_S(\beta)} \bra{i} \TT_{\on}^{(\gamma)}(\ketbra{j}{j}) \ket{i} \nonumber \\
    &= \frac{e^{-\beta \lambda_S(i)}}{\partfun_S(\beta)} \widetilde{\alpha}^2 \frac{\eta_\Delta(i,j)}{\binom{dim_S}{2}} \frac{1}{1 + e^{-\beta \Delta_S(i,j)}} - \frac{e^{-\beta \lambda_S(j)}}{\partfun_S(\beta)} \widetilde{\alpha}^2 \frac{\eta_\Delta(i,j)}{\binom{dim_S}{2}} \frac{e^{-\beta \Delta_S(i,j)}}{1 + e^{-\beta \Delta_S(i,j)}} \\
    &= \frac{\widetilde{\alpha}^2}{\partfun_S(\beta)} \frac{\eta_\Delta(i,j)}{\binom{\dim_S}{2}} \left(\frac{e^{-\beta \lambda_S(i)}}{1 + e^{-\beta \Delta_S(i,j)}} - e^{-\beta \lambda_S(j)} \frac{e^{-\beta \Delta_S(i,j)}}{1 + e^{-\beta \Delta_S(i,j)} } \right) \\
    &= 0.
\end{align}
For $i < j$ we can repeat the same steps to argue that detailed balance also holds in this case.
\begin{align}
    &\frac{e^{-\beta \lambda_S(i)}}{\partfun_S(\beta)} \EE_\gamma \bra{j} \TT_{\on}^{(\gamma)}(\ketbra{i}{i}) \ket{j} - \frac{e^{-\beta \lambda_S(j)}}{\partfun_S(\beta)} \bra{i} \TT_{\on}^{(\gamma)}(\ketbra{j}{j}) \ket{i} \nonumber \\
    &= \frac{e^{-\beta \lambda_S(i)}}{\partfun_S(\beta)} \widetilde{\alpha}^2 \frac{\eta_\Delta(i,j)}{\binom{dim_S}{2}} \frac{e^{-\beta \Delta_S(j, i)}}{1 + e^{-\beta \Delta_S(j, i)}} - \frac{e^{-\beta \lambda_S(j)}}{\partfun_S(\beta)} \widetilde{\alpha}^2 \frac{\eta_\Delta(i,j)}{\binom{dim_S}{2}} \frac{1}{1 + e^{-\beta \Delta_S(j, i)}} \\
    &= \frac{\widetilde{\alpha}^2}{\partfun_S(\beta)} \frac{\eta_\Delta(i,j)}{\binom{\dim_S}{2}} \left(\frac{e^{-\beta \lambda_S(j)}}{1 + e^{-\beta \Delta_S(j, i)}} - \frac{e^{-\beta \lambda_S(j)}}{1 + e^{-\beta \Delta_S(j, i)} } \right) \\
    &= 0.
\end{align}
This is sufficient to show that the thermal state $\rho_S(\beta)$ is a fixed point via Lemma \ref{lem:fixed_point}. As we have also shown that the probability of transitioning from any state $\ketbra{i}{i}$ to any other state $\ketbra{j}{j}$ is nonzero this gives a nonzero expected hitting time for any pair of states. This implies the Markov chain is ergodic and that $\rho_S(\beta)$ is the \emph{unique} fixed point.

Next we bound the total simulation time required. For reasons similar to the harmonic oscillator in Section \ref{sec:harmonic_oscillator} we are unable to compute the spectral gap of the Markov matrix. We start the analysis in a similar manner by using the decomposition
\begin{equation}
    \norm{\rho_S(\beta) - \left(\EE_\gamma \Phi_\gamma\right)^{\circ L}(\rho)}_1 \le \norm{\rho_S(\beta) - \left(\EE_\gamma \identity + \TT_{\on}^{(\gamma)}\right)^{\circ L}(\rho)}_1 + L(\norm{\TT_{\off}}_1 + \norm{R_{\Phi}}_1 ).
\end{equation}
We bound the Markov error via Theorem \ref{thm:markov_chain_bound}. This theorem guarantees that choosing $L$ to satisfy
\begin{align}
    L \ge \frac{\dim_S \binom{\dim_S}{2}}{\widetilde{\alpha}^2 \widetilde{\lambda}_\star(\beta)} J \in \bigotilde{\frac{\dim_S^4}{\alpha^2 t^2 \widetilde{\lambda}_\star(\beta)}},
\end{align}
where $\widetilde{\lambda}_\star(\beta)$ is the spectral gap of the rescaled transition matrix $\EE_\gamma T_\gamma \cdot \frac{\binom{\dim_S}{2}}{\widetilde{\alpha}^2}$, is sufficient for $\norm{\rho_S(\beta) - \left(\EE_\gamma \identity + \TT_{\on}^{(\gamma)}\right)^{\circ L}(\rho)}_1 \in \bigotilde{\epsilon}$. We now use this to bound the total off-resonance and remainder error after balancing the two contributions asymptotically
\begin{align}
    \norm{\TT_{\off}}_1 + \norm{R_{\Phi}}_1 \le \frac{8\alpha^2}{\delta_{\min}^2} + 16 \sqrt{\frac{\pi}{2}} \dim_S (\alpha t)^3 = \frac{\alpha^2}{\delta_{\min}^2} \left( 8+ 16 \sqrt{\frac{\pi}{2}}  \alpha \dim_S \delta_{\min}^2 t^3\right).
\end{align}
Setting $\alpha = \frac{1}{\dim_S \delta_{\min}^2 t^3}$ is sufficient to make the parenthesis a constant. Lastly to get the total error in $\bigotilde{\epsilon}$ we multiply the above by the $L$ chosen before
\begin{align}
    L (\norm{\TT_{\off}}_1 + \norm{R_{\Phi}}_1) \in \bigotilde{\frac{\dim_S^4}{\alpha^2 t^2 \widetilde{\lambda}_\star(\beta)} \frac{\alpha^2}{\delta_{\min}^2}} = \bigotilde{\frac{\dim_S^4}{ t^2 \delta_{\min}^2 \widetilde{\lambda}_\star(\beta)} }.
\end{align}
Choosing
\begin{equation}
    t = \frac{\dim_S^2}{\delta_{\min} \sqrt{\epsilon \widetilde{\lambda}_\star(\beta)}} 
\end{equation}
is sufficient to guarantee $L (\norm{\TT_{\off}}_1 + \norm{R_{\Phi}}_1) \in \bigotilde{\epsilon}$ and that the total error $\norm{\rho_S(\beta) - \left( \EE_\gamma \Phi_\gamma \right)^{\circ L}(\rho)}_1 \in \bigotilde{\epsilon}$. Combining the above results for $\alpha, L$ and $t$ yields the theorem statement for finite $\beta$. 

We now show how to calculate $\widetilde{\lambda}_\star(\beta)$ in the $\beta \to \infty$ limit. From Lemma \ref{lem:fixed_point} we know that $\EE_\gamma T_\gamma$ will be upper triangular, implying again that we can compute the spectrum if we can compute the diagonal elements of the matrix. Using our computation of the off-diagonal elements from the proof of Theorem~\ref{thm:zero_knowledge} we have for $i > 1$
\begin{align}
    \lim_{\beta \to \infty} \EE_\gamma \bra{i} \TT_{\on}^{(\gamma)}(\ketbra{i}{i}) \ket{i} &= - \lim_{\beta \to \infty} \sum_{j \neq i} \bra{j} \TT_{\on}^{(\gamma)}(\ketbra{i}{i}) \ket{j} \\
    &= - \lim_{\beta \to \infty} \sum_{j < i} \widetilde{\alpha}^2 \frac{\eta_\Delta(i,j)}{\binom{\dim_S}{2}} \frac{1}{1 + e^{-\beta \Delta_S(i,j)}} - \lim_{\beta \to \infty} \sum_{j > i} \widetilde{\alpha}^2 \frac{\eta_\Delta(i,j)}{\binom{\dim_S}{2}} \frac{e^{-\beta \Delta_S(j, i)}}{1 + e^{-\beta \Delta_S(j, i)}} \\
    &= - \frac{\widetilde{\alpha}^2}{\binom{\dim_S}{2}} \sum_{j < i} \eta_\Delta(i,j).
\end{align}
For $i = 1$ as we know the ground state is fixed we have $\lim_{\beta \to \infty} \EE_\gamma \bra{1} \TT_{\on}^{(\gamma)}(\ketbra{1}{1}) \ket{1} = 0$. This gives the spectrum of $\EE_\gamma T_\gamma$ as 0 and $- \frac{\widetilde{\alpha}^2}{\binom{\dim_S}{2}} \sum_{j < i} \eta_\Delta(i,j)$ for all $i > 1$. From this spectrum we can conclude that the ground state is the \emph{unique} fixed point as 0 has multiplicity 1 in the spectrum and further that the spectral gap can be bounded from below as
\begin{align}
    \lim_{\beta \to \infty} \widetilde{\lambda}_\star(\beta) = \min_{i > 1} \sum_{j < i} \eta_\Delta(i,j) \ge 1.
\end{align}
\end{proof}

The above theorem shows that if we sample our transitions strategically rather than randomly then we can achieve much faster convergence to the groundstates in our upper bounds.  Importantly, the scaling of the total simulation time is also independent of the norm of $H_S$ in this case, whereas the time required by the zero knowledge case does.  Unfortunately, the dimensional scaling of ${\rm dim}_S^16$ is prohibitive for all but the smallest dimensional systems.  This scaling is again likely loose because of a number of assumptions that we make above and also a result of our insistence that the channel always operate inside the regime of weak coupling.  In contrast, we will see below that equilibration can be much faster if strong coupling is assumed.  Finally, it is worth noting that although perfect knowledge is assumed, a cooling schedule is not used.  By changing the distribution depending on the temperature of the Gibbs state it is possible that even better scaling may be achievable.

\subsection{Hydrogen Chain Numerics} \label{sec:general_numerics}

The analytic results developed in the previous two sections provide strong guarantees on the correctness of our routine for most quantum systems, however the bounds on the total simulation are fairly high degree polynomials in the parameters of interest. One crucial interpretation of the two different results is that knowledge of the eigenvalue differences of $H_S$ can lead to significantly better simulation time bounds, but this knowledge is not \emph{crucial} for thermalization. Another important takeaway is that we cannot bound the simulation time or number of interactions required for finite $\beta$ as we cannot bound the spectral gap of the expected transition matrix $\EE_\gamma T_\gamma$. The purpose of this section is to investigate these two theoretic takeaways numerically with small Hydrogen chain systems. These systems are some of the smallest chemical systems that still display some real-world chemical behavior, and as a result are typically used in many numeric benchmarks for quantum routines. 

Our first experiment conducted is to study the effects of changing $\alpha$ and $t$ on the trace distance error as a function of $L$. The theory developed in prior sections is very prescriptive; to reach a specific trace distance of $\epsilon$ all of our theorems give a value of $\alpha$, $L$, and $t$ that guarantee a distance of at most $\bigotilde{\epsilon}$ but say nothing about what this convergence looks like. In Figure \ref{fig:h_chain_error} we study the effects of different choices of $\alpha$ and $t$ on this convergence rate. To generate the Hamiltonians used in these experiments we created a small chain of equally spaced hydrogen nuclei with an STO-3G active space for the electrons. Hamiltonian creation was done with OpenFermion \cite{mcclean2020openfermion} and PySCF \cite{pyscf}. Once the Hamiltonians were generated, the distance to the thermal state $\rho_S(\beta)$ for each was tracked over $L = 5000$ interactions. For both Hydrogen 2 and Hydrogen 3 we chose $\beta = 4$ for consistency, this gave a ground state overlap of around 0.56 for Hydrogen 2 and 0.26 for Hydrogen 3. 

\begin{figure}
\centering
    \centering
    \begin{subfigure}{0.49\textwidth}
        \includegraphics[width = \textwidth]{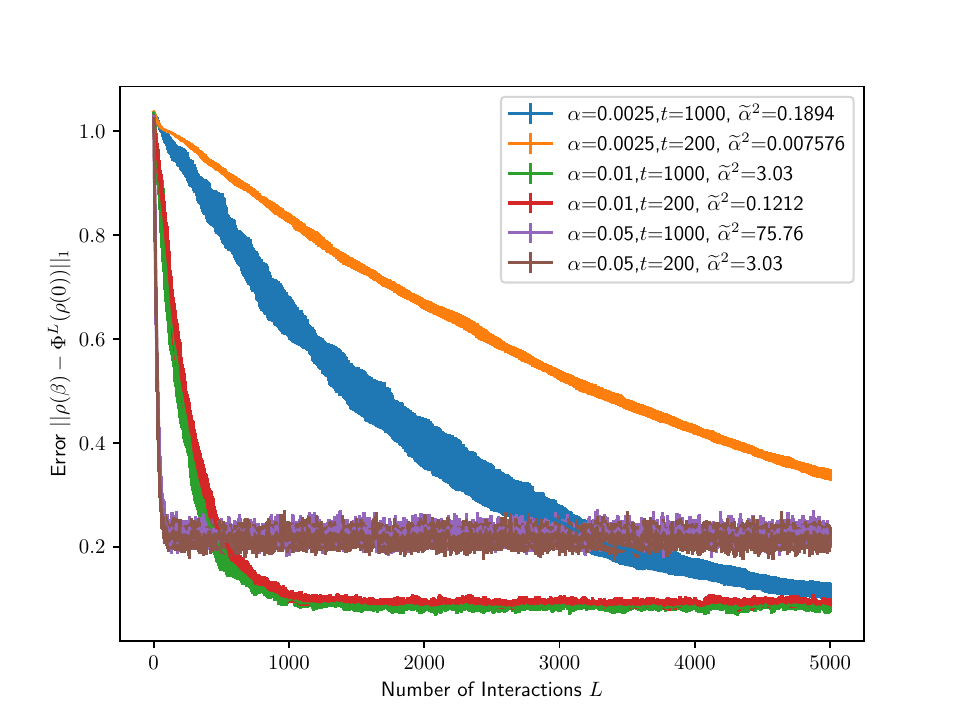}    
        \caption{Hydrogen 2 }\label{fig:h2_error}
    \end{subfigure}    
    \begin{subfigure}{0.49\textwidth}
        \includegraphics[width=\textwidth]{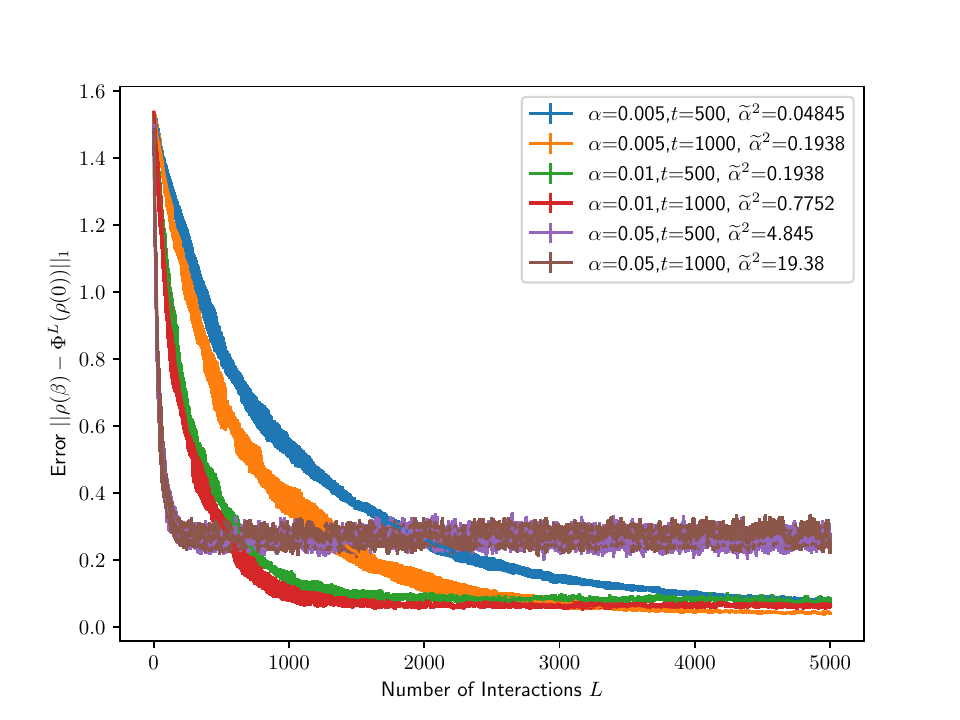}
        \caption{Hydrogen 3 }\label{fig:h3_error}
    \end{subfigure} 
    \caption{These plots show the distance to the target thermal state for Hydrogen 2 and Hydrogen 3 chains as the number of interactions $L$ increases. For both Hydrogen 2 and 3 we set $\beta = 4.0$, which gives a ground state overlap of greater than 0.5 for Hydrogen 2 and 0.25 for Hydrogen 3. $\gamma$ for both \ref{fig:h2_error} and \ref{fig:h3_error} was generated by placing a Gaussian at the average energy $\trace{H_S} / \dim_S$ with a width of $\norm{H_S} / 2$. We note that a variety of $\widetilde{\alpha}^2$ values were chosen to demonstrate the faster convergence, but higher error, of strong coupling.
    }
    \label{fig:h_chain_error}
\end{figure}

There are a few key takeaways from Figure \ref{fig:h_chain_error}. The first is that we observe increasing $\alpha$ and $t$ tend to increase the convergence rate, with $\alpha$ visuallly appearing more important. At higher values of $\alpha$ changes in $t$ appear to make less of an impact on the error. We also observe that our channel is seemingly robust beyond our weak-interaction analysis. For values of $\widetilde{\alpha}^2$, the weak-interaction expansion parameter, we observe that values as high as $\widetilde{\alpha}^2 \approx 3$ can have rapid convergence to fairly low error floors. We note that these coupling values that go beyond weak-interaction seem to lead to faster convergence of the dynamics at the cost of a larger error floor. It remains an open question if dynamic choices of $\alpha$ and $t$ could lead to better performance of the overall routine, one could use very large $\widetilde{\alpha}$ initially to quickly thermalize with large error and then decrease $\widetilde{\alpha}$ to fine-tune the final state.

The second observation we make is on the choice of the environment gaps $\gamma$. For both Hydrogen 2 and 3 we selected $\gamma $ randomly from a Gaussian with mean $\trace{H_S} / \dim_S$ and standard deviation of $\norm{H_S} / 2$. This choice of $\gamma$ is completely heuristic and was intended to have a large overlap with what the typical eigenvalue differences may look like with a large enough deviation to pick up potentially large differences. Although this heuristic works well enough to show convergence, it leads us to question if the error convergence or floors can be improved with better choices of $\gamma$. 

In Figure \ref{fig:h_chain_noise} we demonstrate that better choices of $\gamma$ do in fact reduce the total simulation time needed for thermalization. In Figures \ref{fig:h2_chain_with_noise} and \ref{fig:h3_chain_with_noise} we compute the number of interactions needed at a fixed coupling constant $\alpha$ and a fixed time $t$ as a function of the noise added to our samples for $\gamma$. We generate one sample of $\gamma$ by first computing the eigenvalue spectrum of H2 or H3 exactly, then by choosing two non-equal eigenvalues, and finally sampling a Gaussian centered at the absolute value of the difference. The width of this Gaussian then serves as a proxy for the amount of knowledge one may have about the system's eigenvalues. We plot the total simulation time with respect to this width as it varies from 0 to the spectral width $\max_i \lambda_S(i) - \min_j \lambda_S(j)$. The results align well with our theoretic analysis: having knowledge of the eigenvalues of the system can be used to speed up the thermalization routine but if one does not have any knowledge at all the thermal state can still be prepared. It is an open question if the dependence of the total simulation time $L \cdot t$ on the noise level $\sigma$ can be determined analytically.
 \begin{figure}
     \centering
        \begin{subfigure}{0.49\textwidth}
        \includegraphics[width = \textwidth]{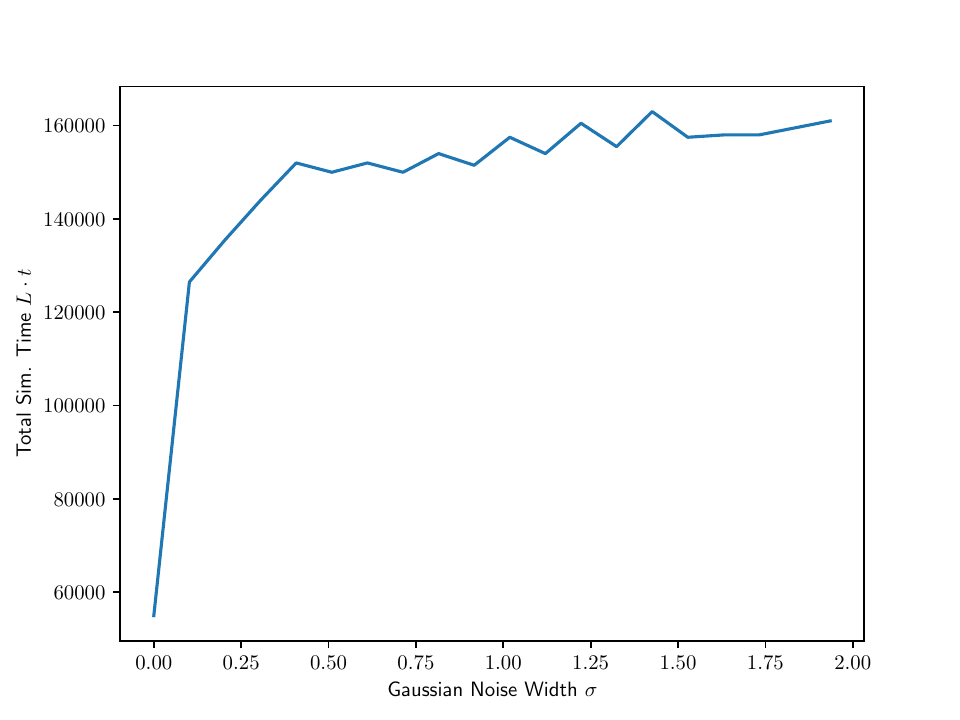}    
        \caption{Hydrogen 2 } \label{fig:h2_chain_with_noise}
    \end{subfigure}    
    \begin{subfigure}{0.49\textwidth}
        \includegraphics[width=\textwidth]{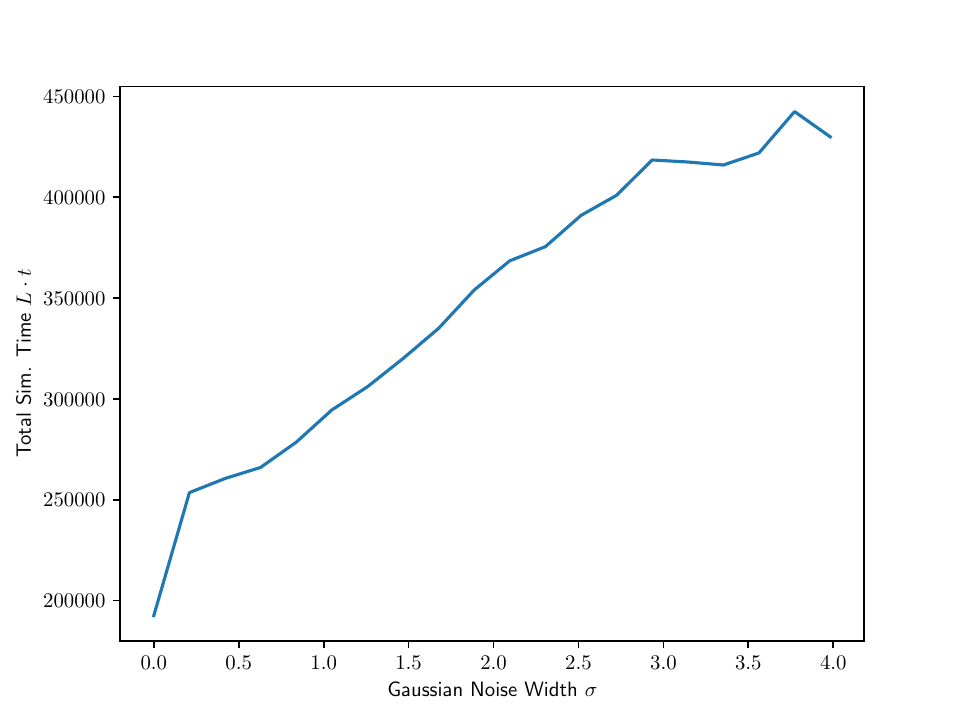}
        \caption{Hydrogen 3 }\label{fig:h3_chain_with_noise}
    \end{subfigure} 
     \caption{ In these plots the amount of total simulation time needed to prepare a $\beta = 2.0$ thermal state with $\alpha = 0.01$ and $t = 500$ is tracked as a function of the noise added to samples of $\gamma$. A sample for $\gamma$ is generated by choosing two non-equal eigenvalues from the system spectrum and adding a Gaussian random variable with standard deviation $\sigma$. For each value of $\sigma$ the resulting state needs to have an average trace distance of less than $0.05$ for 100 samples.}
     \label{fig:h_chain_noise}
 \end{figure}

\section{Conclusion} \label{sec:conclusion}

Thermal state preparation is likely to be a crucial preparation step for the simulation of quantum systems on digital quantum computers. We have presented a thermalization routine for this task that has an optimally minimal number of overhead ancilla qubits and compiles to remarkably simple circuits of time independent Hamiltonian evolution of the unprocessed system Hamiltonian, with no filtering or rejection steps and no Fourier weighted jump operators of Linbladians. Our routine is based on relatively recent classical Monte Carlo techniques, specifically Hamiltonian Monte Carlo \cite{hoffman2011nouturnsampleradaptivelysetting} and the end result bears striking resemblance to the Repeated Interactions framework in open quantum systems \cite{prositto2025equilibrium}. In the Hamiltonian Monte Carlo algorithm thermal states over a position coordinate $q$ is prepared by sampling momentum $p$ from the Boltzmann distribution for Gaussians $e^{-\beta p^2/2m}$ followed by time evolution. Classical Hamiltonian dynamics is enough to couple the position and momentum, leading to the Boltzmann distribution over $q$ with enough time and samples. In the repeated interactions framework a quantum system interacts with many small environments, typically a single photon, that is repeated until the system thermalizes. 

Our work extends these procedures to quantum algorithms. For Hamiltonian Monte Carlo, instead of adding in momentum variables we add in a single ancilla qubit to serve as our extra state space. We do not have the luxury of classical Hamiltonian dynamics that couples these two spaces or registers, so we add in a randomized interaction term to the Hamiltonian. After simulating the time dynamics of this system-ancilla pair and repeating multiple times we are able to thermalize the system to the same $\beta$. On the other hand, the Repeated Interactions framework typically is concerned with thermodynamic limits, such as infinite time or interactions, and specific system-interactions pairs. As our procedure is intended to be used as a subroutine for quantum computers our techniques work for arbitrary, non-degenerate, Hamiltonians and purposefully use randomized interactions as opposed to a fixed interaction model.

One benefit of our thermalization procedure is that it can be compiled all the way down to the circuit level with minimal overhead in complexity. The elements of the channel that need to be compiled are: the ancilla qubit state preparation, the initial state preparation for the system, the time evolution of $H + \alpha G$, and the partial trace. The ancilla state preparation can be done starting from the ground state with a Pauli-$X$ rotation. The initial system state can be any state that commutes with $H_S$, the two leading choices are the maximally mixed state or a Haar random pure state. The maximally mixed state can be prepared using just CNOT gates at a cost of higher qubit count and a Haar random pure state can be prepared with no additional qubit overhead by applying a deep enough random circuit \cite{choi2023preparing}. 

The time evolution of $H + \alpha G$ can be broken down into simulation of $H_S$ and then $\alpha G$ by one application of Trotter. The simulation of $H_S$ can be implemented in two different ways depending on the access model for the system Hamiltonian $H_S$. If $H_S$ is provided as a block-encoding then there exist optimal simulation techniques that add only a single extra ancilla qubit \cite{low2019hamiltonian}. If $H_S$ is provided as a sum of $k$-local Pauli strings or as a sparse matrix then product formula techniques can be used \cite{childs2021theory} at zero extra overhead in ancilla qubits. To simulate the time evolution of $\alpha G$ we can use the following breakdown $e^{i \alpha G t} = e^{i \alpha U_{\haar} D U_{\haar}^\dagger t} = U_{\haar} e^{i \alpha  D  t}U_{\haar}^\dagger$. This unitary can be implemented with a 2-design to approximate the $U_{\haar}$, this is due to the fact that we only expand the channel to second order in $\alpha$. Random Clifford circuits \cite{webb2015clifford} are sufficient for this purpose. To simulate $e^{i \alpha D t}$ random $Z$ rotations and controlled-$Z$ rotations should be sufficient, as we only ever rely on the eigenvalues being pairwise independent in our analysis. In the worst case the number of random gates in the $Z$ basis that would need to be applied would scale with the dimension of the system, adding an overall factor of $\dim_S$ to the preparation. In the product formula case we further remark that $H + \alpha G$ could be simulated in total with a composite technique of using Trotter for $H$ and randomized compilation for $\alpha G$ \cite{hagan2023composite}. 

In classical Hamiltonian Monte Carlo it is well known that sharp gradients in the Hamiltonian require longer simulation time and more samples to converge. Our quantum routine has a much more subtle dependence on the structure of the Hamiltonian. As our single ancilla qubit only has one energy difference $\gamma$, we have to tune this energy difference to allow for energy to be siphoned off from the system into the ancilla. This would present a conundrum, as knowing spectral gaps is as difficult or harder than preparing ground states of arbitrary quantum Hamiltonians, but we are able to prove that our routine is robust to complete ignorance of these differences. We show that this ignorance comes at an asymptotic cost in the amount of resources needed to prepare the thermal state. We numerically verify that knowledge of the eigenvalue differences can be used to speed up the total simulation time, as demonstrated in Figure \ref{fig:h_chain_error}. We posit that this behavior serves as a crucial entry point for heuristics about Hamiltonian spectra into thermal state preparation algorithms. No prior thermal state preparation routines have had such an explicit demonstration of the utility of such knowledge. It was our hope to analytically quantify the speedups gained as a function of the relative entropy between a heuristic guess for the eigenvalue differences and the true spectra, but our numeric evidence will have to suffice until future work can clarify this dependence.

We would like to make a few remarks on potential improvements for the analysis of this channel. As we have demonstrated numerically, our guaranteed analytic values of $\alpha$ and $t$ that lead to thermalization are drastically overestimated. We conjecture that this is due to our truncation of the weak-coupling expansion and in Figure \ref{fig:epsilon_scaling} we demonstrate that taking $\alpha \propto 1/t$ and $t \propto 1/ \sqrt{\epsilon}$ drastically outperforms our analytically derived bounds of $\alpha \propto 1/t^3$, by almost 4 orders of magnitude at $\epsilon \approx 0.005$. It is an open question of how to analyze this channel in the strong-coupling regime, and our numeric results suggest that such an analysis may indicate better performance of our protocol than a weak-coupling expansion can show. It is also an open question of whether dynamically chosen values of $\alpha$ and $t$, such as having strong coupling and low time at the beginning and gradually decreasing $\alpha$ and increasing $t$, can outperform static $\alpha$ and $t$. We also suspect that the Markov relaxation theorem we used greatly overestimates the number of interactions needed. It remains to be seen if better Markov theory is needed or if the convergence time could be characterized based on the overlap of the initial state with the thermal state, which is a property that a few ground state preparation algorithms demonstrate. Another potential avenue for improving the analysis of this channel is whether different randomized interactions or even eigenvector heuristics can be beneficial. For example, in the harmonic oscillator if one has knowledge of the creation and annihilation operators $a^\dagger$ and $a$, could one simply use the interaction $a^\dagger \otimes (X + i Y) + a \otimes (X - i Y)$ instead of involving a randomized $G$ that relies on a Haar average? The last potential improvement is to extend our spectral gap computations using perturbation theory. We are only able to compute the spectral gap $\widetilde{\lambda}_\star(\beta)$ in the limit of $\beta \to \infty$, but it should be possible to compute a perturbation on the order of $1/\beta$. This would give the simulation time needed to prepare low-temperature thermal states as opposed to zero-temperature states.

Lastly, we would like to speculate on possible applications of this routine to other quantum information processing tasks. The first question that arises is if these techniques could be used in the training of quantum Boltzmann machines, which are essentially thermal states. It is an open question if our thermalizing techniques could be used to either train models or to generate output samples from an already trained model. Through the process of demonstrating that this channel prepares the system in the thermal state we have calculated the output of our channel for both the system and the environment registers, and for much larger environments than single qubits. We can turn this protocol on it's head and ask how much information about the system are these ancilla qubits carrying away with them? Preliminary explorations suggest that given knowledge about eigenvalue gaps one can use transition statistics in the ancilla qubits to infer what the inverse temperature $\beta$ is of the system, assuming the system is in a thermal state. Could this thermalizing channel instead be used to develop a Bayesian model to update beliefs about Hamiltonian spectra and system temperatures? This would represent an interaction agnostic model for performing quantum thermometry or spectroscopy, which to the best of our knowledge has not been developed yet. 

\begin{acknowledgments}
We thank Matthew Pocrnic, Alessandro Prositto, and Dvira Segal for useful discussions about repeated interaction models and Ahmet Burak Çatlı and Sophia Simon for helpful discussions on quantum thermodynamics. We thank the University of Toronto Computer Science Department for compute resources. This material is primarily based upon work supported by the U.S. Department of Energy, Office of Science, National Quantum Information Science Research Centers, Co-design Center for Quantum Advantage (C2QA) under contract number DE- SC0012704 (PNNL FWP 76274).  NW also acknowledges support from Google Inc. and Boehringer Ingelheim Inc.
\end{acknowledgments}
\bibliographystyle{unsrt}
\bibliography{bib}

\appendix

\section{Technical Proofs} \label{sec:appendix}
\subsection{Sinc Bounds} \label{sec:appendix_sinc}

\begin{lemma}[Sinc Function Bounds] \label{lem:sinc_poly_approx}
    For $\sinc^2\left( \frac{x t}{2} \right)$ and $\delta_{\min}$ as defined in Eq. \eqref{eq:delta_min_def}, we will make significant use of the following bounds:
    \begin{align}
        |x| \ge \delta_{\min} \implies \sinc^2 \left( \frac{x t}{2} \right) &\le \frac{4}{\delta_{\min}^2 t^2} \label{eq:sinc_upper_bound} \\
        |x| \le \frac{\sqrt{2}}{t} \implies \sinc^2\left(\frac{x t}{2} \right) &\ge 1 - \frac{|x|^2 t^2}{2}. \label{eq:sinc_lower_bound}
\end{align}

\end{lemma}
\begin{proof}
    The first inequality is rather trivial
    \begin{align}
        \sinc^2 \left( \frac{x t}{2} \right) &= \frac{\sin^2 (x t /2)}{(x t / 2)^2} \le \frac{4}{x^2 t^2} \le \frac{4}{\delta_{\min}^2 t^2}.
    \end{align}
    The second involves a Taylor Series for $\sinc^2$, which we compute using the expression of $\sinc$ as $\sinc(x t/ 2) = \frac{\sin xt /2}{xt/2} = \int_0^1 \cos(sxt/2) ds$.  The first two derivatives can then be computed easily
    \begin{align}
        \frac{d \sinc^2(x t /2)}{dx} &= - t \int_0^1 \sin(sx) s ds \int_0^1 \cos(sx) ds \\
        \frac{d^2 \sinc^2(x t /2)}{dx^2} &= -t^2 / 2 \int_0^1 \cos(sx)s^2 ds \int_0^1 \cos(sx) ds + t^2 / 2 \int_0^1 \sin(sx) s ~ds \int_0^1 \sin(sx) s ~ds.
    \end{align}
    We can evaluate each of these derivatives about the origin using continuity of the derivatives along with the limits $\lim_{x \to 0} \cos(sx) = 1$ and $\lim_{x \to 0} \sin(sx) = 0$. We can now compute the mean-value version Taylor series as
    \begin{align}
        \sinc^2 \left(\frac{x t}{2} \right) &= \sinc^2(0) + x \frac{d}{dx} \sinc^2 \left(\frac{x t}{2} \right) \bigg|_{x = 0} + \frac{x^2}{2!} \frac{d^2}{dx^2} \sinc^2 \left(\frac{x t}{2} \right) \bigg|_{x = x_{\star}},
    \end{align}
    where $x_{\star} \in [0,1]$. 
    Plugging in $\sinc^2(0) = 1$ and $\frac{d\sinc^2(x t /2)}{dx}\big|_{x = 0} = 0$ then yields $|\sinc^2(xt/2) - 1| = \frac{|x|^2}{2} \abs{\frac{d^2\sinc^2(x t / 2)}{dx^2}\big|_{x = x_{\star}}}$. We make use of the rather simplistic bound
    \begin{align}
        \abs{\frac{d^2\sinc^2(sxt/2)}{dx^2}\bigg|_{x = x_{\star}} } &\leq t^2 / 2 \abs{\int_0^1 \cos(sx_{\star} t/ 2) s^2 ds \int_0^1 \cos(sx_{\star} t/ 2) ds} + t^2 /2 \abs{\int_0^1 \sin(sx_{\star} t/ 2) s ds \int_0^1 \sin(sx_{\star} t/ 2) s ds} \\
        &\leq t^2 / 2 \int_0^1 \abs{\cos(sx_{\star} t/2)} s^2 ds \int_0^1 \abs{\cos(sx_{\star} t /2 )} ds + t^2 / 2 \parens{\int_0^1 \abs{\sin(sx_{\star} t /2)} |s| ds}^2 \\
        &\leq t^2 / 2 \int_0^1 s^2 ds + t^2 / 2 \parens{\int_0^1 s ds}^2 \\
        &\leq t^2.
    \end{align}
    This yields the final inequality $|\sinc^2(x t /2 ) - 1| \leq \frac{|x|^2 t^2}{2}$ which yields Eq. \eqref{eq:sinc_lower_bound}.
\end{proof}

\subsection{Haar Integral Proofs} \label{sec:haar_integral_appendix}

In this section we present the more technical work needed to state our results in Section \ref{sec:weak_coupling}. Lemmas \ref{lem:two_heisenberg_interactions} and \ref{lem:sandwiched_interaction} are used to compute the effects of the randomized interactions in a form that are usable in the main result of Lemma \ref{lem:big_one}. Lemma \ref{lem:haar_two_moment} can be derived from Appendix C in \cite{brandao2021complexity}.
\begin{restatable}{lemma}{haar_two_moment} \label{lem:haar_two_moment}
    Let $\int (\cdot) dU$ denote the average distributed according to the Haar measure over $\dim$-dimensional unitary matrices $U$. Then for $\ket{i_1},\ket{i_2},\ldots,\ket{k_2}$ drawn from an orthonormal basis
    \begin{align}
        &\int \bra{i_1} U \ket{j_1} \bra{i_2} U \ket{j_2} \bra{k_1} U^\dagger \ket{l_1} ~ \bra{k_2} U^\dagger \ket{l_2} dU \nonumber \\
        &= ~\frac{1}{\dim^2 - 1} \parens{\delta_{i_1, l_1} \delta_{j_1, k_1} \delta_{i_2, l_2} \delta_{j_2, k_2} + \delta_{i_1, l_2} \delta_{j_1, k_2} \delta_{i_2, l_1} \delta_{j_2, k_1}} \nonumber \\
        &\quad - \frac{1}{\dim(\dim^2 - 1)} \parens{\delta_{i_1, l_2} \delta_{j_1, k_1} \delta_{i_2, l_1} \delta_{j_2, k_2} + \delta_{i_1, l_1} \delta_{j_1, k_2} \delta_{i_2, l_2} \delta_{j_2, k_1}}. \label{eq:haar_two_moment_integral}
    \end{align}
\end{restatable}

\begin{lemma} \label{lem:two_heisenberg_interactions}
    Let $G(t)$ denote the Heisenberg evolved random interaction $G(t) = e^{iHt} G e^{-iHt}$ for a total Hamiltonian $H$. After averaging over the interaction measure the product $G(x) G(y)$ can be computed as
    \begin{equation}
        \int G(x) G(y) dG = \frac{1}{\dim + 1} \parens{\sum_{(i,j),(k,l)} e^{i \Delta(i,j|k,l) (x-y)} \ketbra{i,j}{i,j} + \identity}.
    \end{equation}
\end{lemma}
\begin{proof}
The overall structure of this proof is to evaluate the product in the Hamiltonian eigenbasis and split the product into three factors: a phase contribution from the time evolution, a Haar integral from the eigenvalues of the random interaction, and the eigenvalue contribution of the random interaction. Since this involves the use of multiple indices, it will greatly simplify the proof to use a single index over the total Hilbert space $\hilb$ as opposed to two indices over $\hilb_S \otimes \hilb_E$. For example, the index $a$ should be thought of as a pair $(a_s, a_e)$, and functions $\lambda(a)$ should be thought of as $\lambda(a_s, a_e)$. Once the final form of the expression is reached we will substitute in pairs of indices for easier use of the lemma in other places.
    \begin{align}
        \int G(x) G(y) dG &= \int e^{+i H x} U_G D U_G^\dagger e^{-i H x} e^{+i H y} U_G D U_G^\dagger e^{-i H y} dU_G ~dD \\
        &= \int \bigg[\sum_a e^{+i \lambda(a)x}\ketbra{a}{a}  U_G \sum_b D(b)\ketbra{b}{b} U_G^\dagger \nonumber \\
        &\quad \sum_c e^{-i \lambda(c) (x - y)} \ketbra{c}{c} U_G \sum_d D(d)\ketbra{d}{d} U_G^\dagger \sum_e e^{-i \lambda(e) y} \ketbra{e}{e} \bigg] dU_G ~dD\\
        &=\sum_{a,b,c,d,e} \ketbra{a}{e} e^{-i (\lambda(c) - \lambda(a))x} e^{-i (\lambda(e) - \lambda(c))y} \nonumber \\
        &\quad \times \int \bra{a} U_G \ket{b} \bra{c} U_G \ket{d} \bra{b} U_G^{\dagger} \ket{c} \bra{d} U_G^\dagger \ket{e} dU_G \int D(b) D(d) dD \\
        &=  \sum_{a, b, c, d, e} \delta_{bd} \ketbra{a}{e} e^{-i (\lambda(c) - \lambda(a))x} e^{-i (\lambda(e) - \lambda(c))y} \nonumber \\
        &\quad \times \int \bra{a} U_G \ket{b} \bra{c} U_G \ket{d} \bra{b} U_G^{\dagger} \ket{c} \bra{d} U_G^\dagger \ket{e} dU_G. \\
    \end{align}
    Now the summation over $d$ fixes $d=b$ and we use Lemma \ref{lem:haar_two_moment} to compute the Haar integral, which simplifies greatly due to the repeated $b$ index. Plugging the result into the above yields the following
    \begin{align}
        &= \frac{1}{\dim^2 - 1} \sum_{a, b, c, e} \ketbra{a}{e} e^{-i (\lambda(c) - \lambda(a))x} e^{-i (\lambda(e) - \lambda(c))y} \parens{\delta_{ac} \delta_{ce} + \delta_{ae} - \frac{1}{\dim} \parens{\delta_{ac} \delta_{ce} + \delta_{ae}}}  \\
        &= \frac{1}{\dim^2 - 1} \parens{1 - \frac{1}{\dim}} \sum_{a, b, c, e} \ketbra{a}{e} e^{-i (\lambda(c) - \lambda(a))x} e^{-i (\lambda(e) - \lambda(c))y} \delta_{ae} (1 + \delta_{ac}) \\
        &= \frac{1}{\dim^2 - 1} \parens{1 - \frac{1}{\dim}} \sum_{a, b, c} \ketbra{a}{a} e^{i (\lambda(a) - \lambda(c))(x-y)} (1 + \delta_{ac}) \\
        &= \frac{1 \parens{\dim - 1}}{\dim^2 - 1} \sum_{a,c} \ketbra{a}{a} e^{i (\lambda(a) - \lambda(c))(x - y)} (1 + \delta_{ac}) \\
        &= \frac{1}{\dim + 1} \parens{\sum_{a,c} e^{i (\lambda(a) - \lambda(c))(x-y)} \ketbra{a}{a} + \identity}.
    \end{align}
    Reindexing by $a \mapsto i,j$, $c \mapsto k,l$, and plugging in the definition of $\Delta$ yields the statement of the lemma.
\end{proof}

\begin{lemma} \label{lem:sandwiched_interaction}
    Given two Heisenberg evolved random interactions, $G(x)$ and $G(y)$, we compute their action on an outer-product $\ketbra{i,j}{k,l}$ as
    \begin{equation}
        \int G(x) \ketbra{i,j}{k,l} G(y) ~dG = \frac{1}{\dim + 1} \parens{\ketbra{i,j}{k,l} + \braket{i,j}{k,l} \sum_{m,n} e^{i \Delta(m,n | i,j) (x-y)} \ketbra{m,n}{m,n}}
    \end{equation}
\end{lemma}
\begin{proof}
This proof is structured the same as Lemma \ref{lem:two_heisenberg_interactions} and similarly we will use a single index of the total Hilbert space $\hilb$ and switch to two indices to match the rest of the exposition.
    \begin{align}
        \int G(x) \ketbra{a}{b} G(y) dG &=  \int e^{i H x} U_G D U_G^{\dagger} e^{-i H x} \ketbra{a}{b} e^{i H y} U_G D U_G^\dagger e^{-i H y} ~dG \\
        &= \sum_{c, d, e, f} e^{i (\lambda(c) - \lambda(a))x} e^{i (\lambda(b) - \lambda(f))y} \nonumber \\
        &\quad \times \int \ketbra{c}{c} U_G D(d) \ketbra{d}{d} U_G^\dagger \ketbra{a}{b} U_G D(e) \ketbra{e}{e} U_G^\dagger \ketbra{f}{f} dG \\
        &= \sum_{c, d, e, f}  e^{i (\lambda(c) - \lambda(a))x} e^{i (\lambda(b) - \lambda(f))y} \ketbra{c}{f} \nonumber \\
        &\quad \times \int D(d) D(e) dD \int \bra{c} U_G \ket{d} \bra{b} U_G \ket{e} \bra{d} U_G^\dagger \ket{a} \bra{e} U_G^\dagger \ket{f} dU_G \\
        &=  \sum_{c,d,f} e^{i (\lambda(c) - \lambda(a))x} e^{i (\lambda(b) - \lambda(f))y} \ketbra{c}{f} \nonumber \\ 
        &\quad \times \int \bra{c} U_G \ket{d} \bra{b} U_G \ket{d} \bra{a} \overline{U_G} \ket{d} \bra{f} \overline{U_G} \ket{d} dU_G \\
        &= \frac{1}{\dim^2 - 1} \sum_{c,d,f} e^{i (\lambda(c) - \lambda(a))x} e^{i (\lambda(b) - \lambda(f))y} \ketbra{c}{f} (\delta_{ca} \delta_{bf} + \delta_{cf}\delta_{ab})\parens{1 - \frac{1}{\dim}} \\
        &= \frac{1}{\dim + 1} \sum_{c,f} e^{i (\lambda(c) - \lambda(a))x} e^{i (\lambda(b) - \lambda(f))y} \ketbra{c}{f} (\delta_{ca} \delta_{bf} + \delta_{cf}\delta_{ab}) \\
        &= \frac{1}{\dim + 1} \parens{\ketbra{a}{b} + \delta_{ab} \sum_{c} e^{i(\lambda(c) - \lambda(a))(x-y)} \ketbra{c}{c} }.
    \end{align}
    Now re-indexing by $a \mapsto (i,j)$, $b \mapsto (k,l)$ and $c \mapsto (m,n)$ results in the expression given in the statement of the lemma.
\end{proof}

\secondOrderChannelHaar*
\begin{proof}
To start we would like to note that we will use a single index notation to refer to the joint system-environment eigenbasis during this proof to help shorten the already lengthy expressions. We will convert back to a double index notation to match the statement of the theorem. We start from the expression for the first derivative of the channel $\frac{\partial}{\partial \alpha} \Phi_G(\rho_S)$ given by Eq. \eqref{eq:first_order_alpha_derivative}. To take the second derivative there are six factors involving $\alpha$, so we will end up with six terms. We repeat Eq. \eqref{eq:first_order_alpha_derivative} below, add a derivative, and label each factor containing an $\alpha$ for easier computation
\begin{align}
    \frac{\partial^2}{\partial \alpha^2} \Phi_G(\rho_S) =& \frac{\partial}{\partial \alpha} \parens{\int_{0}^{1} \underset{\substack{\downarrow \\ (A)}}{e^{i s (H+\alpha G)t}} (i t G) \underset{\substack{\downarrow \\ (B)}}{e^{i (1-s) (H+\alpha G)t}} ds ~ \rho \underset{\substack{\downarrow \\ (C)}}{e^{-i(H+\alpha G)t}} } \nonumber \\
    &~ ~+\frac{\partial}{\partial \alpha} \parens{ \underset{\substack{\downarrow \\ (D)} }{e^{i(H+\alpha G)t}} \rho \int_{0}^1 \underset{\substack{\downarrow \\ (E)} }{e^{-i s (H+\alpha G) t} } (- i t G) \underset{\substack{\downarrow \\ (F)}}{e^{-i (1-s) (H+\alpha G)t}} ds }. \label{eq:second_derivative_labels}
\end{align}
Our goal is to get each of these terms in a form in which we can use either Lemma \ref{lem:two_heisenberg_interactions} or \ref{lem:sandwiched_interaction}. 
\begin{align}
    (A) &=i t\int_0^1 \parens{\frac{\partial}{\partial \alpha} e^{i s_1 (H+ \alpha G)t}} G e^{i(1-s_1)(H+\alpha G)t} ds_1 \rho e^{-i (H+\alpha G)t} \bigg|_{\alpha=0} \\
    &= (it)^2 \int_0^1 \parens{\int_0^1 e^{i s_1 s_2 (H+\alpha G)t} s_1 G e^{i s_1 (1-s_2) (H+\alpha G)t} ds_2} G e^{i(1-s_1) (H+\alpha G)t} ds_1 \rho e^{-i(H+\alpha G) t} \bigg|_{\alpha=0} \label{eq:second_order_deriv_intermediate_a}\\
    &= -t^2 \int_0^1 \int_0^1 e^{i s_1 s_2 H t} G e^{-i s_1 s_2 H t} e^{i s_1 H t} G e^{-i s_1 H t} s_1 ds_1 ds_2 e^{i H t} \rho e^{-i H t} \\
    &= -t^2 \int_0^1 \int_0^1 G(s_1 s_2 t) G(s_1 t) s_1 ds_1 ds_2 \rho(t). \label{eq:second_deriv_alpha_first_term}
\end{align}

\begin{align}
    (B) &= it \int_0^1 e^{i s_1 (H + \alpha G)t} G \frac{\partial}{\partial \alpha}\parens{e^{i(1-s_1)(H + \alpha G)t}} ds_1 \rho e^{-i(H + \alpha G) t} \bigg|_{\alpha = 0} \\
    &= (it)^2 \int_0^1 e^{i s_1 (H + \alpha G)t} G \parens{\int_0^1 e^{i(1-s_1)s_2 (H + \alpha G)t} (1-s_1) G e^{i(1 - s_1)(1 - s_2)(H + \alpha G)t} ds_2} ds_1 ~ \rho e^{-i ( H + \alpha G)t} \bigg|_{\alpha = 0} \\
    &= -t^2 \int_0^1 \int_0^1 e^{i s_1 H t} G e^{i(1-s_1)s_2 H t} G e^{i(1-s_1)(1-s_2) H t} (1-s_1) ds_1 ds_2 ~ \rho e^{-i H t}\\ 
    &= -t^2 \int_0^1 \int_0^1 e^{i s_1 H t} G e^{-i s_1 H t} e^{i(s_1 + s_2 - s_1 s_2) H t} G e^{-i (s_1 + s_2 - s_1 s_2) H t} (1-s_1) ds_1 ds_2 ~ \rho(t) \\
    &= -t^2 \int_0^1 \int_0^1 G(s_1 t) G((s_1 + s_2 - s_1 s_2)t) (1-s_1) ds_1 ds_2 ~ \rho(t)
\end{align}

\begin{align}
    (C) &= it \int_0^1 e^{i s (H + \alpha G)t} G e^{i(1-s) (H + \alpha G) t} ds ~\rho ~ \frac{\partial}{\partial \alpha} \parens{ e^{-i (H + \alpha G) t} } \bigg|_{\alpha = 0} \\
    &= (i t) (-it) \int_0^1 e^{i s (H + \alpha G)t} G e^{i (1 - s) (H + \alpha G)t} ds ~ \rho ~ \parens{ \int_0^1 e^{-i s (H + \alpha G)t} G e^{-i (1- s) ( H + \alpha G)t } ds}\bigg|_{\alpha = 0} \\
    &= + t^2 \parens{\int_0^1 e^{i s H t} G e^{-i s H t} ds} e^{i H t} \rho e^{-i H t} \parens{\int_0^1 e^{i (1-s) H t} G e^{-i (1-s) H t} ds} \\
    &= + t^2 \int_0^1 G(st) ds ~ \rho(t) \int_0^1 G((1-s)t) ds
\end{align}

\begin{align}
    (D) &= (-it) \frac{\partial}{\partial \alpha} \parens{e^{i(H + \alpha G)t}} \rho \int_0^1 e^{-i s (H + \alpha G)t} G e^{-i (1-s)(H + \alpha G)t} ds \bigg|_{\alpha = 0} \\
    &= t^2 \parens{\int_0^1 e^{i s (H+ \alpha G)t} G e^{i (1-s) (H + \alpha G)t}ds} \rho \int_0^1 e^{-i s (H + \alpha G)t} G e^{-i (1-s)(H + \alpha G)t} ds \bigg|_{\alpha = 0} \\
    &=  t^2 \int_0^1 e^{i s H t} G e^{-i s H t} ds ~\rho(t) \int_0^1 e^{i (1-s) H t} G e^{-i (1-s) H t} ds \\
    &= t^2 \int_0^1 G(st) ds ~ \rho(t) ~ \int_0^1 G((1-s)t) ds
\end{align}

\begin{align}
    (E) &= (-it) e^{i (H+ \alpha G) t} ~ \rho ~\int_0^1 \frac{\partial}{\partial \alpha} \parens{e^{-i s_1 (H + \alpha G)t}} G e^{-i (1-s_1)(H + \alpha G)t} ds_1 \bigg|_{\alpha = 0} \\
    &= - t^2 e^{i(H + \alpha G)t} ~ \rho ~\int_0^1 \parens{\int_0^1 e^{-i s_1 s_2 (H + \alpha G) t} (s_1 G) e^{-i s_1 (1-s_2) (H + \alpha G)t} ds_2} G e^{-i(1-s_1)(H + \alpha G)t} ds_1 \bigg|_{\alpha = 0} \\
    &= -t^2 e^{i H t} \rho e^{-i H t} \int_0^1 \int_0^1 e^{i (1 - s_1 s_2) H t} G e^{-i (s_1 - s_1 s_2)H t} G e^{-i (1-s_1)H t} s_1 ds_1 ds_2 \\
    &= -t^2 \rho(t) \int_0^1 \int_0^1 G((1- s_1 s_2) t) G((1-s_1)t) s_1 ds_1 ds_2
\end{align}

\begin{align}
    (F) &= (-it) e^{i(H + \alpha G) t} \rho \int_0^1 e^{-i s_1 ( H + \alpha G) t} G \frac{\partial}{\partial \alpha} \parens{ e^{-i (1-s_1) ( H +\alpha G)t}} ds_1 \bigg|_{\alpha = 0} \\
    &= (-it)^2 e^{i (H + \alpha G)t} \rho \int_0^1 e^{-i s_1 (H + \alpha G)t} G \parens{\int_0^1 e^{-i(1-s_1) s_2 (H + \alpha G)t} (1-s_1) G e^{-i(1-s_1) (1-s_2) (H + \alpha G) t} ds_2} ds_1 \bigg|_{\alpha = 0} \\
    &= -t^2 e^{-i H t} \rho e^{-i H t} \int_0^1 \int_0^1 e^{i (1- s_1) H t} G e^{-i (1-s_1) H t} e^{i (1-s_1)(1-s_2) H t} G e^{-i(1-s_1)(1-s_2) H t} (1-s_1) ds_1 ds_2 \\
    &= -t^2 \rho(t) \int_0^1 \int_0^1 G((1-s_1)t) G((1-s_1)(1 - s_2) t) (1-s_1)ds_1 ds_2
\end{align}

Now our goal is to compute the effects of averaging over the interaction $G$ on the above terms, starting with $(A)$. As this involves a lot of index manipulations, similarly to the proofs of Lemmas \ref{lem:two_heisenberg_interactions} and \ref{lem:sandwiched_interaction} we will use a single index for the total system-environment Hilbert space and switch back to a double index to state the results. We will make heavy use of Lemma \ref{lem:two_heisenberg_interactions}.
\begin{align}
    \int (A) dG &= -t^2 \int_0^1 \int_0^1 \int G(s_1 s_2 t) G(s_1 t) dG s_1 ds_1 ds_2 \rho(t) \\
    &= \frac{-t^2 }{\dim + 1} \int_0^1 \int_0^1 \parens{\sum_{i,j} e^{i (\lambda(i) - \lambda(j)) (s_1 s_2 t - s_1 t)} \ketbra{i}{i} + \identity} s_1 ds_1 ds_2 \rho(t) \\
    &= \frac{- t^2 }{\dim + 1} \parens{\sum_{i} \sum_{j : \lambda(i) \neq \lambda(j)} \int_0^1 \int_0^1 e^{i(\lambda(i) - \lambda(j))t (s_1 s_2 - s_1)} s_1 ds_1 ds_2 \ketbra{i}{i} + \sum_{i} \sum_{j : \lambda(i) = \lambda(j)}\frac{1}{2} \ketbra{i}{i} + \frac{1}{2} \identity} \rho(t) \\
    &= \frac{- t^2 }{\dim + 1} \parens{\sum_i \sum_{j : \lambda(i) \neq \lambda(j)} \frac{1 - i (\lambda(i) - \lambda(j))t - e^{-i (\lambda(i) - \lambda(j))t}}{t^2 (\lambda(i) - \lambda(j))^2} \ketbra{i}{i} + \frac{1}{2} \sum_{i} (\eta(i) + 1) \ketbra{i}{i} } \rho(t) \\
    &= \frac{- 1}{\dim + 1}\parens{\sum_{i} \sum_{j: \Delta_{ij} \neq 0} \frac{1 - i \Delta_{ij}t - e^{-i \Delta_{ij} t}}{\Delta_{ij}^2} \ketbra{i}{i} + \frac{t^2}{2} \sum_{i} (\eta(i) + 1)\ketbra{i}{i} } \rho(t)
\end{align}

We can similarly compute the averaged $(B)$ term:
\begin{align}
    \int (B) dG &= -t^2 \int_0^1 \int_0^1 \int G(s_1 t) G((s_1 + s_2 - s_1 s_2) t) dG (1-s_1) ds_1 ds_2 ~ \rho(t) \\
    &= \frac{- t^2 }{\dim + 1} \int_0^1 \int_0^1 \parens{\sum_{i,j} e^{i (\lambda(i) - \lambda(j))(s_1 s_2 - s_2) t} \ketbra{i}{i} + \identity} (1 -s_1) ds_1 ds_2 \rho \\
    &= \frac{- t^2 }{\dim + 1} \parens{\sum_{i} \sum_{j : \lambda(i) \neq \lambda(j)} \int_0^1 \int_0^1 e^{i(\lambda(i) - \lambda(j))t (s_1 s_2 - s_2)} (1 - s_1) ds_1 ds_2 \ketbra{i}{i} + \sum_{i} \sum_{j : \lambda(i) = \lambda(j)}\frac{1}{2} \ketbra{i}{i} + \frac{1}{2} \identity} \rho(t) \\
    &= \frac{- t^2 }{\dim + 1} \parens{\sum_i \sum_{j : \lambda(i) \neq \lambda(j)} \frac{1 - i (\lambda(i) - \lambda(j))t - e^{-i (\lambda(i) - \lambda(j))t}}{t^2 (\lambda(i) - \lambda(j))^2} \ketbra{i}{i} + \frac{1}{2} \sum_{i} (\eta(i) + 1) \ketbra{i}{i} } \rho(t) \\
    &= \frac{-1}{\dim + 1}\parens{\sum_{i} \sum_{j: \Delta_{ij} \neq 0} \frac{1 - i \Delta_{ij}t - e^{-i \Delta_{ij} t}}{\Delta_{ij}^2} \ketbra{i}{i} + \frac{t^2}{2} \sum_{i} (\eta(i) + 1)\ketbra{i}{i} } \rho(t),
\end{align}
which we note is identical to $\int (A) dG$. As terms $(C)$ and $(D)$ involve a different method of computation we skip them for now and compute $(E)$ and $(F)$. 
\begin{align}
    \int (E) dG &= -t^2 \rho(t) \int_0^1 \int_0^1 \int G((1- s_1 s_2) t) G((1-s_1)t) dG s_1 ds_1 ds_2 \\
    &= \frac{- t^2}{\dim + 1} \rho(t) \int_0^1 \int_0^1 \parens{\sum_{i,j} e^{i(\lambda(i) - \lambda(j)) t (s_1 - s_1 s_2)} \ketbra{i}{i} + \identity } s_1 ds_1 ds_2 \\
    &= \frac{- t^2}{\dim + 1} \rho(t) \parens{\sum_i \sum_{j : \lambda(i) \neq \lambda(j)} \frac{1 + i (\lambda(i) - \lambda(j))t - e^{i(\lambda(i) - \lambda(j))t}}{t^2 (\lambda(i) - \lambda(j))^2}\ketbra{i}{i} + \frac{1}{2} \sum_{i} (\eta(i) + 1 )\ketbra{i}{i}} \\
    &= \frac{- 1}{\dim + 1} \rho(t) \parens{\sum_i \sum_{j: (\Delta_{ij} \neq 0)} \frac{1 + i \Delta_{ij}t - e^{i\Delta_{ij}t}}{\Delta_{ij}^2} \ketbra{i}{i} + \frac{t^2}{2}\sum_i (\eta(i) + 1) \ketbra{i}{i}}.
\end{align}
Computing $(F)$ yields
\begin{align}
    \int (F) dG &= -t^2 \rho(t) \int_0^1 \int_0^1 \int G((1-s_1)t) G((1-s_1)(1 - s_2) t) dG (1-s_1)ds_1 ds_2 \\
    &= \frac{- t^2 \sigma ^2}{\dim + 1} \rho(t) \int_0^1 \int_0^1 \parens{\sum_{i,j} e^{i(\lambda(i) - \lambda(j))t (s_2 - s_1 s_2)}\ketbra{i}{i} + \identity} (1-s_1) ds_1 ds_2 \\
    &= \frac{- t^2 }{\dim + 1} \rho(t) \parens{\sum_{i} \sum_{j : \lambda(i) \neq \lambda(j)} \frac{1 + i (\lambda(i) - \lambda(j))t - e^{i (\lambda(i) - \lambda(j))t}}{t^2 (\lambda(i) - \lambda(j))^2} \ketbra{i}{i} +\frac{1}{2} \sum_{i} (\eta(i) + 1) \ketbra{i}{i}} \\
    &= \frac{- 1}{\dim + 1} \rho(t) \parens{\sum_i \sum_{j: (\Delta_{ij} \neq 0)} \frac{1 + i \Delta_{ij}t - e^{i\Delta_{ij}t}}{\Delta_{ij}^2} \ketbra{i}{i} + \frac{t^2}{2}\sum_i (\eta(i) + 1) \ketbra{i}{i}}
\end{align}
 which is identical to $\int (E) dG$.

 The last two terms $(C) = (D)$ are computed as follows:
 \begin{align}
     \int (C) dG &= t^2 \int_0^1 \int_0^1 \int G(s_1 t) \rho(t) G((1-s_2)t) ~dG ~ ds_1 ds_2 \\
     &= t^2 \sum_{i,j} \rho_{ij} e^{i(\lambda(i) - \lambda(j))t} \int_0^1 \int_0^1 \int G(s_1 t) \ketbra{i}{j} G((1-s_2)t) ~ dG ~ ds_1 ds_2 \\
     &= \frac{ t^2}{\dim + 1} \sum_{i,j} \rho_{ij} e^{i(\lambda(i) - \lambda(j))t} \parens{ \ketbra{i}{j} + \delta_{ij} \sum_{a} \int_0^1 \int_0^1 e^{i(\lambda(a) - \lambda(i))(s_1 + s_2 - 1)t} ds_1 ds_2 \ketbra{a}{a}} \\
     &= \frac{ t^2}{\dim + 1} \sum_{i,j} \rho_{ij} e^{i \Delta_{ij} t} \parens{\ketbra{i}{j} + \delta_{ij} \sum_{a : \Delta_{ai} \neq 0} \frac{2( 1- \cos (\Delta_{ai} t))}{\Delta_{ai}^2 t^2} \ketbra{a}{a} + \delta_{ij} \sum_{a : \Delta_{ai} = 0} \ketbra{a}{a}}
 \end{align}

 We can now combine each of these terms to offer the full picture of the output of the channel to second order. We make two modifications to the results from each sum: first, we will switch to double index notation to make for easier use in other areas, and secondly we let $\rho = \ketbra{i,j}{k,l}$. We note that the first term in the following equation is provided by $(A) + (B)$, the second through $(E) + (F)$, and the last two through $(C) + (D)$. 
 \begin{align}
     &\int \frac{\partial^2}{\partial \alpha^2} \Phi_G(\ketbra{i,j}{k,l})\bigg|_{\alpha = 0} dG \\
     &= -\frac{2  e^{i \Delta(i,j|k,l) t}}{\dim + 1} \bigg(\sum_{(a,b): \Delta(i,j|a,b) \neq 0} \frac{1 - i \Delta(i,j|a,b)t - e^{-i \Delta(i,j|a,b) t}}{\Delta(i,j|a,b)^2} \nonumber \\
     &~+ \sum_{(a,b): \Delta(k,l|a,b) \neq 0} \frac{1 + i \Delta(k,l|a,b) t - e^{i \Delta(k,l|a,b) t}}{\Delta(k,l|a,b)^2} + \frac{t^2}{2}(\eta(i,j) + \eta(k,l)) \bigg) \ketbra{i,j}{k,l} \nonumber \\
    &~ +\delta_{i,k} \delta_{j,l} \frac{2 e^{i \Delta(i,j|k,l)t}}{\dim+1} \parens{ \sum_{(a,b): \Delta(i,j|a,b) \neq 0 } \frac{2(1- \cos (\Delta(i,j|a,b)t))}{\Delta(i,j|a,b)^2} \ketbra{a,b}{a,b} + t^2 \sum_{(a,b) : \Delta(i,j|a,b) = 0} \ketbra{a,b}{a,b}} \label{eq:second_order_output}
 \end{align}
The last step we need is to use the half angle formula to change the cosine to a sine
\begin{align}
    \frac{2(1 - \cos(\Delta(i,j| a,b)t)}{\Delta(i,j|a,b)^2} &= \frac{2\left( 1 - \left(1 - 2 \sin^2\left(\frac{\Delta(i,j|a,b)t}{2} \right) \right) \right)}{\Delta(i,j|a,b)^2} \label{eq:trig_start} \\
    &= t^2 \frac{\sin^2 \left(\frac{\Delta(i,j|a,b) t}{2} \right)}{\frac{\Delta(i,j|a,b)^2 t^2}{4}} \\
    &= t^2 \sinc^2 \left(\frac{\Delta(i,j|a,b) t}{2 } \right), \label{eq:trig_end}
\end{align}
which yields the statement.

We can compute these by plugging in to Eq. \eqref{eq:el_gigante} again, which yields
\begin{align}s
&\int \bra{i', j'} \mathcal{T} \left( \ketbra{i, j}{i, j} \right) \ket{i', j'} ~dG = \begin{cases}        
\widetilde{\alpha}^2 \sinc^2(\Delta(i,j | i', j') t /2) & (i, j) \neq (i', j') \\
            - \widetilde{\alpha}^2 \sum_{(a,b) \neq (i, j)} \sinc^2(\Delta(a,b|i,j) t / 2) & (i,j) = (i', j')
        \end{cases}. \label{eq:system_environment_transitions}
    \end{align}
    The $(i, j) \neq (i', j')$ case should be apparent, the first term with the coherence factors $\chi$ are zero and the second term is what remains. The $(i,j) = (i', j')$ case can be seen as follows. For the first term we have
    \begin{align}
        - \frac{\alpha^2 e^{i \Delta(i,j| i,j) t}}{\dim + 1}\left(\chi(i,j) + \chi(i,j)^* + \frac{t^2}{2}(\eta(i,j) + \eta(i,j) \right) \ketbra{i,j}{i,j}.
    \end{align}
    We first compute the sum $\chi(i,j) + \chi(i,j)^*$ as
    \begin{align}
        \chi(i,j) + \chi(i,j)^* &= \sum_{a,b: \Delta(i,j,|a,b) \neq 0} \frac{1 - i \Delta(i,j|a,b)t - e^{-i \Delta(i,j|a,b) t}}{\Delta(i,j|a,b)^2} \nonumber\\
&\quad+ \sum_{a,b: \Delta(i,j,|a,b) \neq 0} \frac{1 + i \Delta(i,j|a,b)t - e^{+i \Delta(i,j|a,b) t}}{\Delta(i,j|a,b)^2} \\
    &= \sum_{a,b: \Delta(i,j| a,b) \neq 0} \frac{2 - e^{-i \Delta(i,j| a,b) t} - e^{+i \Delta(i,j| a,b) t}}{\Delta(i,j|a,b)^2} \\
    &= \sum_{a,b: \Delta(i,j| a,b) \neq 0} t^2 \sinc^2 \left( \frac{\Delta(i,j| a,b) t}{2} \right),
    \end{align}
    where the last step follows from a trigonometric identity (see Eqs. \eqref{eq:trig_start} - \eqref{eq:trig_end} in Appendix \ref{sec:haar_integral_appendix} for details). Since $\sinc(0) = 1$ the $\eta(i,j)$ term can be expressed as $\eta(i,j) = \sum_{a,b : \Delta(i,j|a,b) = 0} \sinc^2 \left( \frac{\Delta(i,j| a,b) t}{2} \right)$. Plugging this into Eq. \eqref{eq:el_gigante} gives
    \begin{align}
        &\int \bra{i,j} \mathcal{T} (\ketbra{i,j}{i,j}) \ket{i,j} dG \\
        &= \bra{i,j} \left(-\frac{\alpha^2 t^2}{\dim + 1} \sum_{a,b} \sinc^2 \left( \frac{\Delta(i,j| a,b) t}{2} \right) \ketbra{i,j}{i,j} + \sum_{a,b} \sinc^2\left( \frac{\Delta(i,j | a,b)t}{2} \right) \ketbra{a,b}{a,b} \right) \ket{i,j} \\
        &= -\frac{\alpha^2 t^2}{\dim + 1} \sum_{(a,b) \neq (i,j)} \sinc^2 \left( \frac{\Delta(i,j| a,b) t}{2} \right).
    \end{align}
    As a by-product of this computation we have just shown that $\trace{\mathcal{T}(\rho)} = 0$ and that our mapping is trace preserving to $\bigo{\alpha^2}$.
\end{proof}

\subsection{Weak-Coupling Remainder Bound} \label{sec:weak_coupling_remainder_bound}

\begin{proof}[Proof of Theorem~\ref{thm:remainder_bound}]
First we note that although $R_{\Phi}(\rho) = \frac{\alpha^3}{6} \partial_{\alpha}^3 \Phi(\rho)\big|_{\alpha = \alpha_{\star}}$ for a specific $ \alpha_{\star} > 0$ our proof will hold for any $\alpha_{\star}$. To compute the trace norm we will use the triangle inequality, unitary invariance of the Sch\"{a}tten norms, and submultiplicativity. To start,
\begin{align}
    \|\partial_\alpha^3 \Phi(\rho) \|_1 &= \left\| \frac{\partial^3}{\partial \alpha^3} {\rm Tr}_{H_E} \int e^{i(H+\alpha G)t} \rho_S \otimes \rho_E e^{-i(H+\alpha G)t} dG \bigg|_{\alpha = \alpha_{\star}} \right\|_1 \nonumber\\
    &= \left\| \frac{\partial^3}{\partial \alpha^3} {\rm Tr}_{H_E} \int e^{i(H+\alpha G)t} \rho_S \otimes \rho_E e^{-i(H+\alpha G)t} dG \bigg|_{\alpha = \alpha_{\star}} \right\|_1 \nonumber\\
    &\le    \int \left\|{\rm Tr}_{H_E}\frac{\partial^3}{\partial \alpha^3}\left( e^{i(H+\alpha G)t} \rho_S \otimes \rho_E e^{-i(H+\alpha G)t}\right) \bigg|_{\alpha = \alpha_{\star}} \right\|_1 dG.\label{eq:3derivBd}
\end{align}
To proceed, Proposition 1 of \cite{rastegin2012relations} allows us to eliminate the partial trace using the relation
$\norm{\partrace{\hilb_E}{X}}_{1} \le \norm{X}_{\dim_E} \le \norm{X}_1$. Further we use the decomposition of the second derivatives from the proof of Lemma \ref{lem:big_one}, specifically  Eq. \eqref{eq:second_derivative_labels} for the definition of each term, as $\partial_{\alpha}^2 \Phi_G = (A) + (B) + (C) + (D) +(E) + (F)$. This gives the following 
\begin{align}
    \norm{R_{\Phi}}_1 \le \frac{\alpha^3}{6} &\int \norm{\partial_{\alpha}((A) + (B) + (C) + (D) +(E) + (F)) \big|_{\alpha = \alpha_{\star}} }_1 dG \\
    &\le \frac{\alpha^3}{6} \int \norm{\partial_{\alpha}(A)\big|_{\alpha = \alpha_{\star}} }_1 + \norm{\partial_{\alpha}(B) \big|_{\alpha = \alpha_{\star}} }_1 + \ldots + \norm{\partial_{\alpha}(F) \big|_{\alpha = \alpha_{\star}} }_1 dG
\end{align}
We will demonstrate how this can be computed for the first term $\partial_{\alpha}(A)$. Using Eq. \eqref{eq:second_deriv_alpha_first_term}
\begin{align}
\partial_{\alpha} (A) &= -t^2 \partial_{\alpha} \int_0^1 \int_0^1 e^{i s_1 s_2 (H+\alpha G)t} G e^{i s_1 (1-s_2) (H+\alpha G)t} G e^{i(1-s_1) (H+\alpha G)t} \rho e^{-i(H+\alpha G) t}   s_1 ~ds_1 ds_2.
\end{align}
Due to the multiplication rule the resulting derivative will have 4 terms that each need one application of Duhamel's formula. We will show one of these terms, starting with the leftmost one
\begin{align}
    -t^2  &\int_0^1 \int_0^1 \partial_{\alpha} \left( e^{i s_1 s_2 (H+\alpha G)t} \right) G e^{i s_1 (1-s_2) (H+\alpha G)t} G e^{i(1-s_1) (H+\alpha G)t} \rho e^{-i(H+\alpha G) t}   s_1 ~ds_1 ds_2 \\
    = - i t^3  &\int_0^1 \int_0^1 \int_0^1 \left( e^{i s_1 s_2 s_3 (H+\alpha G)t} G e^{i s_1 s_2 (1 - s_3) (H+\alpha G)t} \right)\nonumber\\
    &\qquad\times G e^{i s_1 (1-s_2) (H+\alpha G)t} G e^{i(1-s_1) (H+\alpha G)t} \rho e^{-i(H+\alpha G) t}   s_1^2 s_2 ~ds_1 ds_2 ds_3.
\end{align}
Our goal is to compute the 1-norm of the above expression at $\alpha = \alpha_{\star}$. We can do so using the triangle inequality to move the norm of the matrices in the integrand. To proceed we will set $\alpha = \alpha_{\star}$ and reduce the norm of the matrices in the integrand using submultiplicativity and unitary invariance of the Sch\"{a}tten 1-norm as 
\begin{align}
    &\norm{\left( e^{i s_1 s_2 s_3 (H+\alpha_{\star} G)t} G e^{i s_1 s_2 (1 - s_3) (H+\alpha_{\star} G)t} \right) G e^{i s_1 (1-s_2) (H+\alpha_{\star} G)t} G e^{i(1-s_1) (H+\alpha_{\star} G)t} \rho e^{-i(H+\alpha_{\star} G) t}}_1 \\
    \le &\norm{e^{i s_1 s_2 s_3 (H+\alpha_{\star} G)t} G}_1 \norm{e^{i s_1 s_2 (1 - s_3) (H+\alpha_{\star} G)t} G}_1 \norm{e^{i s_1 (1-s_2) (H+\alpha_{\star} G)t} G}_1 \norm{e^{i(1-s_1) (H+\alpha_{\star} G)t} \rho e^{-i(H+\alpha_{\star} G) t}}_1 \\
    \le &\norm{G}_1^3 \norm{\rho}_1 = \norm{G}_1^3 .
\end{align}
Similar computations can be carried out for the other three terms for the derivative acting on $(A)$ which yields the inequality
\begin{align}
    \frac{\alpha^3}{6} \int \norm{\partial_{\alpha} (A)}_1 dG &\le \frac{\alpha^3 t^3}{6} \int \int_0^1 \int_0^1 \int_0^1 \norm{G}_1^3 (s_1^2 s_2 + s_1^2 (1 - s_2) + s_1(1-s_1) + s_1) ~ds_1 ds_2 ds_3 dG \\
    &\le \frac{4}{6} (\alpha t)^3 \int \norm{G}_1^3 dG . \label{eq:remainder_bound_on_A}
\end{align}

Now that we have computed the norm of the derivative acting on term $(A)$ we only have terms $(B)$ through $(F)$ to compute. These can all be checked to satisfy the same bound on $(A)$ from Eq. \eqref{eq:remainder_bound_on_A}, and as there are six terms in total we have the inequality
\begin{align}
    \norm{R_{\Phi}(\rho)}_1 &\le 4 (\alpha t)^3 \int \norm{G}_1^3 dG,
\end{align}
which we note holds for all inputs $\rho$. Therefore our last problem is to compute the expected norm of our interaction to the third power. We will decompose $G = U_{\text{Haar}} D U_{\text{Haar}}^\dagger $ to get
\begin{align}
    \int \norm{G}_1^3 dG &= \int \norm{D}_1^3 ~ dU_{\text{Haar}} dD = \sum_{i = 1}^{2 \dim_S} \int \abs{d_i}^3 dd_i = 2 \dim_S \mathbb{E}(\abs{y}^3),
\end{align}

where $y$ is a normal Gaussian random variable. This straightforwardly evaluates to $\mathbb{E}( |y|^3) = 2\sqrt{ \frac{2}{\pi}}$, yielding the final inequality
\begin{align}
    \norm{R_{\Phi}(\rho)}_1 \le 16 \sqrt{\frac{2}{\pi}} \dim_S (\alpha t)^3 ,
\end{align}
thus completing the proof.
\end{proof}

\end{document}